\documentclass[dvipsnames,11pt]{article}
\usepackage[utf8]{inputenc}
\usepackage{verbatim,amsmath,amssymb,amsfonts,amsthm,amscd,graphicx,psfrag,epsfig,mathrsfs}
\usepackage{mathrsfs}
\usepackage{tikz-cd}
\usepackage[title]{appendix}
\usepackage[dvipsnames]{xcolor}

\usepackage{mathtools}

\usepackage[all]{xy}
\usepackage[alphabetic]{amsrefs}
\usepackage{geometry}
\usepackage{bbm}
\usepackage{bm}
\usepackage{tikz-cd,hyperref,bm} 
\usepackage[]{subcaption}
\usetikzlibrary{positioning,calc,arrows,decorations.pathreplacing,decorations.markings,patterns}
\usepackage{ifthen}
\usepackage{hhline}
\usepackage[scr=boondoxupr]{mathalfa}
\tikzcdset{scale cd/.style={every label/.append style={scale=#1},
    cells={nodes={scale=#1}}}}

\makeatletter
\newcommand{\doublewidetilde}[1]{{%
  \mathpalette\double@widetilde{#1}%
}}
\newcommand{\double@widetilde}[2]{%
  \sbox\z@{$\m@th#1\widetilde{#2}$}%
  \ht\z@=.9\ht\z@
  \widetilde{\box\z@}%
}
\makeatother

\usepackage{blkarray,stmaryrd}

\usepackage{verbatim}

\newtheorem{theorem}{Theorem}[section]

\newtheorem{theorem-definition}[theorem]{Theorem-Definition}
\newtheorem{theorem-construction}[theorem]{Theorem-Construction}
\newtheorem{lemma-definition}[theorem]{Lemma--Definition}
\newtheorem{lemma-construction}[theorem]{Lemma--Construction}
\newtheorem{lemma}[theorem]{Lemma}

\newtheorem{proposition}[theorem]{Proposition}
\newtheorem{corollary}[theorem]{Corollary}
\newtheorem{conjecture}[theorem]{Conjecture}
\theoremstyle{definition}
\newtheorem{definition}[theorem]{Definition}

\newtheorem{remark}[theorem]{Remark}
\newtheorem{example}[theorem]{Example}

\newcommand{\old}[1]{}

\newcommand{\wt}{\operatorname{wt}}

\newcommand{\cO}{\mathcal{O}}

\newcommand{\Z}{{\mathbb Z}}

\newcommand{\C}{{\mathbb C}}

\newcommand{\T}{{\mathbb T}}

\renewcommand{\P}{{\mathbb P}}

\newcommand{\bP}{{\mathbb P}}

\def\B{\mathsf B}
\def\W{\mathsf W}
\def\w{\mathsf w}
\def\b{\mathsf b}
\def\vert{\mathsf v}
\def\e{\mathsf e}
\def\f{\mathsf f}
\def\K{\mathsf K}
\def\E{\mathsf E}
\def\F{\mathsf F}
\def\graph{\Gamma}
\def\zz{\bm \alpha}
\def\id{\mathrm{id}}

\def\Spec{\operatorname{Spec}}
\def\mon{\operatorname{mon}}
\def\Ext{\operatorname{Ext}}
\def\tot{\operatorname{Tot}}

\def\cale{\mathcal E}
\def\calf{\mathcal F}
\def\dabel{\bm d}
\def\PD{\mathrm{PD}}

\makeatletter
\newcommand{\Extp}{\@ifnextchar^\@extp{\@extp^{}}}
\def\@extp^#1{\mathop{\bigwedge\nolimits^{ #1}}}
\makeatother
\newcommand\restr[2]{{
  \left.\kern-\nulldelimiterspace 
  #1 
  \vphantom{\big|} 
  \right|_{#2} 
  }}
\definecolor{calpolypomonagreen}{rgb}{0, 0.6, 0.2}

\newcommand{\bmu}{\bm \mu}

\newcommand{\be}{\begin{equation}}
\newcommand{\ee}{\end{equation}}
\newcommand{\bt}{\begin{theorem}}

\newcommand{\et}{\end{theorem}}
\newcommand{\bd}{\begin{definition}}
\newcommand{\ed}{\end{definition}}
\newcommand{\bp}{\begin{proposition}}
\newcommand{\ep}{\end{proposition}}

\newcommand{\bl}{\begin{lemma}}
\newcommand{\el}{\end{lemma}}
\newcommand{\bc}{\begin{corollary}}
\newcommand{\ec}{\end{corollary}}
\newcommand{\bcon}{\begin{conjecture}}
\newcommand{\econ}{\end{conjecture}}
\newcommand{\la}{\label}
\newcommand{\oH}{{H}}

\renewcommand{\L}{{\mathcal L}}

\def\Ext{\operatorname{Ext}}

\DeclareMathOperator{\coker}{coker}
\DeclareMathOperator{\Hom}{Hom}
\DeclareMathOperator{\tr}{tr}
\DeclareMathOperator{\im}{im}

\tikzset{mid arrow/.style={postaction={decorate,decoration={
				markings,
				mark=at position .5 with {\arrow{latex}}
	}}},
	mid rarrow/.style={postaction={decorate,decoration={
				markings,
				mark=at position .5 with {\arrow{latex reversed}}
	}}},
}

\tikzset{qvert/.style={draw,black,circle,fill=gray,minimum size=5pt,inner sep=0pt}  } 
\tikzset{bvert/.style={draw,circle,fill=black,minimum size=5pt,inner sep=0pt}  }  
\tikzset{wvert/.style={draw,circle,fill=white,minimum size=5pt,inner sep=0pt}  } 
\tikzset{fvert/.style={text=blue}  } 
\tikzset{sqvert/.style={draw,black,rectangle,fill=black,minimum size=5pt,inner sep=0pt}  } 
\tikzset{lvert/.style={draw,circle,fill=black,minimum size=4pt,inner sep=0pt}  }  
\usetikzlibrary{arrows}
\tikzcdset{arrow style=tikz, diagrams={>={Stealth[round,length=4pt,width=4.95pt,inset=2.75pt]}}}

\tikzset{>={Latex}}
\newcommand{\Addresses}{{
  \bigskip
  \footnotesize

  \textsc{Massachusetts Institute of Technology, Department of Mathematics, 2-155 Simons Building, 77 Massachusetts Avenue, Cambridge, MA 02139}\par\nopagebreak
  \textit{E-mail address}: \texttt{tegeorge at mit.edu}

 \medskip

  \textsc{University of Washington, Department of Mathematics, C-138 Padelford, Box 354350
     Seattle, WA 98195  }\par\nopagebreak
  \textit{E-mail address}: \texttt{ginchios at uw.edu}

}}

\begin{document}

\title{Dimers and Beauville integrable systems}
\author{Terrence George \and Giovanni Inchiostro}

 \date{\today}

 \maketitle
 
 \begin{abstract}
Associated to a convex integral polygon $N$ in the plane are two integrable systems: the cluster integrable system of Goncharov and Kenyon, constructed from the dimer model on bipartite torus graphs, and the Beauville integrable system associated with the toric surface of $N$. These two systems are related by a birational map called the spectral transform. In this paper we study the case when $N$ is the standard triangle of side length $d$, equivalently when the toric surface is $\P^2$, and prove that the spectral transform is a birational isomorphism of integrable systems. Since the Hamiltonians are identified by construction, the essential content is that the spectral transform intertwines the two Poisson structures. In particular, this shows that Beauville integrable systems admit cluster algebra structures.
 \end{abstract}
\tableofcontents

\section{Introduction}
Let $N$ be a convex lattice polygon in the plane. There are two algebraically completely integrable systems that one can associate to $N$, constructed in rather different settings. In this paper we prove their equivalence in the basic case when $N$ is the standard triangle.

On one side is the \emph{cluster integrable system} of Goncharov--Kenyon~\cite{GK13}, built from the dimer model on bipartite graphs on a torus and the associated cluster Poisson variety. Cluster integrable systems appear in many settings, including Toda-type systems and Poisson-Lie/double Bruhat cell constructions~\cite{GSV,GSVToda,Williams,FockMarshakov,ILP}, discrete geometric integrable systems~\cite{GSTV2,GSTV,GlickPylyavskyy,Izo,AGR}, and ultradiscrete or tropical systems of box-ball type~\cite{IKT,LP}. 

On the other side is the \emph{Beauville integrable system} in algebraic geometry on the projective toric surface associated to the same convex integral polygon~$N$. Such systems belong to a broader circle of ideas involving moduli spaces of sheaves on symplectic surfaces~\cite{Mukai,Tyurin,Beauville} and their Poisson-surface analogues~\cite{Bottacin,Bott1}.

We now briefly introduce the two integrable systems and give a precise statement of our main result.

\subsection{Cluster integrable systems}

Starting from $N$, one obtains a family of bipartite graphs on a torus associated with $N$, whose members are related by local transformations known as spider moves, or equivalently cluster mutations. For each such graph, the space of edge weights modulo gauge transformations is an algebraic torus, and these tori glue under mutation to form a cluster variety $\mathcal X$ in the sense of Fock--Goncharov~\cite{FG}, associated with the cluster algebra of Fomin--Zelevinsky~\cite{FZ}. The variety $\mathcal X$ carries a canonical cluster Poisson structure~\cite{GSV}. Goncharov--Kenyon~\cite{GK13} showed that its symplectic leaves are algebraically completely integrable systems, called \emph{cluster integrable systems}, whose Hamiltonians are partition functions for dimer covers/perfect matchings in each homology class on the torus.

These Hamiltonians admit a natural interpretation in terms of spectral curves. To a weighted bipartite torus graph one associates a periodic finite-difference operator $\K(z,w)$, called the \emph{Kasteleyn matrix}. Its determinant
\[
\mathsf P(z,w):=\det \K(z,w)
\]
is a Laurent polynomial whose Newton polygon is $N$ called the \emph{characteristic polynomial} of the dimer model. The corresponding 
\[
C:=\{ \mathsf P(z,w)=0\} \subset \mathcal N
\]
is called the \emph{spectral curve}. Kasteleyn's theorem on the torus~\cite{Kast} identifies the coefficients of $\mathsf P(z,w)$ with dimer partition functions in the various homology classes, and hence with the Hamiltonians of the cluster integrable system.

\subsection{Beauville integrable systems}

\begin{figure}
\begin{center}

\begin{tikzpicture}[scale=2, line cap=round, line join=round]
  \usetikzlibrary{calc}
  \def\L{2.2mm}   
  \def\Lc{1mm}    

  \coordinate (A) at (0,0);
  \coordinate (B) at (2,0);
  \coordinate (C) at (1,{sqrt(3)});

  \draw ($(A)!-0.25!(B)$) -- ($(A)!1.25!(B)$);
  \draw ($(B)!-0.25!(C)$) -- ($(B)!1.25!(C)$);
  \draw ($(C)!-0.25!(A)$) -- ($(C)!1.25!(A)$);

  \coordinate (P1)  at ($(A)!0.1!(B)$);
  \coordinate (P2)  at ($(A)!0.4!(B)$);
  \coordinate (P3)  at ($(A)!0.6!(B)$);
  \coordinate (P4)  at ($(A)!0.9!(B)$);

  \coordinate (P5)  at ($(B)!0.1!(C)$);
  \coordinate (P6)  at ($(B)!0.4!(C)$);
  \coordinate (P7)  at ($(B)!0.6!(C)$);
  \coordinate (P8)  at ($(B)!0.9!(C)$);

  \coordinate (P9)  at ($(C)!0.1!(A)$);
  \coordinate (P10) at ($(C)!0.4!(A)$);
  \coordinate (P11) at ($(C)!0.6!(A)$);
  \coordinate (P12) at ($(C)!0.9!(A)$);

\def\Ls{5.0mm}   
\def\Lm{2.2mm}   
\def\Lv{0.1mm}   

\draw[Cerulean,thick]
  (P1)  .. controls +(0:\Lv)    and +(240:\Ls)   .. (P2)
        .. controls +(60:\Lm)   and +(120:\Lm)   .. (P3)
        .. controls +(300:\Ls)  and +(180:\Lv)   .. (P4)

        .. controls +(0:\Lv)    and +(240:\Lv)   .. (P5)

        .. controls +(120:\Lv)  and +(0:\Ls)     .. (P6)
        .. controls +(180:\Lm)  and +(240:\Lm)   .. (P7)
        .. controls +(60:\Ls)   and +(300:\Lv)   .. (P8)

        .. controls +(120:\Lv)  and +(0:\Lv)     .. (P9)

        .. controls +(240:\Lv)  and +(120:\Ls)   .. (P10)
        .. controls +(300:\Lm)  and +(0:\Lm)     .. (P11)
        .. controls +(180:\Ls)  and +(60:\Lv)    .. (P12)

        .. controls +(240:\Lv)  and +(180:\Lv)   .. cycle;

\begin{scope}[shift={(0.68,0.42)}, rotate=18]
  \draw[thick,Cerulean] (0,0) ellipse (0.19 and 0.11);
  \fill[Red] (28:0.19 and 0.11) circle (1pt);
\end{scope}

\begin{scope}[shift={(1.26,0.50)}, rotate=-24]
  \draw[thick,Cerulean] (0,0) ellipse (0.215 and 0.125);
  \fill[Red] (205:0.215 and 0.125) circle (1pt);
\end{scope}

\begin{scope}[shift={(1.00,0.94)}, rotate=63]
  \draw[thick,Cerulean] (0,0) ellipse (0.165 and 0.095);
  \fill[Red] (120:0.165 and 0.095) circle (1pt);
\end{scope}

  \foreach \P in {P1,P2,P3,P4,P5,P6,P7,P8,P9,P10,P11,P12}
    \fill[fill=Dandelion] (\P) circle (1pt);
\end{tikzpicture}

\end{center}
\caption{An illustration of the Beauville integrable system on $\P^2$ associated with the standard triangle $N$ of side length $d=4$, for which $g=3$ (cf.~\cite{OkounkovTalk}). The three black lines are the three coordinate axes of $\P^2$ forming the toric boundary $D$, the blue curve is the spectral curve $C$, the red points are the divisor points $(p_i,q_i)_{i=1}^3$, and the green points are the boundary points $C\cap D$.} \label{Fig:beauville}
\end{figure}

On the other hand, the polygon $N$ determines a projective toric surface $\mathcal N$ whose dense torus is $(\C^\times)^2$ with coordinates $(z,w)$. The bivector
\[
\theta = (z\partial_z)\wedge (w\partial_w)
\]
is a torus-invariant Poisson structure on $(\C^\times)^2$ with pole locus along the toric boundary
\[
D:=\mathcal N\setminus (\C^\times)^2.
\]
A fundamental theme, going back to Mukai~\cite{Mukai}, Tyurin~\cite{Tyurin}, and Beauville~\cite{Beauville}, is that symplectic structures on a surface induce symplectic structures on moduli spaces of sheaves and on Hilbert schemes of points; Bottacin~\cite{Bottacin, Bott1} proved this more generally for Poisson surfaces. In our toric setting, we consider the phase space $\mathcal M$ consisting of
\begin{itemize}
\item $g$ points $(p_i,q_i)_{i=1}^g$ in the dense torus $(\C^\times)^2\subset \mathcal N$, where $g$ is the number of interior lattice points in $N$. 
\item a choice of appropriate boundary points on $D$.
\end{itemize}
for which the induced Poisson bracket takes the standard form
\[
\{\log p_i,\log q_j\}=\delta_{ij}.
\]
The symplectic leaves of $\mathcal M$ are algebraically completely integrable systems (cf.~\cite{Beauville,Bottacin,Bott1}) which we call \emph{Beauville integrable systems}.

These Hamiltonians also admit a natural interpretation via spectral curves. For generic choices of the points $(p_i,q_i)_{i=1}^g$ and the boundary points, there is a unique curve
\[
C\subset \mathcal N
\]
with Newton polygon $N$ passing through them, called the \emph{spectral curve} (see Figure~\ref{Fig:beauville}); the number $g$ of interior lattice points in $N$ is the genus of $C$. The Hamiltonians are the coefficients of its defining equation.

\subsection{The spectral transform}

The bridge between the two constructions is the \emph{spectral transform}. In their study of the moduli space of simple Harnack curves, Kenyon and Okounkov~\cite{KO} constructed a rational map
\[
\kappa:\mathcal X \dashrightarrow \mathcal M
\]
called the \emph{spectral transform}. We have already seen that both integrable systems have Hamiltonians that are described in terms of spectral curves, which the spectral transform identifies. To get the points $(p_i,q_i)_{i=1}^g$ from weights, consider the adjugate matrix
\[
\mathsf Q(z,w)=\mathsf P(z,w)\K^{-1}(z,w).
\]
The vanishing of any column of $\mathsf Q(z,w)$ determines $g$ points $(p_i,q_i)_{i=1}^g$ on the spectral curve, while the intersection of the spectral curve with the toric boundary determines the boundary points. In this way one obtains the map $\kappa$ which by construction identifies the spectral curves on the two sides.

\subsection{Main theorem and idea of proof}

Fock~\cite{Fock} and George--Goncharov--Kenyon~\cite{GGK} proved, by different methods, that the spectral transform is a birational map. By construction of the spectral transform, the Hamiltonians of the cluster integrable system are identified with those of the Beauville integrable system. The remaining question, and the main point of this paper, is whether the spectral transform also identifies the Poisson structures, and hence the integrable systems themselves.

The main result of this paper is a proof of the following conjecture of Goncharov--Kenyon (cf.~\cite{GK13}*{Theorem 1.4}) in the triangular case.

\begin{theorem} \la{thm1}
When $N$ is the standard triangle of side length $d$, so that $\mathcal N = \P^2$, the spectral transform is a birational isomorphism of integrable systems.
\end{theorem}

We now discuss the main idea in the proof of Theorem~\ref{thm1}. There is a standard way to rephrase the Beauville side in terms of sheaves on $\P^2$ (cf.~\cite{Hurtubise, Hurtubise2, HH}). Namely, let
\[
L = \mathcal O_{C}\Bigl(\sum_{i=1}^g (p_i,q_i)\Bigr)
\]
be the line bundle on $C$ corresponding to the divisor $\sum_{i=1}^g (p_i,q_i)$, and let
\[
\mathcal L := \iota_*L
\]
be the corresponding sheaf on $\P^2$ called the \emph{spectral sheaf}, where $\iota : C \hookrightarrow \P^2$ is the embedding. The tangent and cotangent spaces at $\mathcal L$ in the corresponding moduli space are computed by
\[
\Ext^1(\mathcal L,\mathcal L), \qquad \Ext^1(\mathcal L,\mathcal L(-D)),
\]
and in terms of these spaces the Poisson structure is induced by the canonical map
\[
\Ext^1(\mathcal L,\mathcal L(-D)) \xrightarrow[]{\theta} \Ext^1(\mathcal L,\mathcal L).
\]

The reason behind rephrasing the Poisson structure in terms of the spectral sheaf is that the Kasteleyn matrix gives a resolution of the spectral sheaf by locally free sheaves. More precisely, after passing to a natural finite cover of $\P^2$ with covering group $G$, the Kasteleyn matrix extends to a $G$-equivariant morphism of $G$-equivariant locally free sheaves whose cokernel is the spectral sheaf. This leads to simple and explicit computations of the relevant $\Ext$-groups via \v{C}ech-$\mathcal Hom$ complexes.

On the graph side, we reformulate the tangent and cotangent spaces of $\mathcal X$ in terms of the homology and relative cohomology of the conjugate ribbon graph. In these terms, the cluster Poisson structure becomes equally explicit. The striking point is that, after making these two descriptions concrete, the comparison maps become almost tautological.

Theorem~\ref{thm1} is conjecturally true for any $N$, not just triangles. We restrict in this paper to the simplest example, namely the hexagonal lattice or equivalently the case $\mathcal N=\P^2$, in order to focus on the algebro-geometric constructions without getting distracted by additional combinatorial complications. We expect the same strategy to extend to general $N$. Alternately, we also expect that the general case can be obtained from the triangular case by suitable degenerations of the weights.

\subsection*{Acknowledgements}
We thank Alexander Goncharov and Richard Kenyon, whose ideas have greatly influenced this work. Part of this work was completed while T. G. was visiting the University of Washington, and T. G. thanks the University of Washington for its hospitality. G. I. was partially supported by NSF Grant DMS-2502104.

\section{The cluster integrable system}

In this section we recall the cluster integrable system associated with the dimer model, following~\cite{GK13}. 

Let $\graph=(\B\sqcup\W,\E,\F)$ be the $d\times d$ fundamental domain of the hexagonal lattice embedded in the torus $\T$ shown in Figure~\ref{fig:hcd}. Here $\B$, $\W$, $\E$, and $\F$ denote the sets of black vertices, white vertices, edges and faces of $\graph$ respectively.

We identify
\[
H_1(\T,\Z)\cong \Z^2
\]
by taking as generators the homology classes of the images of the horizontal and vertical sides of the fundamental domain.

A \emph{zig-zag path} in $\graph$ is a closed path that turns maximally right at each black vertex and maximally left at each white vertex (Figure~\ref{fig:hcd}). Let $\zz$ denote the set of zig-zag paths in $\graph$. 

Each $\alpha\in\zz$ determines a homology class $[\alpha]\in H_1(\T,\Z)$, and the only possibilities are
\[
[\alpha]\in\{(-1,1),(0,-1),(1,0)\}.
\]
We write $\zz_0,\zz_1,$ and $\zz_2$ for the subsets of $\zz$ consisting of zig-zag paths with homology classes $(-1,1)$, $(0,-1)$, and $(1,0)$ respectively. Each $\zz_i$ consists of $d$ parallel zig-zag paths.

\subsection{Edge weights and gauge equivalence}

An \emph{edge weight} on $\Gamma$ is a function 
\[
\wt: \E \rightarrow \C^\times.
\]
Two edge weights $\wt_1$ and $\wt_2$ are said to be \emph{gauge equivalent} if there are functions 
\[
 f: \B \rightarrow \C^\times, \qquad g: \W \rightarrow \C^\times
\]
such that for every $\e=\b\w \in \E$,
\[
\wt_2(\e) = f(\b)^{-1} \wt_1(\e) g(\w). 
\]
We denote the gauge equivalence class of $\wt$ by $[\wt]$.

We define the \emph{$X$ cluster variety} $\mathcal X$ to be the space of edge weights modulo gauge transformations.

\begin{figure}
\begin{center}
\begin{tikzpicture}[scale=1]
\begin{scope}

    \def\dd{3};
    \draw[dashed, gray] (0,0) -- (2*\dd,0) -- (3*\dd,1.732*\dd) -- (1*\dd,1.732*\dd) -- (0,0);

    \foreach \i in {1,...,\dd}
    {
      \foreach \j in {1,...,\dd}
      {
        \pgfmathtruncatemacro{\ee}{\dd+1};
        \pgfmathtruncatemacro{\ff}{\dd};
        \coordinate[] (b\i0) at (2*\i-1,0);
        \coordinate[] (b0\j) at (\j-1+0.5,1.732*\j-1.732+0.866);
        \coordinate[label=above:$w$] (w\i\ee) at (2*\i+\dd-1,1.732*\dd);
        \coordinate[label=right:$z$] (w\ee\j) at (2*\dd+\j-1+0.5,1.732*\j-1.732+0.866);
        \coordinate[bvert] (b\i\j) at (2*\i+\j-1,1.732*\j-0.5773);
        \coordinate[wvert] (w\i\j) at (2*\i+\j-2,1.732*\j-1.732+0.5773);
        \draw[-] (w\i\j)--(b\i\j);
        \pgfmathtruncatemacro{\y}{\i-1};
        \pgfmathtruncatemacro{\x}{\j-1};
        \pgfmathtruncatemacro{\u}{\i+1};
        \pgfmathtruncatemacro{\v}{\j+1};
        \draw[-](b\y\j)--(w\i\j);
        \draw[-](b\i\x)--(w\i\j);
        \draw[-](w\i\ee)--(2*\i+\dd-1,1.732*\dd-0.5773);
        \draw[-](w\ee\j)--(2*\dd+\j-1,1.732*\j-0.5773);
      }
    }
\draw[-,red,line width=0.6mm,->] (b01)--(w11)--(b11)--(w21)--(b21)--(w31)--(b31)--(w41);
  \end{scope}
\end{tikzpicture}

\caption{The $3\times 3$ fundamental domain of the hexagonal lattice and a zig-zag path in $\zz_2$.}
\label{fig:hcd}
\end{center}
\end{figure}

\subsection{Graph homology and cohomology}\label{sec:graph_homology}

To express the above in the language of algebraic topology, we view \(\Gamma\) as a cell complex whose
\(0\)-cells are the vertices \(\B\sqcup \W\) and whose \(1\)-cells are the edges \(\e=\b\w \in \E\) oriented from $\b$ to $\w$. The cellular chain complex is
\[
\Bigl[
\begin{array}{c c c}
\Z^{\B}\oplus \Z^{\W} & \xleftarrow{\ \partial \ } & \Z^{\E}\\[-.4ex]
{\scriptstyle 0} && {\scriptstyle 1}
\end{array}
\Bigr], \qquad \partial(\e)=\w-\b \qquad (\e=\b\w).
\]
Throughout, subscripts under a complex indicate cohomological degrees. Thus,
\[
H_1(\Gamma,\Z)=H_1 \Bigl(\Bigl[
\begin{array}{c c c}
\Z^{\B}\oplus \Z^{\W} & \xleftarrow{\ \partial \ } & \Z^{\E}\\[-.4ex]
{\scriptstyle 0} && {\scriptstyle 1}
\end{array}
\Bigr]\Bigr) = \ker(\partial) = Z_1(\graph,\Z).
\]
In particular, every homology class on $\Gamma$ is represented by a unique cycle; we identify homology classes with cycles hereafter.

Dually, the cellular cochain complex is
\[
\Bigl[
\begin{array}{c c c}
(\C^\times)^{\B}\oplus (\C^\times)^{\W} & \xrightarrow{\ \delta \ } & (\C^\times)^{\E}\\[-.4ex]
{\scriptstyle 0} && {\scriptstyle 1}
\end{array}
\Bigr], \qquad \delta(f,g)(\e)=\frac{g(\w)}{f(\b)}
\qquad (\e=\b\w).
\]
An edge weight \(\wt\) is precisely a \(1\)-cochain, and two edge weights are gauge equivalent if and only if they differ by a \(1\)-coboundary. It follows that the space of gauge-equivalence classes of edge weights is
\be \label{eq:complex_H^1}
\mathcal X
=
H^1(\Gamma,\C^\times) = H^1 \Bigl( \Bigl[
\begin{array}{c c c}
(\C^\times)^{\B}\oplus (\C^\times)^{\W} & \xrightarrow{\ \delta \ } & (\C^\times)^{\E}\\[-.4ex]
{\scriptstyle 0} && {\scriptstyle 1}
\end{array}
\Bigr]\Bigr) = (\C^\times)^{\E}/\im(\delta).
\ee
If we replace \(\C^\times\) by \(\C\), the same construction gives the additive cohomology group \(H^1(\Gamma,\C)\).

\subsection{Monodromies, cluster variables and Casimirs}

Given a cycle $L \in H_1(\graph,\Z)$ and $[\wt] \in \mathcal X$, let 
\[
\mon_L([\wt]) := [\wt](L)
\]
denote the \emph{monodromy} of $[\wt]$ around $L$ i.e. the evaluation of $[\wt]$ on $L$. Concretely, if $L$ is written as 
\[
\w_1 \xrightarrow{\e_1} \b_1 \xrightarrow{\e_2} \w_2 \xrightarrow{\e_3} \b_2 \xrightarrow{\e_4} \cdots
\xrightarrow{\e_{2n-1}} \b_n \xrightarrow{\e_{2n}} \w_1,
\]
then
\begin{equation}\label{eq:monodromy}
\mon_L([\wt])=\prod_{i=1}^n \frac{\wt(\e_{2i})}{\wt(\e_{2i-1})} \in \C^\times.
\end{equation}
For every $L$, $\mon_L$ is a well-defined regular
function on $\mathcal X$ for every cycle $L$ in $\graph$.

Let $\partial \f$ denote the counterclockwise-oriented boundary of a face $\f \in \F$. The functions 
\[
X_\f := \mon_{\partial f}
\]
are called \emph{$X$ cluster variables}. Since 
\[
\sum_{\f \in \F} \partial \f = 0 \in H_1(\graph,\Z)
\]
they satisfy the unique relation $\prod_{\f \in \F} X_\f = 1$.

Similarly, for a zig-zag path $\alpha\in\zz$ we set
\[
C_\alpha:=\mon_\alpha.
\]
Since $\sum_{\alpha \in \zz} \alpha =0 \in H_1(\graph,\Z)$, these functions satisfy the unique relation
\[
\prod_{\alpha\in\zz} C_\alpha=1.
\]
They are called \emph{Casimirs} because they generate the center of the cluster Poisson structure (see~Section~\ref{sec:xpoisson}).

\subsection{Cluster Poisson structure on \texorpdfstring{$\mathcal X$}{X}}\label{sec:xpoisson}

\begin{figure}
\begin{center}
\begin{tikzpicture}[scale=1.7]
\begin{scope}
    \pgfmathsetmacro{\L}{2/sqrt(3)} 

    \coordinate[bvert,label=below:$\b$] (b) at (2,1.732-0.5773);
    \coordinate[wvert,label=above:$\w_0$] (w1) at (1,0.5773);
    \coordinate[wvert,label=above:$\w_1$] (w2) at (3,0.5773);
    \coordinate[wvert,label=above:$\w_2$] (w3) at (2,1.732+0.5773);

    \coordinate[bvert,label=below:$\b_{01}$] (b1) at ($(w1)+(150:\L)$);
    \coordinate[bvert,label=below:$\b_{02}$] (b2) at ($(w1)+(270:\L)$);

    \coordinate[bvert,label=below:$\b_{10}$] (b3) at ($(w2)+(30:\L)$);
    \coordinate[bvert,label=below:$\b_{12}$] (b4) at ($(w2)+(270:\L)$);

    \coordinate[bvert,label=below:$\b_{21}$] (b5) at ($(w3)+(150:\L)$);
    \coordinate[bvert,label=below:$\b_{20}$] (b6) at ($(w3)+(30:\L)$);

    \draw[-]
      (b) edge (w1) edge (w2) edge (w3)
      (w1) edge (b1) edge (b2)
      (w2) edge (b3) edge (b4)
      (w3) edge (b5) edge (b6);

    \coordinate (bt1) at ($(b)+(30:0.5)$);
    \coordinate (bt2) at ($(b)+(150:0.5)$);
    \coordinate (bt3) at ($(b)+(270:0.5)$);
    \draw (bt1) -- (bt2) -- (bt3) -- cycle;

    \coordinate (w1t1) at ($(w1)+(90:0.5)$);
    \coordinate (w1t2) at ($(w1)+(210:0.5)$);
    \coordinate (w1t3) at ($(w1)+(330:0.5)$);
    \draw (w1t1) -- (w1t2) -- (w1t3) -- cycle;

    \coordinate (w2t1) at ($(w2)+(90:0.5)$);
    \coordinate (w2t2) at ($(w2)+(210:0.5)$);
    \coordinate (w2t3) at ($(w2)+(330:0.5)$);
    \draw (w2t1) -- (w2t2) -- (w2t3) -- cycle;

    \coordinate (w3t1) at ($(w3)+(90:0.5)$);
    \coordinate (w3t2) at ($(w3)+(210:0.5)$);
    \coordinate (w3t3) at ($(w3)+(330:0.5)$);
    \draw (w3t1) -- (w3t2) -- (w3t3) -- cycle;

    \coordinate (b1t1) at ($(b1)+(30:0.5)$);
    \coordinate (b1t2) at ($(b1)+(150:0.5)$);
    \coordinate (b1t3) at ($(b1)+(270:0.5)$);
    \draw (b1t1) -- (b1t2) -- (b1t3) -- cycle;

    \coordinate (b2t1) at ($(b2)+(30:0.5)$);
    \coordinate (b2t2) at ($(b2)+(150:0.5)$);
    \coordinate (b2t3) at ($(b2)+(270:0.5)$);
    \draw (b2t1) -- (b2t2) -- (b2t3) -- cycle;

    \coordinate (b3t1) at ($(b3)+(30:0.5)$);
    \coordinate (b3t2) at ($(b3)+(150:0.5)$);
    \coordinate (b3t3) at ($(b3)+(270:0.5)$);
    \draw (b3t1) -- (b3t2) -- (b3t3) -- cycle;

    \coordinate (b4t1) at ($(b4)+(30:0.5)$);
    \coordinate (b4t2) at ($(b4)+(150:0.5)$);
    \coordinate (b4t3) at ($(b4)+(270:0.5)$);
    \draw (b4t1) -- (b4t2) -- (b4t3) -- cycle;

    \coordinate (b5t1) at ($(b5)+(30:0.5)$);
    \coordinate (b5t2) at ($(b5)+(150:0.5)$);
    \coordinate (b5t3) at ($(b5)+(270:0.5)$);
    \draw (b5t1) -- (b5t2) -- (b5t3) -- cycle;

    \coordinate (b6t1) at ($(b6)+(30:0.5)$);
    \coordinate (b6t2) at ($(b6)+(150:0.5)$);
    \coordinate (b6t3) at ($(b6)+(270:0.5)$);
    \draw (b6t1) -- (b6t2) -- (b6t3) -- cycle;

    \draw (w1t1) -- (bt2)
          (w1t3) -- (bt3);

    \draw (w2t1) -- (bt1)
          (w2t2) -- (bt3);

    \draw (w3t2) -- (bt2)
          (w3t3) -- (bt1);

    \draw (w1t2) -- (b1t3)
          (w1t1) -- (b1t1);

    \draw (w1t2) -- (b2t2)
          (w1t3) -- (b2t1);

    \draw (w2t1) -- (b3t2)
          (w2t3) -- (b3t3);

    \draw (w2t2) -- (b4t2)
          (w2t3) -- (b4t1);

    \draw (w3t2) -- (b5t3)
          (w3t1) -- (b5t1);

    \draw (w3t1) -- (b6t2)
          (w3t3) -- (b6t3);
   \draw[-,red,line width=0.6mm,->] (b1)--(w1)--(b)--(w2)--(b3);
\end{scope}

\begin{scope}[shift={(5,0)}]
    \pgfmathsetmacro{\L}{2/sqrt(3)} 

    \coordinate[bvert,label=below:$\b$] (b) at (2,1.732-0.5773);
    \coordinate[wvert,label=above:$\w_0$] (w1) at (1,0.5773);
    \coordinate[wvert,label=above:$\w_2$] (w2) at (3,0.5773);
    \coordinate[wvert,label=above:$\w_1$] (w3) at (2,1.732+0.5773);

    \coordinate[bvert,label=below:$\b_{01}$] (b1) at ($(w1)+(150:\L)$);
    \coordinate[bvert,label=below:$\b_{02}$] (b2) at ($(w1)+(270:\L)$);

    \coordinate[bvert,label=below:$\b_{21}$] (b3) at ($(w2)+(30:\L)$);
    \coordinate[bvert,label=below:$\b_{20}$] (b4) at ($(w2)+(270:\L)$);

    \coordinate[bvert,label=below:$\b_{10}$] (b5) at ($(w3)+(150:\L)$);
    \coordinate[bvert,label=below:$\b_{12}$] (b6) at ($(w3)+(30:\L)$);
   
    \draw[-]
      (b) edge (w1) edge (w2) edge (w3)
      (w1) edge (b1) edge (b2)
      (w2) edge (b3) edge (b4)
      (w3) edge (b5) edge (b6);

    \coordinate (bt1) at ($(b)+(30:0.5)$);
    \coordinate (bt2) at ($(b)+(150:0.5)$);
    \coordinate (bt3) at ($(b)+(270:0.5)$);
    \draw (bt1) -- (bt2) -- (bt3) -- cycle;

    \coordinate (w1t1) at ($(w1)+(90:0.5)$);
    \coordinate (w1t2) at ($(w1)+(210:0.5)$);
    \coordinate (w1t3) at ($(w1)+(330:0.5)$);
    \draw (w1t1) -- (w1t2) -- (w1t3) -- cycle;

    \coordinate (w2t1) at ($(w2)+(90:0.5)$);
    \coordinate (w2t2) at ($(w2)+(210:0.5)$);
    \coordinate (w2t3) at ($(w2)+(330:0.5)$);
    \draw (w2t1) -- (w2t2) -- (w2t3) -- cycle;

    \coordinate (w3t1) at ($(w3)+(90:0.5)$);
    \coordinate (w3t2) at ($(w3)+(210:0.5)$);
    \coordinate (w3t3) at ($(w3)+(330:0.5)$);
    \draw (w3t1) -- (w3t2) -- (w3t3) -- cycle;

    \coordinate (b1t1) at ($(b1)+(30:0.5)$);
    \coordinate (b1t2) at ($(b1)+(150:0.5)$);
    \coordinate (b1t3) at ($(b1)+(270:0.5)$);
    \draw (b1t1) -- (b1t2) -- (b1t3) -- cycle;

    \coordinate (b2t1) at ($(b2)+(30:0.5)$);
    \coordinate (b2t2) at ($(b2)+(150:0.5)$);
    \coordinate (b2t3) at ($(b2)+(270:0.5)$);
    \draw (b2t1) -- (b2t2) -- (b2t3) -- cycle;

    \coordinate (b3t1) at ($(b3)+(30:0.5)$);
    \coordinate (b3t2) at ($(b3)+(150:0.5)$);
    \coordinate (b3t3) at ($(b3)+(270:0.5)$);
    \draw (b3t1) -- (b3t2) -- (b3t3) -- cycle;

    \coordinate (b4t1) at ($(b4)+(30:0.5)$);
    \coordinate (b4t2) at ($(b4)+(150:0.5)$);
    \coordinate (b4t3) at ($(b4)+(270:0.5)$);
    \draw (b4t1) -- (b4t2) -- (b4t3) -- cycle;

    \coordinate (b5t1) at ($(b5)+(30:0.5)$);
    \coordinate (b5t2) at ($(b5)+(150:0.5)$);
    \coordinate (b5t3) at ($(b5)+(270:0.5)$);
    \draw (b5t1) -- (b5t2) -- (b5t3) -- cycle;

    \coordinate (b6t1) at ($(b6)+(30:0.5)$);
    \coordinate (b6t2) at ($(b6)+(150:0.5)$);
    \coordinate (b6t3) at ($(b6)+(270:0.5)$);
    \draw (b6t1) -- (b6t2) -- (b6t3) -- cycle;

    \draw (w1t1) -- (bt2)
          (w1t3) -- (bt3);

    \draw (w2t1) -- (bt1)
          (w2t2) -- (bt3);

    \draw (w3t2) -- (bt2)
          (w3t3) -- (bt1);

    \draw (w1t2) -- (b1t3)
          (w1t1) -- (b1t1);

    \draw (w1t2) -- (b2t2)
          (w1t3) -- (b2t1);

    \draw (w2t1) -- (b3t2)
          (w2t3) -- (b3t3);

    \draw (w2t2) -- (b4t2)
          (w2t3) -- (b4t1);

    \draw (w3t2) -- (b5t3)
          (w3t1) -- (b5t1);

    \draw (w3t1) -- (b6t2)
          (w3t3) -- (b6t3);
\draw[-,red,line width=0.6mm,->] (b1)--(w1)--(b)--(w3)--(b5);
\end{scope}
\end{tikzpicture}
\end{center}

\caption{The ribbon structure induced by the embedding of $\Gamma$ in $\T$ (left) and the conjugate surface $\widehat S$, obtained by reversing the cyclic order at each black vertex (right). The red paths illustrate the bijection between zig-zag paths of $\Gamma$ and boundary components of the conjugate surface (equivalently, faces of $\graph$ in $\widehat S$).} \la{fig:conj_surface}
\end{figure}

We define the conjugate surface and the cluster Poisson structure following \cite[\S1.1.1]{GK13}.

A \emph{ribbon structure} on $\Gamma$ is a cyclic ordering of the half-edges at each vertex. It determines a surface with boundary, called a \emph{ribbon graph}, obtained by replacing each vertex $\vert$ by a small triangle $\Delta_\vert$ and each edge $\e$ by a rectangle $\square_\e$, and gluing the rectangles to the triangles according to the chosen cyclic orders.

The embedding of $\Gamma$ in $\T$ induces such a ribbon structure (see Figure~\ref{fig:conj_surface}, left). The \emph{conjugate ribbon structure} is obtained by reversing the cyclic order at each black vertex. Equivalently, one cuts the ribbon graph along each black triangle $\Delta_\b$ and reglues it with the opposite orientation. We denote the resulting ribbon graph by $\widehat\Gamma$ and call it the \emph{conjugate ribbon graph}. Under this construction, the zig-zag paths in $\Gamma$ are in bijection with the boundary components of $\widehat\Gamma$. Gluing a disk to each boundary component produces a closed oriented surface $\widehat S$, called the \emph{conjugate surface}.

By construction, $\widehat\Gamma$ deformation retracts onto $\Gamma$. Hence
\[
H_1(\widehat\Gamma,\Z)\cong H_1(\Gamma,\Z).
\]
Let
\[
\iota_*:H_1(\widehat\Gamma,\Z)\rightarrow H_1(\widehat S,\Z)
\]
be the map induced by the inclusion $\iota:\widehat\Gamma\hookrightarrow \widehat S$. We define the pairing
\[
\langle L_1,L_2 \rangle:= \iota_*(L_1)\wedge \iota_*(L_2),
\qquad L_1,L_2\in H_1(\Gamma,\Z)\cong H_1(\widehat\Gamma,\Z),
\]
where $\wedge$ denotes the intersection pairing on the closed surface $\widehat S$ and we identify the cycle $L$ with its homology class as explained in Section~\ref{sec:graph_homology}. The \emph{cluster Poisson bracket} on $\mathcal X$ is defined as
\[
\{\mon_{L_1},\mon_{L_2}\}
=
\langle L_1,L_2 \rangle \mon_{L_1}\mon_{L_2},
\qquad L_1,L_2\in H_1(\Gamma,\Z).
\]

Let 
\[
\bm c = (c_\alpha)_{\alpha \in \zz}, \qquad c_\alpha \in \C^\times
\]
be a collection of nonzero complex numbers. The subvarieties
\[
\mathcal X_{\bm c} := \{[\wt] \in \mathcal X: C_\alpha([\wt]) = c_\alpha \text{ for all } \alpha \in \zz\}
\]
are the symplectic leaves of the cluster Poisson structure.

\subsection{The Kasteleyn matrix, spectral curve and integrability}

The \emph{Kasteleyn matrix} is the $\C[z^{\pm1},w^{\pm1}]$-module map
\begin{equation}\label{eq:::kast}
\K:\bigoplus_{\b\in \B}\C[z^{\pm1},w^{\pm1}]
\rightarrow
\bigoplus_{\w\in \W}\C[z^{\pm1},w^{\pm1}],
\end{equation}
whose $(\w,\b)$-entry is
\[
\K_{\w,\b}
:= \begin{cases}
\wt(\e) z^{a(\e)} w^{b(\e)} &\text{if $\e = \b \w$ is an edge},\\
0 &\text{otherwise}.
\end{cases}
\]
Here $a(\e)\in\{0,1\}$ (resp.\ $b(\e)\in\{0,1\}$) is defined by $a(\e)=1$ (resp.\ $b(\e)=1$) if $\e$
crosses the vertical (resp.\ horizontal) boundary of the fundamental domain, and $a(\e)=0$ (resp.\ $b(\e)=0$)
otherwise (see Figure~\ref{fig:hcd}).  

The polynomial
\[
\mathsf P:=\det \K
\]
is called the \textit{characteristic polynomial}.

A \emph{dimer cover} $\mathsf M$ of $\graph$ is a subset of $\E$ such that every vertex is used exactly once. 

\begin{theorem}[\cite{Kast}]
Let $[\wt]\in\mathcal X$, and fix a representative $\wt$. Then 
\[
\mathsf P = \sum_{\mathsf M\ \text{a dimer cover}} \epsilon(\mathsf M)\wt(\mathsf M) z^{a(\mathsf M)} w^{b(\mathsf M)},
\]
where
\[
a(\mathsf M):=\sum_{\e\in\mathsf M} a(\e),\qquad 
b(\mathsf M):=\sum_{\e\in\mathsf M} b(\e),\qquad
\wt(\mathsf M):=\prod_{\e\in\mathsf M}\wt(\e),
\]
and the sign $\epsilon(\mathsf M)\in\{\pm1\}$ is an explicit function of the pair $\bigl(a(\mathsf M),b(\mathsf M)\bigr)$.
\end{theorem}

\begin{figure} \begin{center} \begin{tikzpicture}[scale=1.5] \def\dd{2}; \draw[dashed, gray] (0,0) -- (2*\dd,0) -- (3*\dd,1.732*\dd) -- (1*\dd,1.732*\dd) -- (0,0); \foreach \i in {1,...,\dd} { \foreach \j in {1,...,\dd} { \pgfmathtruncatemacro{\ee}{\dd+1}; \pgfmathtruncatemacro{\ff}{\dd}; \coordinate[] (b\i0) at (2*\i-1,0); \coordinate[] (b0\j) at (\j-1+0.5,1.732*\j-1.732+0.866); \coordinate[label=above:${w}$] (w\i\ee) at (2*\i+\dd-1,1.732*\dd); \coordinate[label=right:${z}$] (w\ee\j) at (2*\dd+\j-1+0.5,1.732*\j-1.732+0.866); \coordinate[bvert,label=below:$\b_{\i \j}$] (b\i\j) at (2*\i+\j-1,1.732*\j-0.5773); \coordinate[wvert,label=left:$\w_{\i \j}$] (w\i\j) at (2*\i+\j-2,1.732*\j-1.732+0.5773); \draw[-] (w\i\j)--node[above]{$a_{\i \j}$}(b\i\j); \pgfmathtruncatemacro{\y}{\i-1}; \pgfmathtruncatemacro{\x}{\j-1}; \pgfmathtruncatemacro{\u}{\i+1}; \pgfmathtruncatemacro{\v}{\j+1}; \draw[-](b\y\j)--node[above]{$b_{\i \j}$}(w\i\j); \draw[-](b\i\x)--node[right]{$c_{\i \j}$}(w\i\j); \draw[-](w\i\ee)--(2*\i+\dd-1,1.732*\dd-0.5773); \draw[-](w\ee\j)--(2*\dd+\j-1,1.732*\j-0.5773); } } \end{tikzpicture} \end{center} \caption{A $\C^\times$-local system $\wt$ (red) on a $2 \times 2$ fundamental domain for the honeycomb lattice along with the $z$ and $w$ factors (green).} \la{fig:2x2} \end{figure}

\begin{example}
Consider the $2\times 2$ fundamental domain for the honeycomb lattice shown in Figure~\ref{fig:2x2}.
With the edge weights labeled in red and the monodromy variables in green, the Kasteleyn matrix is
\begin{align*}
\K=\begin{blockarray}{ccccc}
 \b_{11} & \b_{12} & \b_{21} & \b_{22}  \\
\begin{block}{[cccc]c}
  {a_{11}} & {{c_{11}} {w}} &  {b_{11}} {z}&0 &\w_{11}\\
  {c_{12}} & {a_{12}} &0& {b_{12}}{z}   &\w_{12}\\
{b_{21}}   & 0 &  {a_{21}} &{{c_{21}} {w}}&\w_{21}\\
  0& {b_{22}} &  {{c_{22}} } &{a_{22}}&\w_{22}\\
\end{block}
\end{blockarray}.
\end{align*}
Thus the characteristic polynomial
\begin{align}
\mathsf P
  &=a_{11} a_{12} a_{21} a_{22}-({a_{11} a_{12}
   c_{21} c_{22}}+a_{21} a_{22} c_{11} c_{12}) {w}-\nonumber\\
  &\quad-(a_{11} a_{21} b_{12}
   b_{22}+a_{12} a_{22} b_{11} b_{21})
  { z}+b_{12}
   b_{11} b_{21} b_{22} {z^2}-\nonumber\\
  &\quad-(b_{12} b_{21}
   c_{11} c_{22} +b_{11} b_{22} c_{12}
   c_{21}) {w z}+c_{11} c_{12} c_{21} c_{22} {w^2} .\la{p2x2}
\end{align}
\end{example}

Consider the standard open embedding
\[
(\C^\times)^2 \hookrightarrow \P^2,
\qquad
(z,w)\mapsto [1:z:w].
\]
Let \(C\subset \P^2\) be the projective closure of the curve \(\{\mathsf P=0\}\subset(\C^\times)^2\); we call \(C\) the \emph{spectral curve}. Although the Kasteleyn matrix \(\K\) and the polynomial \(\mathsf P\) depend on the chosen representative \(\wt\), the curves \(C^\circ\) and \(C\) depend only on the class \([\wt]\).

For generic \([\wt]\in\mathcal X\), the curve \(C\) is a plane curve of degree \(d\) (see~\cite[Theorem~1]{KO}); equivalently, the Newton polygon of \(\mathsf P\) is
\[
N=\operatorname{conv}\{(0,0),(d,0),(0,d)\}.
\]
Let \(N^\circ\) denote the interior of \(N\), and set
\[
N_{\Z}:=N\cap \Z^2,
\qquad
N_{\Z}^\circ:=N^\circ\cap \Z^2.
\]
If \(C\) is smooth, its genus is
\[
g=\binom{d-1}{2}=|N_{\Z}^\circ|.
\]

Normalize \(\mathsf P\) so that its constant term is \(H_{0,0}=1\), and write
\[
\mathsf P=\sum_{(i,j)\in N_{\Z}} H_{i,j}z^i w^j.
\]
With this normalization, each coefficient \(H_{i,j}\) depends only on \([\wt]\) (and not on the choice of representative \(\wt\)). 

\begin{theorem}[{\cite[Theorem~1.2]{GK13}}]
For generic $\bm c$, the symplectic leaf $\mathcal X_{\bm c}$ carries an algebraic completely integrable (Liouville) system whose Poisson-commuting Hamiltonians are the functions
\[
\bigl(H_{i,j}\bigr)_{(i,j)\in N^\circ_{\Z}}.
\]
\end{theorem}

\section{The Beauville integrable system}

In this section we recall the Poisson geometry underlying the Beauville integrable system for plane curves, and then introduce the equivariant variant that will be related to the dimer model.

\subsection{The Beauville--Bottacin Poisson structure} \label{sec:Beauville_Poisson_structure}

Fix $d$, and let $g=\binom{d-1}{2}$ denote the genus of a smooth plane curve of degree $d$.
Let $x_0,x_1,x_2$ be homogeneous coordinates on $\P^2$, and set
\[
z:=\frac{x_1}{x_0},\qquad w:=\frac{x_2}{x_0}
\]
for the coordinates on the dense torus $(\C^\times)^2\subset \P^2$.
Let
\[
D:=D_0\cup D_1\cup D_2,\qquad D_i:=\{x_i=0\}\subset \P^2
\]
be the toric boundary divisor. Note that $\cO_{\P^2}(D)\cong \cO_{\P^2}(3)\cong \omega_{\P^2}^{-1}$ and hence $\omega_{\P^2}\cong \cO_{\P^2}(-D)$.

    Let $\mathcal M$ be the moduli space of sheaves on $\P^2$ of the form \[\mathcal L=i_*L,\]
where $i:C\hookrightarrow \P^2$ is the embedding of a smooth degree $d$ curve and $L$ is a line bundle on $C$ of degree $g = \binom{d-1}{2}$, i.e. the genus of $C$.

Instead of giving a Poisson structure on ${\mathcal M}$ as a bivector field $\pi_{\mathcal M}\in H^0( {\mathcal M},\wedge^2 T {\mathcal M})$, we can encode it via the \emph{anchor map}
\[
\pi^\sharp_{\mathcal M}:T^*{\mathcal M} \rightarrow T {\mathcal M},
\]
so that
\[
\pi_{\mathcal M}(\alpha,\beta)=\alpha(\pi_{\mathcal M}^\sharp(\beta)), \qquad \alpha,\beta \in T^* {\mathcal M}.
\]

Standard deformation theory (see~\cite{Hartshorne}) identifies $T_{[\mathcal L]}\mathcal M$, the tangent space of $\mathcal M$ at $[\mathcal L]$, as follows
\[
T_{[\mathcal L]}\mathcal M\cong \Ext^1(\mathcal L,\mathcal L),
\]
and by Serre duality,
\[
T^*_{[\mathcal L]}\mathcal M\cong \Ext^1(\mathcal L,\mathcal L\otimes \omega_{\P^2})
\cong \Ext^1(\mathcal L,\mathcal L(-D)).
\]
Bottacin~\cite{Bottacin}, generalizing earlier work of Mukai~\cite{Mukai} and Tyurin~\cite{Tyurin},
showed that for any \[\theta\in H^0(\P^2,\omega_{\P^2}^{-1})\cong H^0(\P^2,\cO_{\P^2}(D)),\]
the natural map
\[
(\pi_{\mathcal M}^\sharp)_{[\mathcal L]}:\Ext^1(\mathcal L,\mathcal L(-D))
\xrightarrow{\ \theta\ }
\Ext^1(\mathcal L,\mathcal L),
\]
induced by the morphism of sheaves \[\mathcal L(-D)\xrightarrow{\ \theta\ } \mathcal L\]
is the anchor map of a Poisson structure on $\mathcal M$. We will take $\theta=x_0x_1x_2$, which restricts on $(\C^\times)^2$ to the bivector field
\[
(z\partial_z)\wedge(w\partial_w).
\]

\subsection{Equivariant modification and boundary labeling}\label{sec:equiv_bottacin}

The Poisson space $\mathcal M$ is the phase space of the Beauville integrable system in its classical form. For the dimer model, however, we need a modification that is defined on a finite cover and keeps track of additional boundary data.

More precisely, we will:
\begin{enumerate}
\item Pass to a $d^2$-fold cover $p:\widetilde\P^2\to\P^2$ with covering group $G$ and replace $\mathcal M$ by a moduli space $\widetilde{\mathcal M}$ of $G$-equivariant sheaves on $\widetilde\P^2$. 
\item Augment the moduli space by a discrete labeling of the points at infinity
$C\cap D$ by zig-zag paths.
\end{enumerate}
The remainder of the section carries out these two modifications and explains how the resulting Poisson space supports an algebraic completely integrable system whose Hamiltonians are the
interior coefficients of the defining polynomial of $C$.

\begin{remark}\label{rem:why-two-mods}
The second modification is easy to justify: without a labeling of the boundary points $C\cap D$, the spectral transform will only be a finite-to-one map (up to permuting these points).  Fixing a labeling removes this ambiguity and makes it birational.

The first modification is more subtle, and we postpone its justification to Remarks~\ref{rem:stack_1} and~\ref{rem:stack_2}. Here we simply note that it already appears in the work of Kenyon--Okounkov~\cite{KO} and in~\cite{TWZ}.
\end{remark}

Let $\widetilde\P^2$ and $\P^2$ be two copies of the projective plane, equipped with
homogeneous coordinates
\[
[\widetilde x_0:\widetilde x_1:\widetilde x_2]\in \widetilde\P^2,
\qquad
[x_0:x_1:x_2]\in \P^2.
\]
Consider the $d^2$-fold covering map
\[
p:\widetilde\P^2\rightarrow \P^2,
\qquad
[\widetilde x_0:\widetilde x_1:\widetilde x_2]\mapsto
[\widetilde x_0^{d}:\widetilde x_1^{d}:\widetilde x_2^{d}],
\]
which is ramified along $D$.

Let
\[
G:=\bmu_d^{\times 3}/\Delta,
\qquad
\Delta:=\{(\zeta,\zeta,\zeta):\zeta\in\bmu_d\}.
\]
The group $\bmu_d^{\times 3}$ acts on $\widetilde\P^2$ by coordinate-wise scaling,
\[
(\zeta_0,\zeta_1,\zeta_2)\cdot[\widetilde x_0:\widetilde x_1:\widetilde x_2]
=
[\zeta_0\widetilde x_0:\zeta_1\widetilde x_1:\zeta_2\widetilde x_2],
\]
and the diagonal subgroup $\Delta$ acts trivially, so the action factors through $G$.

We now introduce the space $\widetilde{\mathcal M}$, defined as a moduli space of $G$-equivariant sheaves together with certain discrete data. Roughly speaking, a $G$-equivariant sheaf is a sheaf equipped with a compatible action of $G$ called a \emph{$G$-linearization}. Since we will always study these sheaves through their locally free resolutions, it will be enough to understand $G$-equivariant vector bundles. A brief review of the necessary background is included in Appendix~\ref{appendix:G_sheaves}.

\begin{remark}\label{rem:stack_why}
Equivalently, one may work on the quotient stack
\[
 [\widetilde{\P}^2/G],
\]
whose coarse moduli space is $\P^2$. Giving a $G$-equivariant coherent sheaf on $\widetilde{\P}^2$ is the same as giving a coherent sheaf on $[\widetilde{\P}^2/G]$, and whenever we take $G$ invariants of a Hom/Ext group, this is exactly the corresponding Hom/Ext group computed on the stack $[\widetilde{\P}^2/G]$.

Readers comfortable with stacks are therefore welcome to regard the constructions and computations in
this section as taking place on $[\widetilde{\P}^2/G]$.
We have chosen the concrete $G$-equivariant language on the cover $\widetilde{\P}^2$ to avoid introducing
stacky notation and to keep the exposition more accessible.
\end{remark}

\begin{definition}
Let $\widetilde{\mathcal M}$ denote the moduli space of pairs $(\widetilde{\mathcal L},\bm\nu)$ consisting of:
\begin{enumerate}
\item a $G$-equivariant sheaf $\widetilde{\mathcal L}$ on $\widetilde\P^2$ of the form
\[
\widetilde{\mathcal L}=\widetilde i_*\widetilde L,
\]
where $\widetilde i:\widetilde C\hookrightarrow \widetilde\P^2$ is the inclusion of a smooth
$G$-invariant plane curve $\widetilde C=p^{-1}(C)$, with $C\subset \P^2$ a smooth plane curve of degree $d$, and where $\widetilde L$ is a $G$-equivariant line bundle on $\widetilde C$ of degree
\[
\widetilde g:=\binom{\widetilde d-1}{2},
\qquad\text{with}\qquad
\widetilde d:=\deg(\widetilde C)=d^2;
\]
\item a triple of bijections
\[
\bm\nu=(\nu_i)_{i=0,1,2},
\qquad
\nu_i:\zz_i \xrightarrow{\sim} C\cap D_i,
\]
called a \emph{labeling of the points at infinity of \(C\) by zig-zag paths}. 
Via these bijections, we henceforth identify each point \(\nu_i(\alpha)\in C\cap D_i\) with the corresponding zig-zag path \(\alpha\in\zz_i\).
\end{enumerate}
\end{definition}

Let
\[
\widetilde D:=\widetilde D_0\cup \widetilde D_1\cup \widetilde D_2,\qquad \widetilde D_i:=\{\widetilde x_i=0\}\subset \widetilde \P^2
\]
be the toric boundary divisor.

At a point $(\widetilde{\mathcal L},\nu)\in \widetilde{\mathcal M}$, the tangent and cotangent
spaces are given by $G$-invariants:
\[
T_{(\widetilde{\mathcal L},\bm\nu)}\widetilde{\mathcal M}
\cong \Ext^1(\widetilde{\mathcal L},\widetilde{\mathcal L})^{G},
\qquad
T^*_{(\widetilde{\mathcal L},\bm\nu)}\widetilde{\mathcal M}
\cong \Ext^1(\widetilde{\mathcal L},\widetilde{\mathcal L}(-\widetilde D))^{G}.
\]

The section
\[
\widetilde\theta:=\widetilde x_0\widetilde x_1\widetilde x_2 \in
H^0 \bigl(\widetilde\P^2,\mathcal O_{\widetilde\P^2}^G(\widetilde D)\bigr)
\]
is $G$-invariant. Multiplication by $\widetilde\theta$ induces a morphism of sheaves
\[
\widetilde{\mathcal L}(-\widetilde D)\xrightarrow{\ \widetilde\theta\ }\widetilde{\mathcal L},
\]
and hence, a map on $G$-invariants
\[
(\pi^\sharp_{\widetilde{\mathcal M}})_{(\widetilde{\mathcal L},\bm \nu)}:
\Ext^1(\widetilde{\mathcal L},\widetilde{\mathcal L}(-\widetilde D))^{G}
\xrightarrow{\ \widetilde\theta\ }
\Ext^1(\widetilde{\mathcal L},\widetilde{\mathcal L})^{G},
\]
which defines the anchor map of a Poisson structure on $\widetilde{\mathcal M}$.

Let 
\[
\mathsf P=\sum_{(i,j)\in N_{\Z}} H_{i,j}z^i w^j
\]
denote the polynomial defining $C$ normalized so that its constant term is \(H_{0,0}=1\). For $\alpha \in \zz_2$, we define
\[
C_\alpha := (-1)^d  z(\nu_2(\alpha))^{-1}
\]
and similarly $C_\alpha$ for $\alpha \in \zz_0,\zz_1$ is defined using toric coordinates on $D_0,D_1$ respectively. These functions are called \emph{Casimirs} because they generate the center of the Beauville Poisson structure.

Let 
\[
\bm c = (c_\alpha)_{\alpha \in \zz}, \qquad c_\alpha \in \C^\times
\]
be a collection of nonzero complex numbers. Let
\[
\widetilde{\mathcal M}_{\bm c}
\]
denote the subvariety where $C_\alpha = c_\alpha$ for all $\alpha \in \zz$. These are the symplectic leaves of $\widetilde{\mathcal M}$.

\begin{theorem}[{\cite{Beauville}}]
For generic values of the Casimirs $\bm c$, the symplectic leaf $\widetilde{\mathcal M}_{\bm c}$ carries an algebraic completely integrable (Liouville) system whose Poisson-commuting Hamiltonians are the functions
\[
\bigl(H_{i,j}\bigr)_{(i,j)\in N^\circ_{\Z}}.
\]
\end{theorem}

\section{Statement of the main theorem} \la{sec:spec}

\begin{figure}
\begin{center}\begin{tikzpicture}[scale=1.5] 
		\coordinate[bvert,label=below:$ \b$] (b) at (2,1.732-0.5773);
			\coordinate[wvert,label=below:$\w_0$] (w1) at (1,0.5773);
			\coordinate[wvert,label=below:$\w_1$] (w2) at (3,0.5773);
			\coordinate[wvert,label=above:$\w_2$] (w3) at (2,1.732+0.5773);

			\draw[-]
				  (b) edge  (w2) edge node[left] {$ c \widetilde w$} (w3) edge  (w1)
				;
				\node[] (no) at (2.85,1.) {$b \widetilde z$}; 
				\node[] (no) at (2-.75,1.) {$a $}; 
	\end{tikzpicture} \hspace{2cm}
	\begin{tikzpicture}[scale=1.5] 
		\coordinate[bvert,label=below:$ \b$] (b) at (2,1.732-0.5773);
			\coordinate[wvert,label=below:$\w_0$] (w1) at (1,0.5773);
			\coordinate[wvert,label=below:$\w_1$] (w2) at (3,0.5773);
			\coordinate[wvert,label=above:$\w_2$] (w3) at (2,1.732+0.5773);

			\draw[-]
				  (b) edge  (w2) edge node[left] {$ c\widetilde x_2$} (w3) edge  (w1)
				;
				\node[] (no) at (2.85,1.) {$b \widetilde x_1$}; 
				\node[] (no) at (2-.75,1.) {$a\widetilde x_0$}; 
	\end{tikzpicture}
\end{center}
\caption{The regularized Kasteleyn matrix $\widetilde\K$ on the dense torus (left) and its homogeneous form on $\widetilde\P^2$ (right).} \la{fig:homkast}
\end{figure}

In this section, following~\cite[Section~3.2]{KO}, we define the \emph{spectral transform}
\[
\kappa:\ \mathcal X \dashrightarrow \widetilde{\mathcal M},
\]
a rational map which assigns to a generic point $[\wt]\in\mathcal X$ a pair
\[
\kappa([\wt])=(\widetilde{\mathcal L},\bm \nu).
\]
In \cite{Fock,GGK}, it is shown that the spectral transform is birational, so it identifies the phase space of the cluster integrable system with that of the Beauville integrable system. Our main result is:

\begin{theorem}\label{thm:main}
The spectral transform is a birational isomorphism of integrable systems.
\end{theorem}

The rest of this section is devoted to the construction of the spectral transform.

\subsection{Pullback of the Kasteleyn matrix to the cover}

The Kasteleyn matrix $\K$ in~\eqref{eq:::kast} is equivalently a morphism of sheaves
\[
\K:{\bigoplus_{\b \in \B}} \cO_{(\C^\times)^2} \to {\bigoplus_{\w \in \W}} \cO_{(\C^\times)^2}
\]
where $(\C^\times)^2 \subset \P^2$ is the dense torus. Let $\widetilde{(\C^\times)}^2 $ denote the dense torus in $\widetilde\P^2$. On dense tori, the map $p:\widetilde\P^2\rightarrow \P^2$ is 
\[
p: \widetilde{(\C^\times)}^2  \to (\C^\times)^2, \qquad (\widetilde z,\widetilde w) \mapsto (\widetilde z^d,\widetilde w^d).
\]
As observed in \cite[Section~3.1]{KO}, the pullback 
\[
\widetilde \K = p^* \K
\]
is gauge equivalent to the matrix shown in Figure~\ref{fig:homkast}(left). 

\begin{remark} \label{rem:stack_1}
We have essentially ``{regularized}'' the Kasteleyn matrix by redistributing the monomial factors in $z$ and $w$ uniformly throughout the fundamental domain, rather than concentrating
them along the boundary. This will simplify the computations below, and it is
the one reason we pass to the cover $\widetilde\P^2$ instead of working directly on $\P^2$; see also~\ref{rem:stack_2}.
\end{remark}

\subsection{The discrete Abel map and the \texorpdfstring{$G$}{G}-equivariant extension of \texorpdfstring{$\widetilde\K$}{K}}

As already observed in \cite[Section~3.2]{KO}, $\widetilde \K$ can be viewed as a morphism of vector bundles 
\begin{equation} \label{eq:kast_extension}
    \widetilde \K:{\bigoplus_{\b \in \B}} \cO_{\widetilde \P^2}(-1) \to {\bigoplus_{\w \in \W}} \cO_{\widetilde \P^2}
\end{equation}
which in homogeneous coordinates is the matrix shown in
Figure~\ref{fig:homkast}(right). More precisely, to each edge \(\e\) we associate the monomial
\[
\widetilde x(\e):=\widetilde x_i,
\]
where \(\zz_i\) is the unique zig-zag class such that no zig-zag path in that class contains \(\e\). The \((\w,\b)\)-entry of \(\widetilde \K\) is then given by
\[
\widetilde \K_{\w,\b}
:=
\begin{cases}
\wt(\e)\widetilde x(\e), & \text{if } \e=\b\w \text{ is an edge},\\
0, & \text{otherwise}.
\end{cases}
\]

We will strengthen this by showing that $\widetilde\K$ extends to a
$G$-invariant morphism between $G$-equivariant locally free sheaves on $\widetilde\P^2$. The $G$-linearizations are constructed using Fock's \emph{discrete Abel map}~\cite{Fock}, a function
\[
\dabel:\B\sqcup\W \rightarrow \Z\{\widetilde D_0,\widetilde D_1,\widetilde D_2\},
\]
defined as follows.
\begin{enumerate}
\item Choose a base white vertex $\w_0\in\W$ (for concreteness, we take the white vertex at the bottom-left
corner of the fundamental domain) and set $\dabel(\w_0)=0$.
\item Let $\e=\b\w$ be an edge inside the fundamental domain with $\b\in\B$ and $\w\in\W$, and let $\alpha$ and $\beta$
be the two zig-zag paths containing $\e$. If $\alpha\in\zz_i$ and $\beta\in\zz_j$, then we set
\[
\dabel(\b)=\dabel(\w)+\widetilde D_i+\widetilde D_j.
\]
\end{enumerate}

For $\b\in\B$ and $\w\in\W$ define the $G$-equivariant line bundles on $\widetilde\P^2$
\[
\cale_{\b}:=\cO_{\widetilde\P^2}^G\bigl(\dabel(\b)-\widetilde D\bigr),
\qquad
\calf_{\w}:=\cO_{\widetilde\P^2}^G\bigl(\dabel(\w)\bigr),
\]
and set
\[
\cale:=\bigoplus_{\b\in\B}\cale_{\b},
\qquad
\calf:=\bigoplus_{\w\in\W}\calf_{\w}.
\]

The following proposition is a special case of a more general result proved in~\cite{GGK}.

\begin{proposition}\label{prop:kast_ext} With the $G$-linearizations on $\cale$ and $\calf$ determined by the discrete Abel map, the matrix $\widetilde\K$ defines a $G$-equivariant morphism \[ \widetilde\K:\ \cale \rightarrow \calf \] on $\widetilde\P^2$. \end{proposition}

\begin{proof}
It suffices to verify \(G\)-equivariance for each nonzero matrix entry. Let
\(\e=\b\w\) be an edge, and let \(\alpha\in\zz_i\) and \(\beta\in\zz_j\) be the two zig-zag paths containing \(\e\). If
\(k\in\{0,1,2\}\setminus\{i,j\}\), then
\[
\dabel(\b)=\dabel(\w)+\widetilde D_i+\widetilde D_j
+a(\e)d(\widetilde D_1-\widetilde D_0)
+b(\e)d(\widetilde D_2-\widetilde D_0).
\]
Hence, by definition of the line bundles \(\cale_\b\) and \(\calf_\w\),
\[
\cale_\b^\vee\otimes\calf_\w
\cong
\cO_{\widetilde\P^2}^G \left(
\widetilde D_k
-a(\e)d(\widetilde D_1-\widetilde D_0)
-b(\e)d(\widetilde D_2-\widetilde D_0)
\right)
\]
as \(G\)-linearized line bundles. The unique zig-zag class not containing \(\e\) is \(\zz_k\), so by definition
\(\widetilde x(\e)=\widetilde x_k\). By~\eqref{eq:h0G},
\[
H^0 \left(\widetilde\P^2,\cale_\b^\vee\otimes\calf_\w\right)^G
\cong
\C \{ \widetilde x_k \}.
\]
Thus \(\widetilde x(\e)\) is a \(G\)-invariant global section of
\(\cale_\b^\vee\otimes\calf_\w\), and so
\[
\widetilde\K_{\w,\b}=\wt(\e)\widetilde x(\e)
\]
is \(G\)-equivariant. 
\end{proof}

\begin{remark}\label{rem:stack_2}
The proof of Proposition~\ref{prop:kast_ext} shows more than extendability of $\K$ from $(\C^\times)^2$ to $\widetilde \P^2$. It canonically identifies 
\[
\Hom(\cale,\calf)^G \cong \C^\E.
\]
By contrast, on $\P^2$ one can always force an extension of $\K$ by twisting the source sufficiently negatively, but then the corresponding Hom-space of morphisms will be too large and there is no way to single out the distinguished collection of maps coming from edges. This is the main reason we work on $\widetilde\P^2$ rather than on $\P^2$.
\end{remark}

\subsection{The spectral transform}\label{sec:st}

By Proposition~\ref{prop:kast_ext}, the Kasteleyn matrix sits in an exact sequence of $G$-equivariant sheaves
\be \label{eq:kast_es}
0 \rightarrow \cale \xrightarrow{\ \widetilde\K\ } \calf \rightarrow \widetilde{\mathcal L} \rightarrow 0,
\ee
where
\[
\widetilde{\mathcal L}:=\coker(\widetilde\K).
\]
We call $\widetilde{\mathcal L}$ the \emph{spectral sheaf}. Its support is the $G$-invariant curve
\[
\widetilde C \ :=\ \{\det \widetilde\K = 0\}\ =\ p^{-1}(C)\ \subset\ \widetilde\P^2,
\]
where $C\subset\P^2$ denotes the spectral curve. For generic $[\wt]$, the curves $C$ and $\widetilde C$ are smooth, and $\widetilde{\mathcal L}$ is of the form $\widetilde i_* \widetilde L$, where $\widetilde L$ is a $G$-equivariant degree $\widetilde g$ line bundle on $\widetilde C$. The spectral sheaf is the first component of the spectral transform.

\begin{figure}
\begin{center}
	\begin{tikzpicture}[scale=1.5] 
		\coordinate[bvert,label=above:$\b_1$] (b1) at (2,1.732-0.5773);
		\coordinate[bvert,label=above:$\b_2$] (b2) at (4,1.732-0.5773);
			\coordinate[] (b0) at (0,1.732-0.5773);
		\coordinate[] (b3) at (6,1.732-0.5773);
			\coordinate[wvert,label=below:$\w_1$] (w1) at (1,0.5773);
			\coordinate[wvert,label=below:$\w_2$] (w2) at (3,0.5773);
			\coordinate[wvert,label=below:$\w_3$] (w3) at (5,0.5773);
	\draw[-] (w1) -- node[above]{$a_1 \widetilde x_0$}(b1)--node[above]{$b_1 \widetilde x_1$}(w2);
	\draw[-] (w2) -- node[above]{$a_2 \widetilde x_0$}(b2)--node[above]{$b_2 \widetilde x_1$}(w3);
	\draw[dashed] (b0)--node[above]{$b_d\widetilde x_1$}(w1);
		\draw[dashed] (b3)--node[above]{$a_3 \widetilde x_0$}(w3);
	\end{tikzpicture}
\end{center}
\caption{The Kasteleyn matrix $\widetilde \K$  restricted to a zig-zag path $\alpha \in \zz_2$. }\la{fig:kastzz}
\end{figure}

The second component records the intersection of the spectral curve with the boundary divisors. We explain this for the divisor $\widetilde D_2=\{\widetilde x_2=0\}$. Setting $\widetilde x_2=0$, the matrix $\widetilde\K$ becomes block-diagonal, with blocks indexed by zig-zag paths $\alpha\in\zz_2$. For such an $\alpha$, the corresponding block has the form
\begin{equation}\label{mat:kapl}
\widetilde \K_\alpha:=
\begin{bmatrix}
a_1 \widetilde x_0 & & &  b_d \widetilde x_1\\
b_1 \widetilde x_1 & a_2 \widetilde x_0 &&\\
& b_2 \widetilde x_1 & a_3 \widetilde x_0&\\
&&\ddots&\ddots
\end{bmatrix},
\end{equation}
where the edge weights along $\alpha$ are as in Figure~\ref{fig:kastzz}. This block is singular precisely at the points of $\widetilde D_2$ satisfying
\begin{equation} \label{eq:casimir}
    \Bigl(\frac{\widetilde x_1}{\widetilde x_0}\Bigr)^d
= (-1)^d\prod_{i=1}^d \frac{a_i}{b_i}
= (-1)^dC_\alpha^{-1},
\end{equation}
In particular, the image point on \(D_2=\{x_2=0\}\subset\P^2\) satisfies
\[
z  = (-1)^d C_\alpha^{-1}.
\]
We denote this boundary point by $\nu_2(\alpha)\in C\cap D_2$. As $\alpha$ ranges over $\zz_2$, these points are pairwise distinct for generic $[\wt]$, and hence determine a bijection
\[
\nu_2: \zz_2 \xrightarrow[]{\sim} C\cap D_2,\qquad \alpha \mapsto \nu_2(\alpha).
\]
The bijections $\nu_0$ and $\nu_1$ are defined analogously using the restrictions of $\widetilde \K$ to $\widetilde D_0$ and $\widetilde D_1$.

\begin{theorem}[\cite{Fock, GGK}]\label{thm:inverse}
The spectral transform $\kappa:\mathcal X \dashrightarrow \widetilde{\mathcal M}$ is birational.
\end{theorem}

\subsection{Plan of proof}
The remainder of the paper is devoted to the proof of Theorem~\ref{thm:main}.
In both integrable systems, the Hamiltonians are given by the interior coefficients of the spectral curve \(C\), and hence correspond under the spectral transform. Thus it remains to prove that the Poisson structures agree, or equivalently, that the spectral transform $\kappa$ is a Poisson map.

We formulate this in terms of the anchor maps
\[
\pi^\sharp_{\mathcal X}:T^*\mathcal X\to T\mathcal X,
\qquad
\pi^\sharp_{\widetilde{\mathcal M}}:T^*\widetilde{\mathcal M}\to T\widetilde{\mathcal M}
\]
of the two Poisson structures. Thus it is enough to prove that the following diagram commutes:
\begin{equation}\label{eq:small_square}
\begin{tikzcd}
 T^*_{[\wt]}\mathcal X \arrow[r,"\pi^\sharp_{\mathcal X}"]
 & T_{[\wt]}\mathcal X \arrow[d,"d\kappa"] \\
 T^*_{(\widetilde{\mathcal L},\bm\nu)}\widetilde{\mathcal M}
 \arrow[u,"d\kappa^*"]
 \arrow[r,"\pi^\sharp_{\widetilde{\mathcal M}}"]
 & T_{(\widetilde{\mathcal L},\bm\nu)}\widetilde{\mathcal M}.
\end{tikzcd}
\end{equation}

The proof is organized by giving concrete descriptions of the tangent and cotangent spaces on both sides and then identifying the corresponding maps. On each side there are two useful models for the cotangent space, and the key point is that these models match naturally across the spectral transform. Even more importantly, after passing to these concrete models, the comparison maps become almost tautological.

On the graph side, Section~\ref{sec:tangent_X} identifies the tangent space as
\[
T_{[\wt]}\mathcal X \cong H^1(\Gamma,\C)\cong H^1(\widehat\Gamma,\C),
\]
and introduces two models for the cotangent space:
\[
T^*_{[\wt]}\mathcal X
\cong H_1(\Gamma,\C)\cong H_1(\widehat\Gamma,\C),
\qquad
T^*_{[\wt]}\mathcal X
\cong H^1(\widehat\Gamma,\partial\widehat\Gamma;\C),
\]
the second being obtained from the first by Poincar\'e duality
\[
\PD: H^1(\widehat\Gamma,\partial\widehat\Gamma;\C) \xrightarrow[]{\cong} H_1(\widehat\Gamma,\C).
\]
Section~\ref{sec:complex_for_relcoh} computes \(H^1(\widehat\Gamma,\partial\widehat\Gamma;\C)\) explicitly, while Section~\ref{sec:pdandj} gives concrete formulas for the Poincar\'e duality isomorphism and for the natural map
\[
j^*:H^1(\widehat\Gamma,\partial\widehat\Gamma;\C)\to H^1(\widehat\Gamma,\C).
\]
With these identifications, the cluster Poisson anchor map factors as
\[
\pi_{\mathcal X}^\sharp: H_1(\widehat\Gamma,\C)
\xrightarrow{\ \PD^{-1} \ }
H^1(\widehat\Gamma,\partial\widehat\Gamma;\C)
\xrightarrow{\ j^*\ }
H^1(\widehat\Gamma,\C).
\]

On the sheaf side, Section~\ref{sec:dk} identifies the tangent space as
\[
T_{(\widetilde{\mathcal L},\bm\nu)}\widetilde{\mathcal M}
\cong
\Ext^1(\widetilde{\mathcal L},\widetilde{\mathcal L})^G.
\]
The cotangent space again has two models. The first is
\[
T^*_{(\widetilde{\mathcal L},\bm\nu)}\widetilde{\mathcal M}
\cong
\Ext^1(\widetilde{\mathcal L},\widetilde{\mathcal L}(-\widetilde D))^G,
\]
which is computed explicitly in Section~\ref{sec:ext2} and identified with the cotangent space via Serre duality in Section~\ref{sec:duality}. The second is
\[
T^*_{(\widetilde{\mathcal L},\bm\nu)}\widetilde{\mathcal M}
\cong \Ext^1\!\bigl(\widetilde{\mathcal L},
[\widetilde{\mathcal L}\to \widetilde{\mathcal L}|_{\widetilde D}]\bigr)^G,
\]
which is computed explicitly in Section~\ref{sec:ext-ld}. These two sheaf-theoretic models are related by the isomorphism induced by the short exact sequence
\[
0 \rightarrow\widetilde{\mathcal L}(-\widetilde D) \xrightarrow[]{\widetilde \theta} \widetilde{\mathcal L}\to \widetilde{\mathcal L}|_{\widetilde D} \rightarrow 0.
\]
Section~\ref{sec:anchor_sheaf} then computes the sheaf-side anchor map as the composition
\[
\pi^\sharp_{\widetilde{\mathcal M}}:
\Ext^1(\widetilde{\mathcal L},\widetilde{\mathcal L}(-\widetilde D))^{G}
\xrightarrow{\ \Psi^{-1}\ }
\Ext^1\!\bigl(\widetilde{\mathcal L},
[\widetilde{\mathcal L}\to \widetilde{\mathcal L}|_{\widetilde D}]\bigr)^G
\xrightarrow{\ \widetilde\theta\ }
\Ext^1(\widetilde{\mathcal L},\widetilde{\mathcal L})^{G}.
\]

Thus the proof of Theorem~\ref{thm:main} is reduced to showing that the diagram
\begin{equation}\label{eq:big_square}
\begin{tikzcd}
 H_1(\widehat \Gamma,\C) \arrow[r,"{\PD^{-1}}"]  & H^1(\widehat \Gamma,\partial \widehat \Gamma;\C) \arrow[r,"j^*"] \arrow[d,"\Phi"]
&  H^1(\Gamma,\C) \arrow[d,"d \kappa"] \\
\Ext^1(\widetilde{\mathcal L},\widetilde{\mathcal L}(-{\widetilde D}))^G \arrow[u,"d\kappa^*"]
 \arrow[r,"\Psi^{-1}"]
&\Ext^1\!\bigl(\widetilde{\mathcal L},[\widetilde{\mathcal L}\to \widetilde{\mathcal L}|_{\widetilde D}]\bigr)^G
 \arrow[r,"\widetilde \theta"]
& \Ext^1(\widetilde{\mathcal L},\widetilde{\mathcal L})^G
\end{tikzcd}
\end{equation}
commutes. The point is that, after all spaces are replaced by their concrete graph-theoretic models, the vertical maps in~\eqref{eq:big_square} are almost tautological: up to signs, they are obtained simply by inserting or removing the edge weight factors \(\wt\). Here \(\Phi\) is the comparison map between the two cotangent models, which is defined in Section~\ref{sec:last}. Section~\ref{sec:diff_kappa} computes the differential \(d\kappa\) of the spectral transform explicitly, and Section~\ref{sec:last} assembles all of these ingredients to prove that the anchor maps correspond.


\section{The tangent and cotangent spaces on \texorpdfstring{$\mathcal X$}{X}}
\label{sec:tangent_X}

\subsection{The tangent space}

Recall that the Zariski tangent space \(T_{[\wt]}\mathcal X\) is the set of morphisms
\[
\Spec\bigl(\C[\varepsilon]/(\varepsilon^2)\bigr)\to \mathcal X
\]
whose restriction to \(\varepsilon=0\) is \([\wt]\). Concretely, such a tangent vector is represented by an edge weight with values in
\(\bigl(\C[\varepsilon]/(\varepsilon^2)\bigr)^\times\) of the form
\[
\wt(\varepsilon)=\wt(1+\varepsilon u),
\]
modulo gauge transformations with values in
\(\bigl(\C[\varepsilon]/(\varepsilon^2)\bigr)^\times\).

\begin{proposition}\label{prop:xtan}
For any \([\wt]\in \mathcal X\), the map
\[
T_{[\wt]}\mathcal X \rightarrow H^1(\Gamma,\C),
\qquad
[\wt(1+\varepsilon u)] \mapsto [u]
\]
is an isomorphism.
\end{proposition}

\begin{proof}
Fix a representative edge weight \(\wt\) of the class \([\wt]\). Consider two tangent vectors $[\wt_1(\varepsilon)]$ and $[\wt_2(\varepsilon)]$ represented by
\[
\wt_i(\varepsilon)=\wt(1+\varepsilon u_i),
\qquad
u_i\in \C^{\E},
\]
for \(i=1,2\).

Suppose that \([\wt_1(\varepsilon)]=[\wt_2(\varepsilon)]\) determine the same tangent vector. Then there is a gauge transformation
\[
(f(\varepsilon),g(\varepsilon))
\in
\bigl(\C[\varepsilon]/(\varepsilon^2)\bigr)^{\B}
\times
\bigl(\C[\varepsilon]/(\varepsilon^2)\bigr)^{\W}
\]
such that for every edge \(\e=\b\w\),
\[
\wt_2(\varepsilon)(\e)
=
f(\varepsilon)(\b)^{-1}\,\wt_1(\varepsilon)(\e)\,g(\varepsilon)(\w).
\]
Write
\[
f(\varepsilon)=f_0(1+\varepsilon \phi),
\qquad
g(\varepsilon)=g_0(1+\varepsilon \psi),
\]
where
\[
f_0:\B\to\C^\times,\qquad g_0:\W\to\C^\times,
\qquad
\phi:\B\to\C,\qquad \psi:\W\to\C.
\]
Substituting
\[
\wt_i(\varepsilon)=\wt(1+\varepsilon u_i)
\]
into the gauge relation gives
\[
\wt(\e)\bigl(1+\varepsilon u_2(\e)\bigr)
=
f_0(\b)^{-1}(1-\varepsilon \phi(\b))\,
\wt(\e)\bigl(1+\varepsilon u_1(\e)\bigr)\,
g_0(\w)(1+\varepsilon \psi(\w)).
\]
Comparing constant terms, we obtain
\[
\wt(\e)=f_0(\b)^{-1}\wt(\e)g_0(\w),
\]
hence
\[
f_0(\b)=g_0(\w)
\qquad
\text{for every edge } \e=\b\w.
\]
After dividing the gauge transformation by this constant gauge, we may assume
\[
f_0\equiv 1,
\qquad
g_0\equiv 1.
\]
Then the above relation becomes
\[
1+\varepsilon u_2(\e)
=
(1-\varepsilon \phi(\b))(1+\varepsilon u_1(\e))(1+\varepsilon \psi(\w))
=
1+\varepsilon\bigl(u_1(\e)-\phi(\b)+\psi(\w)\bigr),
\]
so
\[
u_2(\e)=u_1(\e)-\phi(\b)+\psi(\w).
\]
Equivalently,
\[
u_2-u_1=\delta(\phi,\psi).
\]
Therefore two first-order deformations determine the same tangent vector if and only if the corresponding \(1\)-cochains differ by a \(1\)-coboundary. This shows that the map in the statement is well defined and injective.

Conversely, given any \(u\in \C^{\E}\), the formula
\[
\wt(\varepsilon):=\wt(1+\varepsilon u)
\]
defines a first-order deformation of \(\wt\), hence a tangent vector at \([\wt]\). If \(u'=u+\delta(\phi,\psi)\), then
\[
\wt(1+\varepsilon u')
=
(1+\varepsilon \phi)^{-1}\,\wt(1+\varepsilon u)\,(1+\varepsilon \psi),
\]
so the deformations defined by \(u\) and \(u'\) are gauge equivalent. Thus the tangent vector depends only on the class of \(u\) in
\[
\C^{\E}/\im(\delta)=H^1(\Gamma,\C).
\]
This proves surjectivity.
\end{proof}

\subsection{Two models for the cotangent space}

By Proposition~\ref{prop:xtan} and the identification
\[
H_1(\Gamma,\C) \cong H^1(\Gamma,\C)^*,
\]
we obtain the first model for the cotangent space on \(\mathcal X\):
\[
T^*_{[\wt]}\mathcal X \cong H_1(\Gamma,\C)\cong H_1(\widehat\Gamma,\C).
\]
The second model is obtained via the Poincar\'e duality isomorphism
\[
\PD: H^1(\widehat\Gamma,\partial\widehat\Gamma;\C) \rightarrow H_1(\widehat\Gamma,\C).
\]
We now compute \(H^1(\widehat\Gamma,\partial\widehat\Gamma;\C)\) concretely.

\subsection{A complex computing \texorpdfstring{$H^1(\widehat\Gamma,\partial\widehat\Gamma;\C)$}{the relative cohomology}}
\label{sec:complex_for_relcoh}

For \(\alpha\in\zz\) and \(\b\in\B\cap\alpha\), let
\[
\w_\alpha^{-}(\b),\ \w_\alpha^{+}(\b)\in \W\cap\alpha
\]
denote the white vertices immediately preceding and following \(\b\) along \(\alpha\). Thus
\[
\w_\alpha^{-}(\b)\;-\;\b\;-\;\w_\alpha^{+}(\b)
\]
is the \emph{wedge of \(\alpha\) at \(\b\)} (see~Figure~\ref{fig:wedge}).

\begin{figure}
\begin{center}
\begin{tikzpicture}[scale=1]
\coordinate[bvert,label=above:$\b$] (b) at (2,1.732-0.5773);
\coordinate[wvert,label=below:$\w_\alpha^{-}(\b)$] (w1) at (1,0.5773);
\coordinate[wvert,label=below:$\w_\alpha^{+}(\b)$] (w2) at (3,0.5773);
\draw[-] (b) edge (w2) edge (w1);
\end{tikzpicture}
\end{center}
\caption{The wedge of $\alpha$ at the black vertex \(\b\).}
\la{fig:wedge}
\end{figure}

\begin{figure}
\begin{center}
\begin{tikzpicture}[scale=1.7]
\begin{scope}
    \pgfmathsetmacro{\L}{2/sqrt(3)} 

    \coordinate[bvert,label=below:$\b$] (b) at (2,1.732-0.5773);
    \coordinate[wvert,label=93:$\w_0$] (w1) at (1,0.5773);
    \coordinate[wvert,label=93:$\w_1$] (w2) at (3,0.5773);
    \coordinate[wvert,label=93:$\w_2$] (w3) at (2,1.732+0.5773);

    \coordinate[bvert,label=below:$\b_{01}$] (b1) at ($(w1)+(150:\L)$);
    \coordinate[bvert,label=below:$\b_{02}$] (b2) at ($(w1)+(270:\L)$);

    \coordinate[bvert,label=below:$\b_{10}$] (b3) at ($(w2)+(30:\L)$);
    \coordinate[bvert,label=below:$\b_{12}$] (b4) at ($(w2)+(270:\L)$);

    \coordinate[bvert,label=below:$\b_{21}$] (b5) at ($(w3)+(150:\L)$);
    \coordinate[bvert,label=below:$\b_{20}$] (b6) at ($(w3)+(30:\L)$);

    \draw[-]
      (b) edge (w1) edge (w2) edge (w3)
      (w1) edge (b1) edge (b2)
      (w2) edge (b3) edge (b4)
      (w3) edge (b5) edge (b6);

    \coordinate (bt1) at ($(b)+(30:0.5)$);
    \coordinate (bt2) at ($(b)+(150:0.5)$);
    \coordinate (bt3) at ($(b)+(270:0.5)$);

    \coordinate (w1t1) at ($(w1)+(90:0.5)$);
    \coordinate (w1t2) at ($(w1)+(210:0.5)$);
    \coordinate (w1t3) at ($(w1)+(330:0.5)$);
    \draw (w1) edge (w1t1) edge (w1t2) edge (w1t3) ;

    \coordinate (w2t1) at ($(w2)+(90:0.5)$);
    \coordinate (w2t2) at ($(w2)+(210:0.5)$);
    \coordinate (w2t3) at ($(w2)+(330:0.5)$);
    \draw (w2) edge (w2t1) edge (w2t2) edge (w2t3) ;

    \coordinate (w3t1) at ($(w3)+(90:0.5)$);
    \coordinate (w3t2) at ($(w3)+(210:0.5)$);
    \coordinate (w3t3) at ($(w3)+(330:0.5)$);
    \draw (w3) edge (w3t1) edge (w3t2) edge (w3t3) ;

    \coordinate (b1t1) at ($(b1)+(30:0.5)$);
    \coordinate (b1t2) at ($(b1)+(150:0.5)$);
    \coordinate (b1t3) at ($(b1)+(270:0.5)$);

    \coordinate (b2t1) at ($(b2)+(30:0.5)$);
    \coordinate (b2t2) at ($(b2)+(150:0.5)$);
    \coordinate (b2t3) at ($(b2)+(270:0.5)$);

    \coordinate (b3t1) at ($(b3)+(30:0.5)$);
    \coordinate (b3t2) at ($(b3)+(150:0.5)$);
    \coordinate (b3t3) at ($(b3)+(270:0.5)$);

    \coordinate (b4t1) at ($(b4)+(30:0.5)$);
    \coordinate (b4t2) at ($(b4)+(150:0.5)$);
    \coordinate (b4t3) at ($(b4)+(270:0.5)$);

    \coordinate (b5t1) at ($(b5)+(30:0.5)$);
    \coordinate (b5t2) at ($(b5)+(150:0.5)$);
    \coordinate (b5t3) at ($(b5)+(270:0.5)$);

    \coordinate (b6t1) at ($(b6)+(30:0.5)$);
    \coordinate (b6t2) at ($(b6)+(150:0.5)$);
    \coordinate (b6t3) at ($(b6)+(270:0.5)$);

    \draw (w1t1) -- (bt2)
          (w1t3) -- (bt3);

    \draw (w2t1) -- (bt1)
          (w2t2) -- (bt3);

    \draw (w3t2) -- (bt2)
          (w3t3) -- (bt1);

    \draw (w1t2) -- (b1t3)
          (w1t1) -- (b1t1);

    \draw (w1t2) -- (b2t2)
          (w1t3) -- (b2t1);

    \draw (w2t1) -- (b3t2)
          (w2t3) -- (b3t3);

    \draw (w2t2) -- (b4t2)
          (w2t3) -- (b4t1);

    \draw (w3t2) -- (b5t3)
          (w3t1) -- (b5t1);

    \draw (w3t1) -- (b6t2)
          (w3t3) -- (b6t3);

\end{scope}
\end{tikzpicture}
\end{center}

\caption{The cellular decomposition of \(\widehat \Gamma\) used to compute \(H^1(\widehat\Gamma,\partial\widehat\Gamma;\C)\).}
\la{fig:cell_decomp}
\end{figure}

\begin{proposition}\label{prop:relcoh}
There is an identification
\begin{equation}\label{eq:relcoh-concrete}
\begin{aligned}
H^1(\widehat\Gamma,\partial\widehat\Gamma;\C)
\cong
H^1\Bigl(\Bigl[
&
\underset{0}{\C^{\B}\oplus \C^{\W}}
\xrightarrow{\,d^{0}\,}
\underset{1}{\C^{\E}\oplus \displaystyle\bigoplus_{\alpha\in\zz}\Hom(\C^{\W\cap\alpha},\C)}\\
&\xrightarrow{\,d^{1}\,}
\underset{2}{\displaystyle\bigoplus_{\alpha\in\zz}\Hom(\C^{\B\cap\alpha},\C)}
\Bigr]\Bigr).
\end{aligned}
\end{equation}
The differentials are given by
\begin{align*}
d^{0}(f,g)
&=
\left(
\bigl(g(\w)-f(\b)\bigr)_{\e=\b\w\in\E},
\ (-g|_{\W\cap\alpha})_{\alpha\in\zz}
\right),\\
d^{1}(u,t)_\alpha
&=
\left(
u\bigl(\b \w_\alpha^{-}(\b)\bigr)
-
u\bigl(\b \w_\alpha^{+}(\b)\bigr)
+
t_\alpha\bigl(\w_\alpha^{-}(\b)\bigr)
-
t_\alpha\bigl(\w_\alpha^{+}(\b)\bigr)
\right)_{\b\in \B\cap\alpha},
\end{align*}
for each \(\alpha\in\zz\), where \(t=(t_\alpha)_{\alpha\in\zz}\).
\end{proposition}

\begin{proof}
Consider the cellular decomposition of \(\widehat\Gamma\) shown in Figure~\ref{fig:cell_decomp}. The \(0\)-cells are the vertices \(\B\sqcup \W\) of \(\Gamma\) together with the boundary vertex
\[
v_{\w,\alpha}:=\Delta_\w\cap \alpha
\]
for each \(\alpha\in\zz\) and \(\w\in \W\cap\alpha\). The \(1\)-cells are the edges \(\e=\b\w\) of \(\Gamma\), oriented from \(\b\) to \(\w\), together with, for each \(\alpha\in\zz\) and \(\w\in \W\cap\alpha\), the line segment \(r_{\w,\alpha}\) oriented from \(\w\) to \(v_{\w,\alpha}\), and the boundary segments of \(\partial\widehat\Gamma\) joining consecutive vertices \(v_{\w,\alpha}\). The \(2\)-cells are the topological disks cut out by these \(0\)- and \(1\)-cells.

For each \(\alpha\in\zz\) and \(\b\in \B\cap\alpha\), let \(D_{\b,\alpha}\) denote the unique \(2\)-cell adjacent to the boundary arc of \(\partial\widehat\Gamma\) indexed by \(\b\), oriented so that its relative boundary is
\begin{equation}\label{eq:rel_boundary_cell_prop}
\partial D_{\b,\alpha}
=
(\b\,\w_\alpha^{-}(\b))
-
(\b\,\w_\alpha^{+}(\b))
+
r_{\w_\alpha^{-}(\b),\alpha}
-
r_{\w_\alpha^{+}(\b),\alpha}.
\end{equation}

In the relative cellular chain complex, the boundary vertices \(v_{\w,\alpha}\) and the boundary segments are zero. Hence the degree-\(0\) cochains are indexed by \(\B\sqcup\W\), the degree-\(1\) cochains by the edges of \(\Gamma\) together with the segments \(r_{\w,\alpha}\), and the degree-\(2\) cochains by the disks \(D_{\b,\alpha}\). This identifies the relative cellular cochain complex \(C^\bullet(\widehat\Gamma,\partial\widehat\Gamma;\C)\) with the complex in~\eqref{eq:relcoh-concrete}.

The formula for \(d^0\) follows from
\[
\partial(\b\w)=\w-\b,
\qquad
\partial r_{\w,\alpha}=v_{\w,\alpha}-\w = -\w.
\]
Similarly, the formula for \(d^1\) is obtained by dualizing~\eqref{eq:rel_boundary_cell_prop}.
\end{proof}

\section{The anchor map on the graph side}\label{sec:anchor_cps}

\subsection{The cluster Poisson anchor map}

Recall from Proposition~\ref{prop:xtan} that
\[
T_{[\wt]}\mathcal X \cong H^1(\Gamma,\C)\cong H^1(\widehat\Gamma,\C),
\qquad
T^*_{[\wt]}\mathcal X \cong H_1(\Gamma,\C)\cong H_1(\widehat\Gamma,\C).
\]
For \(L \in H_1(\Gamma,\Z)\), the differential \(d\mon_L\in T^*_{[\wt]}\mathcal X\) is determined by the first-order expansion
\[
\mon_L([\wt(1+\varepsilon u)])
=
\mon_L([\wt])\bigl(1+\varepsilon\,\mon_L([u])\bigr).
\]
Hence under the identification in Proposition~\ref{prop:xtan},
\be \label{eq:diff}
d \mon_L \in T^*_{[\wt]}\mathcal X \mapsto \mon_L([\wt]) \cdot  L \in H_1(\graph,\C).
\ee
Thus, the anchor map of the cluster Poisson structure is given by
\[
\pi^\sharp_{\mathcal X}:H_1(\graph,\C) \rightarrow H^1(\graph,\C), \qquad L_1 \mapsto (L_2 \mapsto \langle L_1,L_2 \rangle ).
\]

\subsection{Poincar\'e duality in coordinates}

Next, we compute the Poincar\'e duality isomorphism.

Each edge \(\e=\b\w\in\E\) is contained in exactly two zig-zag paths. We denote them by \(\alpha_\e^+\) and \(\alpha_\e^-\), where \(\alpha_\e^+\) traverses \(\e\) from \(\w\) to \(\b\) and \(\alpha_\e^-\) traverses \(\e\) from \(\b\) to \(\w\).

\begin{proposition}\label{prop:PDiso}
With the identification in Proposition~\ref{prop:relcoh}, the isomorphism
\[
\PD:H^1(\widehat\Gamma,\partial\widehat\Gamma;\C)\xrightarrow{\cong} H_1(\widehat\Gamma,\C) \cong \ker(\C^\E \xrightarrow[]{\partial} \C^{\B} \oplus \C^\W)
\]
is induced by the map
\[
\C^\E\oplus \bigoplus_{\alpha\in\zz}\Hom(\C^{\W\cap\alpha},\C)
\rightarrow \C^\E,\qquad
(u,t)
\mapsto
\bigl(t_{\alpha_\e^+}(\w)-t_{\alpha_\e^-}(\w)\bigr)_{\e=\b\w\in \E}.
\]
\end{proposition}

\begin{proof}
For each edge \(\e=\b\w\in\E\), choose an oriented arc
\[
\e^\vee\subset \square_\e
\]
from the boundary component \(\alpha_\e^{-}\) to the boundary component \(\alpha_\e^{+}\). By our choice of orientation, the edge \(\e\) and the arc \(\e^\vee\) have local intersection number \(1\). The relative homology classes \([\e^\vee]\) generate \(H_1(\widehat\Gamma,\partial\widehat\Gamma;\C)\), and the intersection pairing
\[
\wedge : H_1(\widehat\Gamma,\C)\otimes H_1(\widehat\Gamma,\partial\widehat\Gamma;\C)\to \C
\]
is induced by
\[
\e_1\otimes \e_2^\vee \mapsto \delta_{\e_1,\e_2}.
\]
Moreover, \(\e^\vee\) is homologous in \(H_1(\widehat\Gamma,\partial\widehat\Gamma;\C)\) to
\[
r_{\w,\alpha_\e^+}-r_{\w,\alpha_\e^-}.
\]
The intersection pairing induces the isomorphism \(\PD\). Since \([u,t]\) evaluates on \([\e^\vee]=[r_{\w,\alpha_\e^+}-r_{\w,\alpha_\e^-}]\) to give \(t_{\alpha_\e^+}(\w)-t_{\alpha_\e^-}(\w)\), the isomorphism is induced by the map
\[
(u,t) \mapsto \sum_{\e \in \E} (t_{\alpha_\e^+}(\w)-t_{\alpha_\e^-}(\w)) \e
\]
as claimed.
\end{proof}

\subsection{Factorization through relative cohomology}

We now rephrase the anchor map in terms of Poincar\'e duality on \(\widehat \Gamma\). Let
\[
\PD:
H^1(\widehat\Gamma,\partial\widehat\Gamma;\C) \xrightarrow{\cong}
H_1(\widehat\Gamma,\C)
\]
be the Poincar\'e duality isomorphism, and let
\[
j^*:H^1(\widehat\Gamma,\partial\widehat\Gamma;\C)\to H^1(\widehat\Gamma,\C)
\]
be the natural map obtained by forgetting the relative structure.

\begin{proposition}
Under the identification in Proposition~\ref{prop:xtan}, the anchor map \(\pi^\sharp_{\mathcal X}\) is given by the composition
\[
H_1(\widehat\Gamma,\C) \xrightarrow[]{\PD^{-1}} H^1(\widehat\Gamma,\partial\widehat\Gamma;\C)\xrightarrow[]{j^*} H^1(\widehat\Gamma,\C).
\]
\end{proposition}

\begin{proof}
For \(L_1\in H_1(\widehat\Gamma,\C)\), the class
\[
\PD^{-1}(L_1)\in H^1(\widehat\Gamma,\partial\widehat\Gamma;\C)
\]
is the relative cohomology class whose evaluation on a relative \(1\)-cycle \(L_2\) is the intersection number
\[
L_2\mapsto \langle L_1,L_2\rangle.
\]
Under the identification
\[
H^1(\widehat\Gamma,\C)\cong H_1(\widehat \Gamma,\C)^*,
\]
the map
\[
j^*:H^1(\widehat\Gamma,\partial\widehat\Gamma;\C)\to H^1(\widehat\Gamma,\C)
\]
is dual to the natural map
\[
j:H_1(\widehat\Gamma,\C)\to H_1(\widehat\Gamma,\partial\widehat\Gamma;\C)
\]
which views an absolute cycle as a relative cycle. Therefore, for every
\(L_2\in H_1(\widehat\Gamma,\C)\),
\[
\bigl(j^*\PD^{-1}(L_1)\bigr)(L_2)
=
\PD^{-1}(L_1)\bigl(j(L_2)\bigr)
=
\langle L_1,L_2\rangle.
\]
Hence \(j^*\PD^{-1}(L_1)\) is the functional
\[
L_2\mapsto \langle L_1,L_2\rangle,
\]
which is exactly the anchor map.
\end{proof}

\subsection{The map \texorpdfstring{$j^*$}{j*} in coordinates}
\label{sec:pdandj}

\begin{proposition}\label{prop:jstar}
With the identification in Proposition~\ref{prop:relcoh}, the map
\[
j^*:H^1(\widehat\Gamma,\partial\widehat\Gamma;\C)\to H^1(\widehat\Gamma,\C)\cong H^1(\Gamma,\C)
\]
is induced by the map
\[
\C^\E\oplus \bigoplus_{\alpha\in\zz}\Hom(\C^{\W\cap\alpha},\C)
\rightarrow \C^\E,
\qquad
(u,t)\mapsto u.
\]
\end{proposition}

\begin{proof}
Let \([u,t]\in H^1(\widehat\Gamma,\partial\widehat\Gamma;\C)\). Under the identification
\[
H^1(\Gamma,\C)\cong H_1(\Gamma,\C)^*,
\]
it suffices to compute the functional corresponding to \(j^*([u,t])\).

Let \(L\in H_1(\Gamma,\C)\). Viewing \(L\) as an absolute \(1\)-cycle in \(\widehat\Gamma\), its support is contained in the graph \(\Gamma\), so it involves only the edges \(\e\in\E\). Hence
\[
j^*([u,t])(L)
=
[u,t](j(L))
=
\sum_{\e\in\E} u(\e)\,L(\e).
\]
Therefore \(j^*([u,t])\) is exactly the cohomology class represented by \(u\in \C^\E\). This proves the claim.
\end{proof}

\section{The tangent and cotangent spaces on the sheaf side}

\subsection{The tangent space}\la{sec:dk}

The goal of this subsection is to prove the following explicit description of $\Ext^\bullet_{\widetilde\P^2}(\widetilde{\mathcal L},\widetilde{\mathcal L})^{G}$.

\begin{proposition}\label{prop:repext}
There is an identification
\[
\Ext^\bullet_{\widetilde\P^2}(\widetilde{\mathcal L},\widetilde{\mathcal L})^{G}
\cong
H^\bullet\Bigl(\Bigl[
\begin{array}{c c c}
\C^{\B}\oplus \C^{\W} & \xrightarrow{\ d\ } & \C^{\E}\\[-.4ex]
{\scriptstyle 0} && {\scriptstyle 1}
\end{array}
\Bigr]\Bigr),
\]
where
\[
d(f,g)(\e)=\wt(\e)\bigl(f(\b)-g(\w)\bigr)
\qquad\text{for }\e=\b\w\in \E.
\]
In particular,
\[
\Ext^1_{\widetilde\P^2}(\widetilde{\mathcal L},\widetilde{\mathcal L})^{G}
\cong
\C^{\E}/\operatorname{im}(d).
\]
\end{proposition}

The rest of this subsection is devoted to the proof of Proposition~\ref{prop:repext}. Let $K^\bullet$ denote the two-term complex
\[
K^\bullet := \left[
\begin{array}{c c c}
\cale & \xrightarrow{\ \widetilde \K\ } & \calf\\[-.4ex]
{\scriptstyle -1} & & {\scriptstyle 0}
\end{array}
\right].
\] 
Since $K^\bullet$ is a resolution of $\widetilde{\mathcal L}$ by locally free sheaves, $\Ext$ may be computed as the hypercohomology of the \v{C}ech-$\mathcal Hom$ double complex
\[
\Ext^i(\widetilde{\mathcal L},\widetilde{\mathcal L})^G
\cong
\mathbb H^i\!\bigl(\widetilde \P^2,\mathcal H om^\bullet(K^\bullet,K^\bullet)\bigr)^G.
\]
The computation of this hypercohomology is an especially simple application of a spectral sequence. The relevant background is reviewed in Appendix~\ref{app:ss}.

Recall (cf.~\cite[\S10.1.1]{HL}) that for bounded complexes $E^\bullet,F^\bullet$ of locally free sheaves,
\[
\mathcal Hom^n(E^\bullet,F^\bullet)
=\bigoplus_{i\in\Z}\mathcal Hom(E^i,F^{i+n}),
\]
with differential $(df)^i=d_F^{\,i+n}\circ f^i-(-1)^n f^{i+1}\circ d_E^{\,i}$.
Applying this to $K^\bullet$ gives 
\[ \mathcal Hom^\bullet(K^\bullet,K^\bullet) = \left[\, \begin{array}{c c c c c} \mathcal Hom(\calf,\cale) & \xrightarrow{\ d^{-1}\ } & \mathcal Hom(\cale,\cale)\oplus \mathcal Hom(\calf,\calf) & \xrightarrow{\ d^{0}\ } & \mathcal Hom(\cale,\calf)\\[-.4ex] {\scriptstyle -1} && {\scriptstyle 0} && {\scriptstyle 1} \end{array} \,\right], \]
where
\[
d^{-1}(f)=\bigl(f\circ \widetilde\K(x),\ \widetilde\K(x)\circ f\bigr),\qquad
d^{0}(f,g)=\widetilde\K(x)\circ f - g \circ \widetilde\K(x).
\]
As in Appendix~\ref{appendix:G_sheaves}, we compute \v Cech cohomology using the standard affine cover \(\widetilde{\mathcal U}=(\widetilde U_i)_{i=0,1,2}\) of \(\widetilde{\P}^2\). Then the double complex computing $\Ext$ is
\[
\check C^{\,p}\!\bigl(\widetilde{\mathcal U},\,\mathcal Hom^{q}(K^\bullet,K^\bullet)\bigr)^{G},
\]
with total differential
\[
D = \check d + (-1)^p d,
\]
where $\check d$ is the \v{C}ech differential and $d$ is the differential of the complex $\mathcal Hom^\bullet(K^\bullet,K^\bullet)$.

Since only the rows $q=-1,0,1$ are nonzero, the double complex has the form
\begin{equation}\label{ss diagram}
\begin{tikzcd}[scale cd=0.9]
\check C^0(\mathcal H om(\cale,\calf))^G \arrow[r]
& \check C^1(\mathcal H om(\cale,\calf))^G \arrow[r]
& \check C^2(\mathcal H om(\cale,\calf))^G \\
\check C^0(\mathcal H om(\cale,\cale)\oplus\mathcal H om(\calf,\calf))^G \arrow[u] \arrow[r]
& \check C^1(\mathcal H om(\cale,\cale)\oplus\mathcal H om(\calf,\calf))^G \arrow[u] \arrow[r]
& \check C^2(\mathcal H om(\cale,\cale)\oplus\mathcal H om(\calf,\calf))^G \arrow[u] \\
\check C^0(\mathcal H om(\calf,\cale))^G \arrow[u] \arrow[r]
& \check C^1(\mathcal H om(\calf,\cale))^G \arrow[u] \arrow[r]
& \check C^2(\mathcal H om(\calf,\cale))^G \arrow[u].
\end{tikzcd}
\end{equation}

\begin{lemma}\label{lem:Ext1}
There is an identification
\begin{equation}\label{eq:complexext-tilde}
\Ext^\bullet_{\widetilde\P^2}(\widetilde{\mathcal L},\widetilde{\mathcal L})^{G}
\cong
H^\bullet\Bigl(\Bigl[
\begin{array}{c c c}
H^0(\mathcal H om(\cale,\cale)\oplus \mathcal H om(\calf,\calf))^G
& \xrightarrow{\ d^0\ } &
H^0(\mathcal H om(\cale,\calf))^G\\[-.4ex]
{\scriptstyle 0} & & {\scriptstyle 1}
\end{array}
\Bigr]\Bigr).
\end{equation}
\end{lemma}

\begin{proof}
Taking cohomology first in the \v{C}ech direction, 
\begin{equation}\label{e1h}
\begin{tikzcd}[scale cd=0.9]
H^0(\mathcal H om(\cale,\calf))^G \arrow[r]
& H^1(\mathcal H om(\cale,\calf))^G \arrow[r]
& H^2(\mathcal H om(\cale,\calf))^G \\
H^0(\mathcal H om(\cale,\cale)\oplus\mathcal H om(\calf,\calf))^G \arrow[u] \arrow[r]
& H^1(\mathcal H om(\cale,\cale)\oplus\mathcal H om(\calf,\calf))^G \arrow[u] \arrow[r]
& H^2(\mathcal H om(\cale,\cale)\oplus\mathcal H om(\calf,\calf))^G \arrow[u] \\
H^0(\mathcal H om(\calf,\cale))^G \arrow[u] \arrow[r]
& H^1(\mathcal H om(\calf,\cale))^G \arrow[u] \arrow[r]
& H^2(\mathcal H om(\calf,\cale))^G \arrow[u].
\end{tikzcd}
\end{equation}
Since line bundles on $\widetilde\P^2$ have no $H^1$, the $p=1$ column of~\eqref{e1h} vanishes. Moreover, after forgetting the $G$-equivariant structure we have
\[
\cale_{\b}\cong \mathcal O_{\widetilde\P^2}(-1)
\qquad\text{and}\qquad
\calf_{\w}\cong \mathcal O_{\widetilde\P^2},
\]
so $\mathcal H om(\calf,\cale)$ is a sum of copies of $\mathcal O_{\widetilde\P^2}(-1)$, while
$\mathcal H om(\cale,\calf)$, $\mathcal H om(\cale,\cale)$, and $\mathcal H om(\calf,\calf)$ are sums of copies of
$\mathcal O_{\widetilde\P^2}(1)$ and $\mathcal O_{\widetilde\P^2}$. Hence
\begin{align*}
H^0(\widetilde\P^2,\mathcal H om(\calf,\cale))^G&=0
\\
H^2(\widetilde\P^2,\mathcal H om(\calf,\cale))^G&=
H^2(\widetilde\P^2,\mathcal H om(\cale,\cale)\oplus\mathcal H om(\calf,\calf))^G
=
H^2(\widetilde\P^2,\mathcal H om(\cale,\calf))^G=0.
\end{align*}
Therefore the only nonzero terms are
\[
H^0(\mathcal H om(\cale,\cale)\oplus \mathcal H om(\calf,\calf))^G,
\qquad
H^0(\mathcal H om(\cale,\calf))^G.
\]
By Proposition~\ref{prop:ttt}, the hypercohomology is computed by the two-term complex in~\eqref{eq:complexext-tilde}.
\end{proof}

\begin{proof}[Proof of Proposition~\ref{prop:repext}]
We now identify the terms appearing in~\eqref{eq:complexext-tilde} with vector spaces associated to the graph. By~\eqref{eq:h0G}, we have
\begin{align*}
\Hom(\cale_{\b},\calf_{\w})^G &\cong
\begin{cases}
\C\{ \widetilde x(\e)\} & \text{if $\e= \b\w\in \E$},\\
0 & \text{otherwise},
\end{cases}\\
\Hom(\cale_{\b},\cale_{\b'})^G &\cong
\begin{cases}
\C \{1\} & \text{if $\b=\b'$},\\
0 & \text{otherwise},
\end{cases}
\\
\Hom(\calf_{\w},\calf_{\w'})^G &\cong
\begin{cases}
\C\{1\} & \text{if $\w=\w'$},\\
0 & \text{otherwise}.
\end{cases}
\end{align*}
Taking direct sums gives canonical identifications
\begin{align}
H^0(\mathcal H om(\cale,\calf))^G
&=
\Hom(\cale,\calf)^G \cong \bigoplus_{\e \in \E} \C\{ \widetilde x(\e)\} \cong  \C^{\E},\nonumber \\
H^0(\mathcal H om(\cale,\cale)\oplus \mathcal H om(\calf,\calf))^G
&=
\Hom(\cale,\cale)^G\oplus \Hom(\calf,\calf)^G \cong \C^{\B}\{1\}\oplus \C^{\W}\{1\} \cong \C^\B \oplus \C^\W. \label{eq:concrete_vs}
\end{align}
Under the identifications~\eqref{eq:concrete_vs}, the complex~\eqref{eq:complexext-tilde} becomes
\[
\Bigl[
\begin{array}{c c c}
\C^{\B}\oplus \C^{\W} & \xrightarrow{\ d \ } & \C^{\E}\\[-.4ex]
{\scriptstyle 0} && {\scriptstyle 1}
\end{array}
\Bigr].
\]
Moreover, if $\e=\b\w\in \E$, then the $\e$-component of
\[
d^0(f,g)=\widetilde\K\circ f-g\circ \widetilde\K
\]
is
\[
\wt(\e)\bigl(f(\b)-g(\w)\bigr).
\]
This proves Proposition~\ref{prop:repext}.
\end{proof}

\subsection{Two models for the cotangent space}

We now turn to the cotangent space on the sheaf side. As on the graph side, there are two useful models.

The first model is
\[
T^*_{(\widetilde{\mathcal L},\bm\nu)}\widetilde{\mathcal M}
\cong
\Ext^1(\widetilde{\mathcal L},\widetilde{\mathcal L}(-\widetilde D))^G.
\]
This identification comes from Serre duality, which identifies
\[
\Ext^1(\widetilde{\mathcal L},\widetilde{\mathcal L}(-\widetilde D))^G
\cong
\Ext^1(\widetilde{\mathcal L},\widetilde{\mathcal L})^{G,*}
=
T^*_{(\widetilde{\mathcal L},\bm\nu)}\widetilde{\mathcal M}.
\]
The second model is obtained from the restriction map
\[
r:\widetilde{\mathcal L}\to \widetilde{\mathcal L}|_{\widetilde D},
\]
which induces an isomorphism
\[
\Psi:
\Ext^1\!\bigl(\widetilde{\mathcal L},
[\widetilde{\mathcal L}\to \widetilde{\mathcal L}|_{\widetilde D}]\bigr)^G
\xrightarrow{\ \cong\ }
\Ext^1(\widetilde{\mathcal L},\widetilde{\mathcal L}(-\widetilde D))^G.
\]
Thus the cotangent space admits the two descriptions
\[
T^*_{(\widetilde{\mathcal L},\bm\nu)}\widetilde{\mathcal M}
\cong
\Ext^1(\widetilde{\mathcal L},\widetilde{\mathcal L}(-\widetilde D))^G
\cong
\Ext^1\!\bigl(\widetilde{\mathcal L},
[\widetilde{\mathcal L}\to \widetilde{\mathcal L}|_{\widetilde D}]\bigr)^G.
\]
The next three subsections compute these models and the Serre duality pairing concretely.

\subsection{The first model for the cotangent space}\label{sec:ext2}

The goal of this subsection is to prove the following explicit description of
\(
\Ext^\bullet(\widetilde{\mathcal L},\widetilde{\mathcal L}(-\widetilde D))^G.
\)

\begin{proposition}\label{prop:repext2}
There is an identification
\begin{equation}\label{eq:chain_ext-d}
\Ext^\bullet(\widetilde{\mathcal L},\widetilde{\mathcal L}(-\widetilde D))^G
\cong
H^\bullet\Bigl(\Bigl[
\begin{array}{c c c}
\C^{\E} & \xrightarrow{\ d^1\ } & \C^{\B}\oplus \C^{\W}\\[-.4ex]
{\scriptstyle 1} && {\scriptstyle 2}
\end{array}
\Bigr]\Bigr),
\end{equation}
where
\[
d^1(v)
=
\left(
\left(\sum_{\w:\b\w\in\E}\wt(\b\w)v(\b\w)\right)_{\b\in\B},
\left(\sum_{\b:\b\w\in\E}\wt(\b\w)v(\b\w)\right)_{\w\in\W}
\right).
\]
In particular,
\[
\Ext^1(\widetilde{\mathcal L},\widetilde{\mathcal L}(-\widetilde D))^G
\cong
\ker(d^1), \qquad
\Ext^2(\widetilde{\mathcal L},\widetilde{\mathcal L}(-\widetilde D))^G
\cong
(\C^{\B}\oplus \C^{\W})/\operatorname{im}(d^1).
\]
\end{proposition}

The rest of this subsection is devoted to the proof of Proposition~\ref{prop:repext2}. We compute $\Ext$ as the hypercohomology
\[
\Ext^i(\widetilde{\mathcal L},\widetilde{\mathcal L}(-\widetilde D))^G
\cong
\mathbb H^i\!\bigl(\widetilde\P^2,\mathcal H om^\bullet(K^\bullet,K^\bullet(-\widetilde D))\bigr)^G,
\]
where $\mathcal H om^\bullet(K^\bullet,K^\bullet(-\widetilde D))$ is
\[
\left[\,
\begin{array}{c c c c c}
\mathcal H om(\calf,\cale(-\widetilde D))
&\xrightarrow{\ d^{-1}\ }&
\mathcal H om(\cale,\cale(-\widetilde D))\oplus \mathcal H om(\calf,\calf(-\widetilde D))
&\xrightarrow{\ d^0\ }&
\mathcal H om(\cale,\calf(-\widetilde D))\\[-.4ex]
{\scriptstyle -1} && {\scriptstyle 0} && {\scriptstyle 1}
\end{array}
\,\right],
\]
with
\[
d^{-1}(u)=\bigl(u\circ \widetilde\K,\ \widetilde\K\circ u\bigr),
\qquad
d^0(a,b)=\widetilde\K\circ a-b\circ \widetilde\K.
\]

The double complex whose hypercohomology we need to compute is
\[
\check C^{\,p}\!\bigl(\widetilde{\mathcal U},\mathcal H om^{q}(K^\bullet,K^\bullet(-\widetilde D))\bigr)^G.
\]
The nonzero portion is
\begin{equation}\label{eq:e0-trace}
\begin{tikzcd}[scale cd=0.7]
\check C^0(\mathcal H om(\cale,\calf(-\widetilde D)))^G \arrow[r] &
\check C^1(\mathcal H om(\cale,\calf(-\widetilde D)))^G \arrow[r] &
\check C^2(\mathcal H om(\cale,\calf(-\widetilde D)))^G \\
\check C^0(\mathcal H om(\cale,\cale(-\widetilde D))\oplus \mathcal H om(\calf,\calf(-\widetilde D)))^G
\arrow[u] \arrow[r] &
\check C^1(\mathcal H om(\cale,\cale(-\widetilde D))\oplus \mathcal H om(\calf,\calf(-\widetilde D)))^G
\arrow[u] \arrow[r] &
\check C^2(\mathcal H om(\cale,\cale(-\widetilde D))\oplus \mathcal H om(\calf,\calf(-\widetilde D)))^G
\arrow[u]\\
\check C^0(\mathcal H om(\calf,\cale(-\widetilde D)))^G \arrow[r] \arrow[u] &
\check C^1(\mathcal H om(\calf,\cale(-\widetilde D)))^G \arrow[r] \arrow[u] &
\check C^2(\mathcal H om(\calf,\cale(-\widetilde D)))^G \arrow[u]
\end{tikzcd}.
\end{equation}

\begin{lemma}\label{lem:Ext2}
There is an identification
\begin{equation}\label{eq:Ext2-concrete}
\begin{aligned}
\Ext^\bullet(\widetilde{\mathcal L},\widetilde{\mathcal L}(-\widetilde D))^G
= H^\bullet\Bigl(\Bigl[ &
\underset{1}{H^2(\mathcal H om(\calf,\cale(-\widetilde D)))^G} \\
&\xrightarrow{\,d^{-1}\,}
\underset{2}{H^2(\mathcal H om(\cale,\cale(-\widetilde D))\oplus \mathcal H om(\calf,\calf(-\widetilde D)))^G}
\Bigr]\Bigr).
\end{aligned}
\end{equation}
\end{lemma}

\begin{proof}
Taking cohomology first in the \v Cech direction, all groups vanish except the two \(H^2\)-terms shown above. Hence we obtain
\[
\begin{tikzcd}
0 \arrow[r] & 0 \arrow[r] & 0 \\
0 \arrow[u] \arrow[r] & 0 \arrow[u] \arrow[r] &
H^2(\mathcal H om(\cale,\cale(-\widetilde D))\oplus \mathcal H om(\calf,\calf(-\widetilde D)))^G \arrow[u]\\
0 \arrow[u] \arrow[r] & 0 \arrow[u] \arrow[r] &
H^2(\mathcal H om(\calf,\cale(-\widetilde D)))^G \arrow[u]
\end{tikzcd}.
\]
Proposition~\ref{prop:ttt} gives the claimed two-term complex.
\end{proof}

We now identify the two cohomology groups appearing in Lemma~\ref{lem:Ext2} with vector spaces associated to the graph. By~\eqref{eq:h2G}, we have
\[
H^2(\mathcal H om(\cale_\b,\cale_{\b'}(-\widetilde D)))^G
\cong
\begin{cases}
\frac{1}{\widetilde x_0 \widetilde x_1 \widetilde x_2}\C\{1\} & \text{if $\b=\b'$},\\
0 & \text{otherwise},
\end{cases}
\]
and similarly
\[
H^2(\mathcal H om(\calf_\w,\calf_{\w'}(-\widetilde D)))^G
\cong
\begin{cases}
\frac{1}{\widetilde x_0 \widetilde x_1 \widetilde x_2}\C\{1\} & \text{if $\w=\w'$},\\
0 & \text{otherwise},
\end{cases}
\]
while
\[
H^2(\mathcal H om(\calf_\w,\cale_\b(-\widetilde D)))^G
\cong
\begin{cases}
\frac{1}{\widetilde x_0 \widetilde x_1 \widetilde x_2}\C\!\left\{\frac{1}{\widetilde x(\e)}\right\} & \text{if $\e=\b\w\in \E$},\\
0 & \text{otherwise}.
\end{cases}
\]
Taking direct sums, we obtain canonical identifications
\begin{equation}\label{eq:concrete-vs-trace}
H^2(\mathcal H om(\calf,\cale(-\widetilde D)))^G \cong \C^{\E},
\qquad
H^2(\mathcal H om(\cale,\cale(-\widetilde D))\oplus \mathcal H om(\calf,\calf(-\widetilde D)))^G
\cong \C^{\B}\oplus \C^{\W}.
\end{equation}

Under these identifications, a class in
\[
H^2(\mathcal H om(\cale,\cale(-\widetilde D))\oplus \mathcal H om(\calf,\calf(-\widetilde D)))^G
\]
is represented by a \(2\)-\v Cech cocycle of the form
\[
\frac{1}{\widetilde x_0 \widetilde x_1 \widetilde x_2}(f(\b),g(\w))_{\b\in\B,\ \w\in\W},
\qquad
f\in \C^\B,\ g\in \C^\W,
\]
and a class in
\[
H^2(\mathcal H om(\calf,\cale(-\widetilde D)))^G
\]
is represented by a cocycle of the form
\[
\frac{1}{\widetilde x_0 \widetilde x_1 \widetilde x_2}
\left(
\frac{v(\e)}{\widetilde x(\e)}
\right)_{\e\in\E},
\qquad
v\in \C^\E.
\]

\begin{lemma}\label{lem:diff-ext2}
Under the identifications~\eqref{eq:concrete-vs-trace}, the differential in Lemma~\ref{lem:Ext2} becomes the map
\[
d^1:\C^\E\rightarrow \C^\B\oplus \C^\W
\]
given by
\[
d^1(v)
=
\left(
\left(\sum_{\w:\b\w\in\E}\wt(\b\w)v(\b\w)\right)_{\b\in\B},
\left(\sum_{\b:\b\w\in\E}\wt(\b\w)v(\b\w)\right)_{\w\in\W}
\right).
\]
\end{lemma}

\begin{proof}
The differential is induced by the map
\[
d^{-1}:
\mathcal H om(\calf,\cale(-\widetilde D))
\rightarrow
\mathcal H om(\cale,\cale(-\widetilde D))
\oplus
\mathcal H om(\calf,\calf(-\widetilde D)),
\qquad
u\mapsto \bigl(u\circ \widetilde\K,\ \widetilde\K\circ u\bigr).
\]
Let
\[
u=\frac{1}{\widetilde x_0 \widetilde x_1 \widetilde x_2}
\left(
\frac{v(\e)}{\widetilde x(\e)}
\right)_{\e\in\E}.
\]
Then the diagonal \((\b,\b)\)-entry of \(u\circ \widetilde\K\) is
\[
\sum_{\w:\b\w\in\E}
\frac{\wt(\b\w)v(\b\w)}{\widetilde x_0 \widetilde x_1 \widetilde x_2},
\]
while the diagonal \((\w,\w)\)-entry of \(\widetilde\K\circ u\) is
\[
\sum_{\b:\b\w\in\E}
\frac{\wt(\b\w)v(\b\w)}{\widetilde x_0 \widetilde x_1 \widetilde x_2}.
\]
These are exactly the two components of \(d^1(v)\).
\end{proof}

Combining Lemma~\ref{lem:Ext2}, the identifications~\eqref{eq:concrete-vs-trace}, and Lemma~\ref{lem:diff-ext2}, we obtain
\[
\Ext^\bullet(\widetilde{\mathcal L},\widetilde{\mathcal L}(-\widetilde D))^G
\cong
H^\bullet\Bigl(\Bigl[
\begin{array}{c c c}
\C^{\E} & \xrightarrow{\ d^1\ } & \C^{\B}\oplus \C^{\W}\\[-.4ex]
{\scriptstyle 1} && {\scriptstyle 2}
\end{array}
\Bigr]\Bigr),
\]
with \(d^1\) as above. This proves Proposition~\ref{prop:repext2}.

\subsection{Serre duality}\la{sec:duality}

The goal of this subsection is to compute the Serre duality pairing
\[
\Ext^1(\widetilde{\mathcal L},\widetilde{\mathcal L})^G
\otimes
\Ext^{1}\!\bigl(\widetilde{\mathcal L},\widetilde{\mathcal L}(-\widetilde D)\bigr)^G
\xrightarrow{\ \cup \ }
\Ext^{2}\!\bigl(\widetilde{\mathcal L},\widetilde{\mathcal L}(-\widetilde D)\bigr)^G
\xrightarrow{\ \mathrm{tr}\ }
H^2(\widetilde\P^2,\mathcal O_{\widetilde\P^2}(-\widetilde D))^G
\xrightarrow[\cong]{\widetilde \theta}
\C.
\]
We first state the result.

\begin{proposition}\la{prop:gsdual}
With the identifications in Proposition~\ref{prop:repext} and Proposition~\ref{prop:repext2}, the Serre duality pairing
\[
\Ext^1(\widetilde{\mathcal L},\widetilde{\mathcal L})^G
\otimes
\Ext^{1}\!\bigl(\widetilde{\mathcal L},\widetilde{\mathcal L}(-\widetilde D)\bigr)^G
\rightarrow \C
\]
is induced by the map
\begin{align}
\C^\E \otimes \C^\E &\rightarrow \C \nonumber\\
u\otimes v &\mapsto -\sum_{\e\in\E}u(\e)v(\e).
\la{eq:gsd}
\end{align}
\end{proposition}

\subsubsection{The cup product}

In this subsubsection we compute the cup product using \v{C}ech-$\mathcal Hom$ double complex.

\begin{lemma}\la{lem::cup1}
With the identifications in Proposition~\ref{prop:repext} and Proposition~\ref{prop:repext2}, the cup product map
\[
\cup:
\Ext^1(\widetilde{\mathcal L},\widetilde{\mathcal L})^G \otimes
\Ext^1(\widetilde{\mathcal L},\widetilde{\mathcal L}(-\widetilde D))^G
\rightarrow
\Ext^2(\widetilde{\mathcal L},\widetilde{\mathcal L}(-\widetilde D))^G
\]
is induced by the map
\[
\cup: \C^\E\otimes \C^\E \rightarrow \C^\B\oplus \C^\W
\]
given by
\[
u \cup v
=
\left(
\left(\sum_{\w:\b\w\in \E} u(\b\w)v(\b\w)\right)_{\b\in \B},
0
\right).
\]
\end{lemma}

\begin{proof}
By definition, the cup product is induced by the map
\[
\cup :
\check C^{p}\bigl(\widetilde{\mathcal U},\mathcal H om(K^\bullet,K^\bullet)^q\bigr)^G
\otimes
\check C^{p'}\bigl(\widetilde{\mathcal U},\mathcal H om(K^\bullet,K^\bullet(-\widetilde D))^{q'}\bigr)^G
\rightarrow
\check C^{p+p'}\bigl(\widetilde{\mathcal U},\mathcal H om(K^\bullet,K^\bullet(-\widetilde D))^{q+q'}\bigr)^G
\]
given by
\[
(\alpha\cup\beta)_{i_0\cdots i_{p+p'}}
=
(-1)^{q p'}
\beta_{i_p\cdots i_{p+p'}}\big|_{\widetilde U_{i_0\cdots i_{p+p'}}}
\circ
\alpha_{i_0\cdots i_p}\big|_{\widetilde U_{i_0\cdots i_{p+p'}}}.
\]

By Proposition~\ref{prop:repext}, an element \(u\in \C^\E\) determines an element of
\[
H^0(\mathcal H om(\cale,\calf))^G,
\]
represented by the \v Cech \(0\)-cocycle
\[
l=\bigl(u(\e)\,\widetilde x(\e)\bigr)_{\e\in\E}
\in \check C^0(\widetilde{\mathcal U},\mathcal H om(\cale,\calf))^G,
\]
where each \(u(\e)\,\widetilde x(\e)\) is viewed as a global section, hence as a \v Cech \(0\)-cocycle with identical components on the three open sets of \(\widetilde{\mathcal U}\). This class is represented by the $D$-cocycle
\[
(l,0,0)
\in
\check C^0(\widetilde{\mathcal U},\mathcal H om(\cale,\calf))^G
\oplus
\check C^1(\widetilde{\mathcal U},\mathcal H om(\cale,\cale)\oplus \mathcal H om(\calf,\calf))^G
\oplus
\check C^2(\widetilde{\mathcal U},\mathcal H om(\calf,\cale))^G,
\]
where \(\check d(l)=0\).

Similarly, let \(v\in \ker(d^1)\subset \C^\E\). A class in
\[
\Ext^1(\widetilde{\mathcal L},\widetilde{\mathcal L}(-\widetilde D))^G
\]
may be represented by a $D$-cocycle $(l',m',n')$ in
\[
\check C^0(\widetilde{\mathcal U},\mathcal H om(\cale,\calf(-\widetilde D)))^G
\oplus
\check C^1\!\bigl(\widetilde{\mathcal U},
\mathcal H om(\cale,\cale(-\widetilde D))
\oplus
\mathcal H om(\calf,\calf(-\widetilde D))\bigr)^G
\oplus
\check C^2(\widetilde{\mathcal U},\mathcal H om(\calf,\cale(-\widetilde D)))^G,
\]
satisfying
\begin{align*}
\check d(l')-d^0(m')=0,\qquad
\check d(m')+d^{-1}(n')=0,\qquad
\check d(n')=0.
\end{align*}
Here the \(\check C^2\)-component is
\[
n'
=
\frac{1}{\widetilde x_0 \widetilde x_1 \widetilde x_2}
\left( \frac{v(\e)}{\widetilde x(\e)} \right)_{\e\in\E}.
\]
The nonzero part of the cup product is
\[
(r,s)=(l,0,0)\cup(l',m',n'),
\]
where
\[
(r,s)\in
\check C^1(\widetilde{\mathcal U},\mathcal H om(\cale,\calf(-\widetilde D)))^G
\oplus
\check C^2\!\bigl(\widetilde{\mathcal U},
\mathcal H om(\cale,\cale(-\widetilde D))
\oplus
\mathcal H om(\calf,\calf(-\widetilde D))\bigr)^G
\]
is explicitly given by
\begin{align*}
r_{ij}
&=
-m_{ij}^{\prime,\calf}\circ l_i
\in
\Hom(\cale|_{\widetilde U_{ij}},\calf(-\widetilde D)|_{\widetilde U_{ij}})^G,\\[0.5ex]
s_{ijk}
&=
\bigl(n'_{ijk}\circ l_i,\ 0\bigr)
\in
\Hom(\cale|_{\widetilde U_{ijk}},\cale(-\widetilde D)|_{\widetilde U_{ijk}})^G
\oplus
\Hom(\calf|_{\widetilde U_{ijk}},\calf(-\widetilde D)|_{\widetilde U_{ijk}})^G.
\end{align*}

Under the identification of Lemma~\ref{lem:Ext2}, the class of the total cocycle \((r,s)\) is sent precisely to the class of \(s\) in
\[
H^2(\mathcal H om(\cale,\cale(-\widetilde D))\oplus \mathcal H om(\calf,\calf(-\widetilde D)))^G
\Big/
\operatorname{im}\bigl(d^{-1}\bigr).
\]
It remains to identify this class explicitly. The \((\b,\b)\)-entry of \(n'\circ l\) is
\[
\sum_{\w:\b\w\in \E}
\frac{u(\b\w)v(\b\w)}{\widetilde x_0\widetilde x_1\widetilde x_2}.
\]
Hence under the identification
\[
H^2(\mathcal H om(\cale,\cale(-\widetilde D))\oplus \mathcal H om(\calf,\calf(-\widetilde D)))^G
\cong \C^\B\oplus \C^\W,
\]
the class of \(s\) is
\[
\left(
\left(\sum_{\w:\b\w\in \E}u(\b\w)v(\b\w)\right)_{\b\in \B},
0
\right).
\]
\end{proof}

\subsubsection{The trace map}\label{sec:trace}

The goal of this subsubsection is to compute the trace map
\[
\widetilde \theta\circ\tr:
\Ext^2(\widetilde{\mathcal L},\widetilde{\mathcal L}(-\widetilde D))^G
\rightarrow
H^2(\widetilde\P^2,\mathcal O_{\widetilde\P^2}(-\widetilde D))^G
\xrightarrow[\cong]{\widetilde \theta}
\C.
\]

\begin{proposition}\label{prop:trace-ext2}
Under the identification of Proposition~\ref{prop:repext2}, the trace map
\[
\widetilde \theta\circ\tr:
\Ext^2(\widetilde{\mathcal L},\widetilde{\mathcal L}(-\widetilde D))^G
\rightarrow
H^2(\widetilde\P^2,\mathcal O_{\widetilde\P^2}(-\widetilde D))^G
\xrightarrow[\cong]{\widetilde \theta}
\C
\]
is induced by
\[
\C^{\B}\oplus \C^{\W}\rightarrow \C,
\qquad
(f,g)\mapsto \sum_{\w\in\W}g(\w)-\sum_{\b\in\B}f(\b).
\]
\end{proposition}

\begin{proof}
Consider the trace morphism of complexes
\[
\tr:\mathcal H om^\bullet(K^\bullet,K^\bullet(-\widetilde D)) \rightarrow
\left[
\begin{array}{c}
\mathcal O_{\widetilde\P^2}(-\widetilde D)\\[-.4ex]
{\scriptstyle 0}
\end{array}
\right],
\]
defined by
\[
\restr{\tr}{\mathcal H om(K^i,K^j(-\widetilde D))}
=
\begin{cases}
(-1)^i\tr_{K^i} & \text{if $i=j$},\\
0 & \text{otherwise}.
\end{cases}
\]
We have
\[
H^2(\widetilde\P^2,\mathcal O_{\widetilde\P^2}(-\widetilde D))^G
=
\frac{1}{\widetilde x_0 \widetilde x_1 \widetilde x_2}\C\{1\},
\qquad
\widetilde \theta \left(\frac{1}{\widetilde x_0 \widetilde x_1 \widetilde x_2}\right)=1.
\]
Hence \(\widetilde\theta\circ\tr\) is determined by its value on a cocycle
\[
\frac{1}{\widetilde x_0 \widetilde x_1 \widetilde x_2}(f(\b),g(\w))_{\b\in\B,\ \w\in\W}\in H^2(\mathcal H om(\cale,\cale(-\widetilde D))\oplus \mathcal H om(\calf,\calf(-\widetilde D)))^G.
\]
Since the trace is the alternating sum of the diagonal traces, we obtain
\[
(\widetilde\theta\circ\tr)(f,g)
=
\sum_{\w\in\W}g(\w)-\sum_{\b\in\B}f(\b).
\]
\end{proof}

\begin{proof}[Proof of Proposition~\ref{prop:gsdual}]
By Lemma~\ref{lem::cup1}, the cup product
\[
\Ext^1(\widetilde{\mathcal L},\widetilde{\mathcal L})^G \otimes
\Ext^1(\widetilde{\mathcal L},\widetilde{\mathcal L}(-\widetilde D))^G
\rightarrow
\Ext^2(\widetilde{\mathcal L},\widetilde{\mathcal L}(-\widetilde D))^G
\]
is induced by
\[
u\otimes v\mapsto
\left(
\left(\sum_{\w:\b\w\in \E}u(\b\w)v(\b\w)\right)_{\b\in \B},
0
\right).
\]
Applying Proposition~\ref{prop:trace-ext2}, we obtain
\begin{align*}
(\widetilde\theta\circ\tr)(u\cup v)
&=
-\sum_{\b\in\B}\sum_{\w:\b\w\in \E}u(\b\w)v(\b\w)\\
&=
-\sum_{\e\in\E}u(\e)v(\e).
\end{align*}
Thus the Serre duality pairing is induced by the map~\eqref{eq:gsd}.
\end{proof}

\subsection{The second model for the cotangent space}\label{sec:ext-ld}

In this section, we compute $\Ext^\bullet\!\bigl(\widetilde{\mathcal L},[\widetilde{\mathcal L}\to \widetilde{\mathcal L}|_{\widetilde D}]\bigr)^G$ by the same \v{C}ech-$\mathcal Hom$ double-complex method used in the previous subsection.

For each zig-zag path \(\alpha\in\zz\) and each black vertex \(\b\in \B\cap\alpha\), recall that
\[
\w_\alpha^{-}(\b)\;-\;\b\;-\;\w_\alpha^{+}(\b)
\]
denotes the wedge of \(\alpha\) at \(\b\). 

\begin{proposition}\label{prop:repextld}
There is an identification
\begin{equation}\label{eq:Fbullet-concrete}
\begin{aligned}
\Ext^\bullet\!\bigl(\widetilde{\mathcal L},[\widetilde{\mathcal L}\to \widetilde{\mathcal L}|_{\widetilde D}]\bigr)^G
\cong
H^\bullet\Bigl(\Bigl[
&
\underset{0}{\C^{\B}\oplus \C^{\W}}
\xrightarrow{\,d^{0}\,}
\underset{1}{\C^{\E}\oplus \displaystyle\bigoplus_{\alpha\in\zz}\Hom(\C^{\W\cap\alpha},\C)}\\
&\xrightarrow{\,d^{1}\,}
\underset{2}{\displaystyle\bigoplus_{\alpha\in\zz}\Hom(\C^{\B\cap\alpha},\C)}
\Bigr]\Bigr).
\end{aligned}
\end{equation}
The differentials are given by
\begin{align*}
d^{0}(f,g)
&=
\left(
\bigl(\wt(\e)(f(\b)-g(\w))\bigr)_{\e=\b\w\in\E},
\ (g|_{\W\cap\alpha})_{\alpha\in\zz}
\right),\\
d^{1}(u,t)_\alpha
&=
\left(
\frac{u(\b \w_\alpha^{-}(\b))}{\wt(\b \w_\alpha^{-}(\b))}
-
\frac{u(\b \w_\alpha^{+}(\b))}{\wt(\b \w_\alpha^{+}(\b))}
+
t_\alpha(\w_\alpha^{-}(\b))
-
t_\alpha(\w_\alpha^{+}(\b))
\right)_{\b\in \B\cap\alpha},
\end{align*}
for each \(\alpha\in\zz\), where \(t=(t_\alpha)_{\alpha\in\zz}\). 
\end{proposition}

The rest of this subsection is devoted to the proof of Proposition~\ref{prop:repextld}. Let
\[
K_{\widetilde D}^\bullet
:=
\Bigl[
\begin{array}{c c c c c}
\cale & \xrightarrow{\ \widetilde \K\ } & \calf & \to & \widetilde{\mathcal L}|_{\widetilde D}\\[-.4ex]
{\scriptstyle -1} && {\scriptstyle 0} && {\scriptstyle 1}
\end{array}
\Bigr]
\]
be the three-term complex obtained by composing \(\calf\to \widetilde{\mathcal L}\) with the restriction map
\[
r:\widetilde{\mathcal L}\rightarrow \widetilde{\mathcal L}|_{\widetilde D}.
\]
Then
\[
\Ext^i\!\bigl(\widetilde{\mathcal L},[\widetilde{\mathcal L}\to \widetilde{\mathcal L}|_{\widetilde D}]\bigr)^G
\cong
\mathbb H^i\!\bigl(\widetilde \P^2,\mathcal H om^\bullet(K^\bullet,K_{\widetilde D}^\bullet)\bigr)^G,
\]
where
\[
\mathcal H om^\bullet(K^\bullet,K_{\widetilde D}^\bullet)
=
\begin{aligned}
\Bigl[
\underset{-1}{\mathcal Hom(\calf,\cale)}
&\xrightarrow{\,d^{-1}\,}
\underset{0}{\mathcal Hom(\cale,\cale)\oplus\mathcal Hom(\calf,\calf)} \\
&\xrightarrow{\,d^{0}\,}
\underset{1}{\mathcal Hom(\cale,\calf)\oplus\mathcal Hom(\calf,\widetilde{\mathcal L}|_{\widetilde D})}
\xrightarrow{\,d^{1}\,}
\underset{2}{\mathcal Hom(\cale,\widetilde{\mathcal L}|_{\widetilde D})}
\Bigr].
\end{aligned}
\]
and, writing \(m=(m^\cale,m^\calf)\),
\begin{align}
d^{-1}(f)&=\bigl(f\circ \widetilde \K,\ \widetilde \K \circ f\bigr), \nonumber \\
d^{0}(m^\cale,m^\calf)&=\bigl(\widetilde \K \circ m^\cale - m^\calf\circ \widetilde \K,\ r\circ m^\calf\bigr),\nonumber \\
d^{1}(u,t)&=r\circ u + t\circ \widetilde \K.
\label{eq:diff-ld}
\end{align}

\begin{lemma}\label{lem:Fbullet}
There is an identification
\begin{equation}\label{eq:Fbullet}
\begin{aligned}
\Ext^\bullet\!\bigl(\widetilde{\mathcal L},[\widetilde{\mathcal L}\to \widetilde{\mathcal L}|_{\widetilde D}]\bigr)^G
\cong
H^\bullet\Bigl(\Bigl[
&
\underset{0}{\Hom(\cale,\cale)^G\oplus\Hom(\calf,\calf)^G}\\
&\xrightarrow{\,d^{0}\,}
\underset{1}{\Hom(\cale,\calf)^G\oplus\Hom(\calf,\widetilde{\mathcal L}|_{\widetilde D})^G}
\xrightarrow{\,d^{1}\,}
\underset{2}{\Hom(\cale,\widetilde{\mathcal L}|_{\widetilde D})^G}
\Bigr]\Bigr),
\end{aligned}
\end{equation}
where the differentials are induced by~\eqref{eq:diff-ld}.
\end{lemma}

\begin{proof}
Since \(\mathcal H om^q(K^\bullet,K_{\widetilde D}^\bullet)=0\) unless \(q\in\{-1,0,1,2\}\), the
\v{C}ech-$\mathcal Hom$ double complex has four nonzero rows:
\begin{equation}\label{eq:e0-ld}
\begin{tikzcd}[scale cd=0.85]
\check C^0(\mathcal H om(\cale,\widetilde{\mathcal L}|_{\widetilde D}))^G \arrow[r] &
\check C^1(\mathcal H om(\cale,\widetilde{\mathcal L}|_{\widetilde D}))^G\\
\check C^0(\mathcal H om(\cale,\calf)\oplus \mathcal H om(\calf,\widetilde{\mathcal L}|_{\widetilde D}))^G \arrow[u] \arrow[r] &
\check C^1(\mathcal H om(\cale,\calf)\oplus \mathcal H om(\calf,\widetilde{\mathcal L}|_{\widetilde D}))^G \arrow[u]\\
\check C^0(\mathcal H om(\cale,\cale)\oplus \mathcal H om(\calf,\calf))^G \arrow[u] \arrow[r] &
\check C^1(\mathcal H om(\cale,\cale)\oplus \mathcal H om(\calf,\calf))^G \arrow[u]\\
\check C^0(\mathcal H om(\calf,\cale))^G \arrow[u] \arrow[r] &
\check C^1(\mathcal H om(\calf,\cale))^G \arrow[u].
\end{tikzcd}
\end{equation}
Taking cohomology first in the \v{C}ech direction and using the same vanishing statements as in Lemma~\ref{lem:Ext1}, the only surviving terms occur in \v{C}ech degree \(p=0\). It follows by Proposition~\ref{prop:ttt} that the hypercohomology is computed by the cohomology of the complex~\eqref{eq:Fbullet}.
\end{proof}

We now describe the two spaces
\[
\Hom(\cale,\widetilde{\mathcal L}|_{\widetilde D})^G,
\qquad
\Hom(\calf,\widetilde{\mathcal L}|_{\widetilde D})^G
\]
in concrete terms. For \(\alpha\in\zz\), let
\[
\pi^{-1}(\alpha)=\{\widetilde\alpha_1,\dots,\widetilde\alpha_d\}\subset \widetilde D
\]
be the set of points of \(\widetilde D\) lying over \(\alpha\). Since \(\widetilde{\mathcal L}|_{\widetilde D}\) is a collection of one dimensional spaces supported on the finite set \( \widetilde C \cap \widetilde D\), we have
\[
\Hom(\cale,\widetilde{\mathcal L}|_{\widetilde D})
\cong
\bigoplus_{\alpha\in\zz}\ \bigoplus_{i=1}^d
\Hom\!\bigl(\restr{\cale}{\widetilde\alpha_i},\restr{\widetilde{\mathcal L}}{\widetilde\alpha_i}\bigr)
,\qquad 
\Hom(\calf,\widetilde{\mathcal L}|_{\widetilde D})
\cong
\bigoplus_{\alpha\in\zz}\ \bigoplus_{i=1}^d
\Hom\!\bigl(\restr{\calf}{\widetilde\alpha_i},\restr{\widetilde{\mathcal L}}{\widetilde\alpha_i}\bigr).
\]
Moreover, at the point \(\widetilde\alpha_i\), only the \(\alpha\)-block of \(\widetilde\K\) contributes. Define
\[
E_{\alpha,i}:=\bigoplus_{\b\in\B\cap\alpha}\restr{\cale_\b}{\widetilde\alpha_i},
\qquad
F_{\alpha,i}:=\bigoplus_{\w\in\W\cap\alpha}\restr{\calf_\w}{\widetilde\alpha_i},
\qquad 
L_{\alpha,i}:=\coker\!\left(\restr{\widetilde\K_\alpha}{\widetilde\alpha_i}:E_{\alpha,i}\to F_{\alpha,i}\right) \cong \restr{\widetilde{\mathcal L}}{\widetilde\alpha_i},
\]
so that
\[
\Hom(\cale,\widetilde{\mathcal L}|_{\widetilde D})
\cong
\bigoplus_{\alpha\in\zz}\ \bigoplus_{i=1}^d \Hom(E_{\alpha,i},L_{\alpha,i}),
\]
\[
\Hom(\calf,\widetilde{\mathcal L}|_{\widetilde D})
\cong
\bigoplus_{\alpha\in\zz}\ \bigoplus_{i=1}^d \Hom(F_{\alpha,i},L_{\alpha,i}).
\]

To obtain explicit coordinates, choose for each \(\alpha\in\zz\) and each \(\widetilde\alpha_i\in\pi^{-1}(\alpha)\) a gauge transformation (i.e. left and right multiplication by diagonal matrices) identifying
\[
\restr{\widetilde\K_\alpha}{\widetilde\alpha_i}
\quad\text{with}\quad
\partial_\alpha:=
\begin{bmatrix}
1 & & & -1\\
-1 & 1 &&\\
& -1 & 1 &\\
&& \ddots & \ddots
\end{bmatrix}.
\]
Such a gauge transformation exists because when viewed as edge weights, $(\restr{\widetilde\K_\alpha}{\widetilde\alpha_i})_{\b,\w}$ define a cohomology class with monodromy $(-1)^d$ around $\alpha$ by~\eqref{eq:casimir}. Since \(\coker(\partial_\alpha)\cong \C\) via the row vector \((1,\dots,1)\), these choices give identifications
\[
E_{\alpha,i}\cong \C^{\B\cap\alpha},
\qquad
F_{\alpha,i}\cong \C^{\W\cap\alpha},
\qquad
L_{\alpha,i}\cong \coker(\partial_\alpha)\cong \C,
\]
Therefore
\[
\Hom(E_{\alpha,i},L_{\alpha,i})
\cong \Hom(\C^{\B\cap\alpha},\C),
\qquad
\Hom(F_{\alpha,i},L_{\alpha,i})
\cong \Hom(\C^{\W\cap\alpha},\C),
\]
and hence
\[
\Hom(\cale,\widetilde{\mathcal L}|_{\widetilde D})
\cong
\bigoplus_{\alpha\in\zz}\ \bigoplus_{i=1}^d \Hom(\C^{\B\cap\alpha},\C),
\]
\[
\Hom(\calf,\widetilde{\mathcal L}|_{\widetilde D})
\cong
\bigoplus_{\alpha\in\zz}\ \bigoplus_{i=1}^d \Hom(\C^{\W\cap\alpha},\C).
\]
These identifications depend on the chosen gauges and are therefore noncanonical.

\begin{lemma}\label{lem:ident}
Assume that the above gauges are chosen \(G\)-equivariantly. Then
\[
\Hom(\cale,\widetilde{\mathcal L}|_{\widetilde D})^G
\cong
\bigoplus_{\alpha\in\zz}\Hom(\C^{\B\cap\alpha},\C),
\qquad
\Hom(\calf,\widetilde{\mathcal L}|_{\widetilde D})^G
\cong
\bigoplus_{\alpha\in\zz}\Hom(\C^{\W\cap\alpha},\C).
\]
\end{lemma}

\begin{proof}
We prove the first isomorphism; the second is identical. Fix \(\alpha\in\zz\). Because the gauges are chosen \(G\)-equivariantly, the induced \(G\)-action on
\[
\bigoplus_{i=1}^d \Hom(\C^{\B\cap\alpha},\C)
\]
is the permutation action on the \(d\) summands corresponding to the \(d\) points
\(\widetilde\alpha_1,\dots,\widetilde\alpha_d\). Therefore the \(G\)-invariant subspace is the diagonal copy:
\[
\left(\bigoplus_{i=1}^d \Hom(\C^{\B\cap\alpha},\C)\right)^G
=
\{(f,\dots,f):f\in \Hom(\C^{\B\cap\alpha},\C)\}
\cong
\Hom(\C^{\B\cap\alpha},\C).
\]
Summing over all \(\alpha\in\zz\) gives the claim.
\end{proof}

We next compute the two maps entering the complex~\eqref{eq:Fbullet}.

\begin{lemma}\label{lem:map1}
Under the identifications
\[
\Hom(\cale,\calf)^G\cong \C^\E,
\qquad
\Hom(\cale,\widetilde{\mathcal L}|_{\widetilde D})^G
\cong
\bigoplus_{\alpha\in\zz}\Hom(\C^{\B\cap\alpha},\C)
\]
from Proposition~\ref{prop:repext} and Lemma~\ref{lem:ident}, the map
\[
\Hom(\cale,\calf)^G\rightarrow \Hom(\cale,\widetilde{\mathcal L}|_{\widetilde D})^G
\]
is given by
\[
u\mapsto \frac{u(\b \w_\alpha^{-}(\b))}{\wt(\b \w_\alpha^{-}(\b))}
-
\frac{u(\b \w_\alpha^{+}(\b))}{\wt(\b \w_\alpha^{+}(\b))}
\qquad (\b\in\B\cap\alpha).
\]
\end{lemma}

\begin{proof}
Fix \(\alpha\in\zz\) and a point \(\widetilde\alpha_i \in \widetilde D\) lying over \(\alpha\). Under the above identifications, the \(\alpha\)-component of
\[
r\circ u:\cale\rightarrow \widetilde{\mathcal L}|_{\widetilde D}
\]
at \(\widetilde\alpha_i\) is the composition
\[
E_{\alpha,i}
\xrightarrow{\ \restr{u}{\widetilde\alpha_i}\ }
F_{\alpha,i}
\rightarrow
L_{\alpha,i}\cong \C.
\]
Under the gauge transformation identifying $\restr{\widetilde\K_\alpha}{\widetilde\alpha_i}$ with $\partial_\alpha$, the matrix representing \(\restr{u}{\widetilde\alpha_i}\) has, in the column indexed by \(\b\in\B\cap\alpha\), exactly two nonzero entries:
\[
\frac{u(\b \w_\alpha^{-}(\b))}{\wt(\b \w_\alpha^{-}(\b))}
\quad\text{in row }\w_\alpha^{-}(\b),
\qquad
-\frac{u(\b \w_\alpha^{+}(\b))}{\wt(\b \w_\alpha^{+}(\b))}
\quad\text{in row }\w_\alpha^{+}(\b).
\]
Multiplying by the row vector \((1,\dots,1)\) shows that the induced functional on \(\C^{\B\cap\alpha}\) is
\[
\b\mapsto
\frac{u(\b \w_\alpha^{-}(\b))}{\wt(\b \w_\alpha^{-}(\b))}
-
\frac{u(\b \w_\alpha^{+}(\b))}{\wt(\b \w_\alpha^{+}(\b))}.
\]
This proves the formula.
\end{proof}

\begin{lemma}\label{lem:map2}
Under the identifications of Lemma~\ref{lem:ident}, the map
\[
\Hom(\calf,\widetilde{\mathcal L}|_{\widetilde D})^G
\rightarrow
\Hom(\cale,\widetilde{\mathcal L}|_{\widetilde D})^G
\]
is given by
\[
t=(t_\alpha)_{\alpha\in\zz}
\mapsto
t_\alpha(\w_\alpha^{-}(\b))-t_\alpha(\w_\alpha^{+}(\b))
\qquad (\b\in\B\cap\alpha).
\]
\end{lemma}

\begin{proof}
Fix \(\alpha\in\zz\). Under the chosen identifications, the map
\[
\Hom(\C^{\W\cap\alpha},\C)\rightarrow \Hom(\C^{\B\cap\alpha},\C)
\]
is induced by composition with the matrix \(\partial_\alpha\). This gives the stated formula.
\end{proof}

We can now finish the proof of Proposition~\ref{prop:repextld}.

\begin{proof}[Proof of Proposition~\ref{prop:repextld}]
By Lemma~\ref{lem:Fbullet}, the Ext groups are computed by the cohomology of the three-term complex
\[
\Hom(\cale,\cale)^G\oplus\Hom(\calf,\calf)^G
\rightarrow
\Hom(\cale,\calf)^G\oplus\Hom(\calf,\widetilde{\mathcal L}|_{\widetilde D})^G
\rightarrow
\Hom(\cale,\widetilde{\mathcal L}|_{\widetilde D})^G.
\]
Using Proposition~\ref{prop:repext}, we identify
\[
\Hom(\cale,\cale)^G\oplus\Hom(\calf,\calf)^G \cong \C^\B\oplus \C^\W,
\qquad
\Hom(\cale,\calf)^G\cong \C^\E.
\]
Using Lemma~\ref{lem:ident}, we identify
\[
\Hom(\calf,\widetilde{\mathcal L}|_{\widetilde D})^G
\cong
\bigoplus_{\alpha\in\zz}\Hom(\C^{\W\cap\alpha},\C),
\qquad
\Hom(\cale,\widetilde{\mathcal L}|_{\widetilde D})^G
\cong
\bigoplus_{\alpha\in\zz}\Hom(\C^{\B\cap\alpha},\C).
\]
Under these identifications, the first differential is induced by
\[
d^{0}(f,g)
=
\left(
\bigl(\wt(\e)(f(\b)-g(\w))\bigr)_{\e=\b\w\in\E},
\ (g|_{\W\cap\alpha})_{\alpha\in\zz}
\right),
\]
and the second differential is the sum of the two maps computed in Lemma~\ref{lem:map1} and Lemma~\ref{lem:map2}. Thus, for \(u\in\C^\E\) and \(t=(t_\alpha)_{\alpha\in\zz}\),
\[
d^{1}(u,t)_\alpha
=
\left(
\frac{u(\b \w_\alpha^{-}(\b))}{\wt(\b \w_\alpha^{-}(\b))}
-
\frac{u(\b \w_\alpha^{+}(\b))}{\wt(\b \w_\alpha^{+}(\b))}
+
t_\alpha(\w_\alpha^{-}(\b))
-
t_\alpha(\w_\alpha^{+}(\b))
\right)_{\b\in \B\cap\alpha}.
\]
This is exactly the complex~\eqref{eq:Fbullet-concrete}.
\end{proof}

\section{The anchor map on the sheaf side}
\label{sec:anchor_sheaf}

Under the identifications introduced above, the Beauville--Bottacin anchor map is given by the composition
\[
\pi^\sharp_{\widetilde{\mathcal M}}:
\Ext^1(\widetilde{\mathcal L},\widetilde{\mathcal L}(-\widetilde D))^{G}
\xrightarrow{\ \Psi^{-1}\ }
\Ext^1\!\bigl(\widetilde{\mathcal L},
[\widetilde{\mathcal L}\to \widetilde{\mathcal L}|_{\widetilde D}]\bigr)^G
\xrightarrow{\ \widetilde\theta\ }
\Ext^1(\widetilde{\mathcal L},\widetilde{\mathcal L})^{G}.
\]
Thus it remains to compute the two maps \(\Psi^{-1}\) and \(\widetilde\theta\).

\subsection{The map \texorpdfstring{$\widetilde\theta$}{theta} in coordinates}

In this subsection, we compute the map $\widetilde \theta$ that plays the role of the map \(j^*\) on the graph side.

\begin{proposition}\label{prop:extlldtoextld}
With the identifications in Proposition~\ref{prop:repext} and Proposition~\ref{prop:repextld}, the map
\[
\widetilde \theta:\ \Ext^1(\widetilde{\mathcal L},\,[\widetilde{\mathcal L}\to \widetilde{\mathcal L}|_{\widetilde D}])^G
\rightarrow \Ext^1(\widetilde{\mathcal L},\widetilde{\mathcal L})^G
\]
is induced by the projection
\[
\C^{\E}\oplus \bigoplus_{\alpha\in \zz}\Hom(\C^{\W\cap \alpha},\C)\rightarrow \C^{\E},
\qquad
(u,t)\mapsto u.
\]
\end{proposition}

\begin{proof}
Recall that the map
\[
\widetilde \theta:\Ext^1(\widetilde{\mathcal L},[\widetilde{\mathcal L}\to \widetilde{\mathcal L}|_{\widetilde D}])^G
\rightarrow
\Ext^1(\widetilde{\mathcal L},\widetilde{\mathcal L})^G
\]
is induced by the morphism of complexes
\[
[\widetilde{\mathcal L}\xrightarrow{r}\widetilde{\mathcal L}|_{\widetilde D}]
\rightarrow
\widetilde{\mathcal L},
\]
given by projection onto the first term.

We compute both Ext groups using the locally free resolutions $K_{\widetilde{D}}^\bullet$ and $K^\bullet$. The morphism above is represented by the morphism of complexes
\[
\Bigl[
\begin{array}{c c c c c}
\cale & \xrightarrow{\ \widetilde \K\ } & \calf & \to & \widetilde{\mathcal L}|_{\widetilde D}\\[-.4ex]
{\scriptstyle -1} && {\scriptstyle 0} && {\scriptstyle 1}
\end{array}
\Bigr] \rightarrow \Bigl[ \begin{array}{c c c} \cale & \xrightarrow{\ \widetilde \K\ } & \calf\\[-.4ex] {\scriptstyle -1} & & {\scriptstyle 0} \end{array} \Bigr] 
\]
given by projection onto the first two terms. Applying
\[
\check C^\bullet\bigl(\widetilde{\mathcal U},\mathcal Hom(K^\bullet,-)\bigr)^G
\]
therefore gives a morphism of double complexes
from the double complex computing
\(\Ext^\bullet(\widetilde{\mathcal L},[\widetilde{\mathcal L}\to \widetilde{\mathcal L}|_{\widetilde D}])^G\)
to the one computing
\(\Ext^\bullet(\widetilde{\mathcal L},\widetilde{\mathcal L})^G\). Passing to \v Cech cohomology, we get the map stated in the proposition.
\end{proof}

\subsection{Comparison of the two cotangent models}

In this section, we compute the map $\Psi$ identifying the two cotangent space models.

\begin{proposition} \label{prop:isom_ext}
With the identifications in Proposition~\ref{prop:repext2} and Proposition~\ref{prop:repextld}, the isomorphism 
\[
\Psi: \Ext^1(\widetilde{\mathcal L},\,[\widetilde{\mathcal L}\to \widetilde{\mathcal L}|_{\widetilde D}])^G \xrightarrow[]{\ \cong\ }\Ext^1(\widetilde{\mathcal L},\widetilde{\mathcal L}(-\widetilde D))^G
\]
is induced by the map
\[
\C^{\E}\oplus \bigoplus_{\alpha\in \zz}\Hom(\C^{\W\cap \alpha},\C)\rightarrow \C^{\E},
\qquad
(v,t)\mapsto \left( \frac{t_{\alpha_\e^-}(\w)-t_{\alpha_\e^+}(\w)}{\wt(\e)} \right)_{\e \in \E}.
\]
\end{proposition}

The goal of this subsection is to compute the above isomorphism explicitly. Rather than work directly with representatives, we characterize it by compatibility with cup products. The argument has three steps:
\begin{enumerate}
\item Compute the cup product
\[
\Ext^1(\widetilde{\mathcal L},\widetilde{\mathcal L})^G
\otimes
\Ext^1\!\bigl(\widetilde{\mathcal L},[\widetilde{\mathcal L}\to \widetilde{\mathcal L}|_{\widetilde D}]\bigr)^G
\to
\Ext^2\!\bigl(\widetilde{\mathcal L},[\widetilde{\mathcal L}\to \widetilde{\mathcal L}|_{\widetilde D}]\bigr)^G.
\]
\item Compare the two linear functionals on the resulting one-dimensional \(\Ext^2\)-spaces.
\item Recover the isomorphism from nondegeneracy of the Serre duality pairing.
\end{enumerate}

Let
\[
\Psi_2:
\Ext^2\!\bigl(\widetilde{\mathcal L},[\widetilde{\mathcal L}\to \widetilde{\mathcal L}|_{\widetilde D}]\bigr)^G
\xrightarrow{\ \cong\ }
\Ext^2(\widetilde{\mathcal L},\widetilde{\mathcal L}(-\widetilde D))^G
\]
denote the map induced by the restriction morphism in the second variable. By functoriality of the cup product, we obtain a commuting diagram
\[
\begin{tikzcd}
\Ext^1(\widetilde{\mathcal L},\widetilde{\mathcal L})^G
\otimes
\Ext^1\!\bigl(\widetilde{\mathcal L},[\widetilde{\mathcal L}\to \widetilde{\mathcal L}|_{\widetilde D}]\bigr)^G  \arrow[r,"\cup"] \arrow[d,"1\otimes \Psi"']
&  \Ext^2\!\bigl(\widetilde{\mathcal L},[\widetilde{\mathcal L}\to \widetilde{\mathcal L}|_{\widetilde D}]\bigr)^G\arrow[d,"\Psi_2"] \\
\Ext^1(\widetilde{\mathcal L},\widetilde{\mathcal L})^G \otimes
\Ext^1(\widetilde{\mathcal L},\widetilde{\mathcal L}(-\widetilde D))^G  \arrow[r,"\cup"] 
&  \Ext^2(\widetilde{\mathcal L},\widetilde{\mathcal L}(-\widetilde D))^G.
\end{tikzcd}
\]
Thus we are reduced to computing the simpler isomorphism 
\[
\Ext^2(\widetilde{\mathcal L},\widetilde{\mathcal L}(-\widetilde D))^G \cong  \Ext^2\!\bigl(\widetilde{\mathcal L},[\widetilde{\mathcal L}\to \widetilde{\mathcal L}|_{\widetilde D}]\bigr)^G
\]
between one-dimensional spaces.

\subsubsection{The cup product}\label{subsection cup product}

The main result of this subsubsection is the following.

\begin{lemma}\label{lem::cup}
With the identifications of Proposition~\ref{prop:repext} and Proposition~\ref{prop:repextld}, the cup product
\[
\cup:
\Ext^1(\widetilde{\mathcal L},\widetilde{\mathcal L})^G
\otimes
\Ext^1\!\bigl(\widetilde{\mathcal L},[\widetilde{\mathcal L}\to \widetilde{\mathcal L}|_{\widetilde D}]\bigr)^G
\rightarrow
\Ext^2\!\bigl(\widetilde{\mathcal L},[\widetilde{\mathcal L}\to \widetilde{\mathcal L}|_{\widetilde D}]\bigr)^G
\]
is induced by the map
\[
\C^\E \otimes
\left(
\C^\E \oplus \bigoplus_{\alpha\in\zz}\Hom(\C^{\W\cap\alpha},\C)
\right)
\rightarrow
\bigoplus_{\alpha\in\zz}\Hom(\C^{\B\cap\alpha},\C),
\qquad
u\otimes(v,t)\mapsto (s_\alpha)_{\alpha\in\zz},
\]
where, for each \(\alpha\in\zz\) and each wedge
\[
\w_\alpha^{-}(\b)\;-\;\b\;-\;\w_\alpha^{+}(\b)
\]
along \(\alpha\), we have
\[
s_\alpha(\b)
=
\frac{u(\b \w_\alpha^{-}(\b))}{\wt(\b\w_\alpha^{-}(\b))}\,t_\alpha(\w_\alpha^{-}(\b))
-
\frac{u(\b \w_\alpha^{+}(\b))}{\wt(\b\w_\alpha^{+}(\b))}\,t_\alpha(\w_\alpha^{+}(\b)).
\]
\end{lemma}

\begin{proof}
Using the \v{C}ech--\(\mathcal Hom\) double complexes from Sections~\ref{sec:dk} and~\ref{sec:ext-ld}, it suffices to compute the induced map on the \(p=0\) column:
\[
\Hom(\cale,\calf)^G
\otimes
\Bigl(
\Hom(\cale,\calf)^G\oplus \Hom(\calf,\widetilde{\mathcal L}|_{\widetilde D})^G
\Bigr)
\rightarrow
\Hom(\cale,\widetilde{\mathcal L}|_{\widetilde D})^G.
\]
Only the summand \(\Hom(\calf,\widetilde{\mathcal L}|_{\widetilde D})^G\) in the second factor can contribute to the target, so the component \(v\in\Hom(\cale,\calf)^G\) plays no role in this cup product.
By Proposition~\ref{prop:repext}, \(u\in \C^\E\) is identified with
\[
l = \bigl(u(\e)\,\widetilde x(\e)\bigr)_{\e\in\E} \in \Hom(\cale,\calf)^G.
\]
The cup product is therefore represented by
\[
t \circ l
\in
\Hom(\cale,\widetilde{\mathcal L}|_{\widetilde D})^G.
\]
It remains to identify this map explicitly.

Fix \(\alpha\in\zz\) and a point \(\widetilde\alpha_i\in \widetilde D\) lying over \(\alpha\). Under the identifications in Lemma~\ref{lem:ident}, the \(\alpha\)-component of \(t \circ l\) is the composition
\[
E_{\alpha,i}
\xrightarrow{\ \restr{l}{\widetilde\alpha_i}\ }
F_{\alpha,i}
\xrightarrow{\ t_\alpha\ }
L_{\alpha,i}\cong \C.
\]
Under the gauge transformation identifying \(\restr{\widetilde\K_\alpha}{\widetilde\alpha_i}\) with \(\partial_\alpha\), the matrix representing \(\restr{l}{\widetilde\alpha_i}\) has, in the column $\b$, exactly two nonzero entries:
\[
\frac{u(\b \w_\alpha^{-}(\b))}{\wt(\b \w_\alpha^{-}(\b))}
\quad\text{in row }\w_\alpha^{-}(\b),
\qquad
-\frac{u(\b \w_\alpha^{+}(\b))}{\wt(\b \w_\alpha^{+}(\b))}
\quad\text{in row }\w_\alpha^{+}(\b).
\]
Composing with the row vector $t_\alpha$, it follows that the composition
\[
s_\alpha=t_\alpha\circ \restr{l}{\widetilde\alpha_i}\in \Hom(\C^{\B\cap\alpha},\C)
\]
is given by
\[
s_\alpha(\b)
=
\frac{u(\b \w_\alpha^{-}(\b))}{\wt(\b\w_\alpha^{-}(\b))}\,t_\alpha(\w_\alpha^{-}(\b))
-
\frac{u(\b \w_\alpha^{+}(\b))}{\wt(\b\w_\alpha^{+}(\b))}\,t_\alpha(\w_\alpha^{+}(\b)).
\]
\end{proof}

\subsubsection{The isomorphism between the two \texorpdfstring{$\Ext^2$}{Ext2(L-tilde,L-tilde(-D-tilde)), G-invariants}-spaces}

Define the map 
\[
\phi: \Ext^2\!\bigl(\widetilde{\mathcal L},[\widetilde{\mathcal L}\to \widetilde{\mathcal L}|_{\widetilde D}]\bigr)^G \rightarrow \C
\]
induced by the map
\[
\bigoplus_{\alpha\in\zz}\Hom(\C^{\B\cap\alpha},\C) \rightarrow \C, \qquad (s_\alpha)_{\alpha \in \zz} \mapsto \sum_{\alpha \in \zz} \sum_{\b \in \B \cap \alpha} s_\alpha(\b).
\]
Recall the trace map
\[
\widetilde \theta\circ\tr:
\Ext^2(\widetilde{\mathcal L},\widetilde{\mathcal L}(-\widetilde D))^G
\rightarrow \C
\]
induced by
\[
\C^{\B}\oplus \C^{\W}\rightarrow \C,
\qquad
(f,g)\mapsto \sum_{\w\in\W}g(\w)-\sum_{\b\in\B}f(\b).
\]
from Proposition~\ref{prop:trace-ext2}. 

\begin{proposition} \label{prop:scalar}
The scalar $c \in \C^\times$ that makes the diagram
\[
\begin{tikzcd}
 \Ext^2(\widetilde{\mathcal L},\widetilde{\mathcal L}(-\widetilde D))^G \arrow[r,"{\widetilde \theta\circ\tr}"]  \arrow[d,"\Psi_2"]
& \C \arrow[d,"c"]\\
\Ext^2\!\bigl(\widetilde{\mathcal L},[\widetilde{\mathcal L}\to \widetilde{\mathcal L}|_{\widetilde D}]\bigr)^G  \arrow[r,"\phi"]
&  \C
\end{tikzcd},
\]
commute is $c=-1$, where the map on the right is multiplication by $c$. 
\end{proposition}

The proof of Proposition~\ref{prop:scalar} is a straightforward but somewhat tedious diagram chase, so we defer it to Appendix~\ref{app:prop:scalar}.

\begin{proof}[Proof of Proposition~\ref{prop:isom_ext}]
Let 
\[
u\in \C^\E,
\qquad
(v,t)\in
\C^{\E}\oplus \bigoplus_{\alpha\in \zz}\Hom(\C^{\W\cap \alpha},\C).
\]
By Lemma~\ref{lem::cup}, we have
\[
\phi\bigl(u\cup (v,t)\bigr)
=
\sum_{\alpha\in\zz}\sum_{\b\in\B\cap\alpha}
\Bigl(
\frac{u(\b \w_\alpha^{-}(\b))}{\wt(\b\w_\alpha^{-}(\b))}\,t_\alpha(\w_\alpha^{-}(\b))
-
\frac{u(\b \w_\alpha^{+}(\b))}{\wt(\b\w_\alpha^{+}(\b))}\,t_\alpha(\w_\alpha^{+}(\b))
\Bigr).
\]
Regrouping the sum by edges gives
\[
\phi\bigl(u\cup (v,t)\bigr)
=
\sum_{\e=\b\w\in\E}
\frac{u(\e)}{\wt(\e)}
\bigl(t_{\alpha_\e^-}(\w)-t_{\alpha_\e^+}(\w)\bigr).
\]
On the other hand, under the lower cup product and Proposition~\ref{prop:gsdual}, the Serre duality pairing gives
\[
u\otimes \Psi(v,t) \longmapsto -\sum_{\e\in\E}u(\e)\Psi(v,t)(\e).
\]
By Proposition~\ref{prop:scalar}, the functional \(\phi\) is equal to \(-1\) times \(\widetilde\theta\circ\tr\) under the induced isomorphism on \(\Ext^2\). Therefore
\[
\phi\bigl(u\cup (v,t)\bigr)
=
\sum_{\e\in\E}u(\e)\Psi(v,t)(\e).
\]
Since this holds for all \(u\in \C^\E\),
\[
\Psi(v,t)(\e)
=
\frac{t_{\alpha_\e^-}(\w)-t_{\alpha_\e^+}(\w)}{\wt(\e)}
\qquad (\e=\b\w\in\E),
\]
as claimed.
\end{proof}

\section{The differential of the spectral transform}
\label{sec:diff_kappa}

The following proposition identifies the differential of the spectral transform.

\begin{proposition}\label{prop:dk}
Let \((\widetilde{\mathcal L},\bm \nu)=\kappa([\wt])\). Under the identification of Proposition~\ref{prop:repext}, the differential
\[
d\kappa:\ T_{[\wt]}\mathcal X \cong H^1(\graph,\C)\rightarrow
T_{\kappa([\wt])}\widetilde{\mathcal M}\cong \Ext^1(\widetilde{\mathcal L},\widetilde{\mathcal L})^G
\]
is induced by the identity map
\[
u \mapsto \wt u \in \C^{\E}.
\]
\end{proposition}

\begin{proof}
A total \(1\)-cochain is a triple $(l,m,n)$ where
\begin{align*}
l&=(l_i)\in \check C^0 \bigl(\widetilde{\mathcal U},\mathcal H om(\cale,\calf)\bigr)^G,\\
m&=(m_{ij})=(m_{ij}^\cale,m_{ij}^\calf)\in
\check C^1 \bigl(\widetilde{\mathcal U},\mathcal H om(\cale,\cale)\oplus\mathcal H om(\calf,\calf)\bigr)^G,\\
n&=(n_{ijk})\in \check C^2 \bigl(\widetilde{\mathcal U},\mathcal H om(\calf,\cale)\bigr)^G.
\end{align*}
Since
\[
d^0(f,g)=\widetilde\K\circ f-g\circ \widetilde\K
\qquad\text{and}\qquad
d^{-1}(u)=\bigl(u\circ \widetilde\K,\widetilde\K\circ u\bigr),
\]
the condition \(D(l,m,n)=0\) is equivalent to
\begin{align*}
&\check d(l)-d^0(m)=0,\\
&\check d(m)+d^{-1}(n)=0,\\
&\check d(n)=0.
\end{align*}
Explicitly, these equations are
\begin{align}
l_j-l_i-\bigl(\widetilde\K\circ m_{ij}^\cale-m_{ij}^\calf\circ\widetilde\K\bigr)&=0,\label{eq:dk-cocycle-1}\\
m_{jk}-m_{ik}+m_{ij}+\bigl(n_{ijk}\circ\widetilde\K,\widetilde\K\circ n_{ijk}\bigr)&=0,\label{eq:dk-cocycle-2}\\
\check d(n)&=0.\label{eq:dk-cocycle-3}
\end{align}
Since the cover \(\widetilde{\mathcal U}\) has only three open sets,~\eqref{eq:dk-cocycle-3} holds automatically.

We now recall the deformation-theoretic meaning of these equations following~\cite{dgla}. For clarity, we temporarily write \(\widetilde\K(\wt)\) instead of \(\widetilde\K\). Fix \(\wt\) and consider the complex
\[
K^\bullet=
\bigl[\cale\xrightarrow{\widetilde\K(\wt)}\calf\bigr].
\]
Over an affine open \(\widetilde U_i\in\widetilde{\mathcal U}\), a first-order deformation of the differential is given by
\[
\widetilde\K(\wt)\mapsto \widetilde\K(\wt)+\varepsilon l_i,
\qquad
l_i\in\Hom \bigl(\cale|_{\widetilde U_i},\calf|_{\widetilde U_i}\bigr)^G,
\]
and hence defines a local deformation of complexes over \(\C[\varepsilon]/(\varepsilon^2)\):
\[
\left[
\cale|_{\widetilde U_i}\otimes\C[\varepsilon]/(\varepsilon^2)
\xrightarrow{\ \widetilde\K(\wt)+\varepsilon l_i\ }
\calf|_{\widetilde U_i}\otimes\C[\varepsilon]/(\varepsilon^2)
\right].
\]

To glue these local deformations over \(\widetilde U_{ij}:=\widetilde U_i\cap \widetilde U_j\), we choose chain isomorphisms
\[
1+\varepsilon m_{ij},
\qquad
m_{ij}=(m_{ij}^\cale,m_{ij}^\calf),
\]
so that the diagram
\[
\begin{tikzcd}[row sep=large, column sep=large]
\cale|_{\widetilde U_{ij}}
\arrow[r,"\widetilde\K(\wt)+\varepsilon l_j"]
\arrow[d,"1+\varepsilon m_{ij}^\cale"']
&
\calf|_{\widetilde U_{ij}}
\arrow[d,"1+\varepsilon m_{ij}^\calf"]
\\
\cale|_{\widetilde U_{ij}}
\arrow[r,"\widetilde\K(\wt)+\varepsilon l_i"']
&
\calf|_{\widetilde U_{ij}}
\end{tikzcd}
\]
commutes modulo \(\varepsilon^2\). Comparing the \(\varepsilon\)-terms gives precisely~\eqref{eq:dk-cocycle-1}.

On triple overlaps \(\widetilde U_{ijk}\), the gluing maps satisfy the cocycle condition up to chain homotopy. Thus there exist
\[
n_{ijk}\in \Hom \bigl(\calf|_{\widetilde U_{ijk}},\cale|_{\widetilde U_{ijk}}\bigr)^G
\]
such that
\[
(1+\varepsilon m_{ik})-(1+\varepsilon m_{ij})(1+\varepsilon m_{jk})
=
\varepsilon d^{-1}(n_{ijk}).
\]
Comparing the \(\varepsilon\)-terms gives exactly~\eqref{eq:dk-cocycle-2}. Hence a first-order deformation of \(\widetilde{\mathcal L}\) determines a cocycle \((l,m,n)\) representing its class in
\(\Ext^1(\widetilde{\mathcal L},\widetilde{\mathcal L})^G\).

Now let
\[
 u \in T_{[\wt]}\mathcal X\cong H^1(\graph,\C),
\]
and consider the first-order deformation
\[
\wt(\varepsilon):=\wt(1+\varepsilon u)
\qquad\text{over }\C[\varepsilon]/(\varepsilon^2).
\]
Since \(\widetilde\K(\cdot)\) depends linearly on the edge weights,
\[
\widetilde\K(\wt(\varepsilon))
=
\widetilde\K(\wt)+\varepsilon\widetilde\K(\wt u).
\]
Thus the induced deformation of the spectral sheaf is the cokernel of
\[
0\rightarrow \cale\otimes\C[\varepsilon]/(\varepsilon^2)
\xrightarrow{\ \widetilde\K(\wt)+\varepsilon\widetilde\K(\wt u)\ }
\calf\otimes\C[\varepsilon]/(\varepsilon^2)
\rightarrow \widetilde{\mathcal L}_\varepsilon
\rightarrow 0.
\]
In this case we may take
\[
l_i=\widetilde\K(\wt u)
\qquad\text{for all }i,
\]
and
\[
m_{ij}=0,\qquad n_{ijk}=0.
\]
Thus the resulting class in \(\Ext^1(\widetilde{\mathcal L},\widetilde{\mathcal L})^G\) is represented by the global section
\[
\widetilde\K(\wt u)\in \Hom(\cale,\calf)^G.
\]
Under the identification
\[
\Hom(\cale,\calf)^G\cong \C^{\E}
\]
from Proposition~\ref{prop:repext}, this section corresponds exactly to the edgewise product
\[
\wt u \in \C^{\E}.
\]
Therefore \(d\kappa\) is induced by the map
\[
\dot\wt\mapsto \wt u,
\]
as claimed.
\end{proof}


\section{Proof of Theorem~\ref{thm:main}} \label{sec:last}

Consider the complexes~\eqref{eq:relcoh-concrete} and~\eqref{eq:Fbullet-concrete}
computing \(H^1(\widehat \Gamma,\partial \widehat \Gamma;\C)\) and \(
\Ext^1\!\bigl(\widetilde{\mathcal L},[\widetilde{\mathcal L}\to \widetilde{\mathcal L}|_{\widetilde D}]
\bigr)^G \) respectively. Define a chain isomorphism
\[
\Phi^\bullet:\eqref{eq:relcoh-concrete} \rightarrow\eqref{eq:Fbullet-concrete}
\]
by
\[
\Phi^0(f,g)=(-f,-g),\qquad
\Phi^1(u,t)=(\wt u,t),\qquad
\Phi^2=\id.
\]
This induces an isomorphism
\[
\Phi:H^1(\Gamma,\partial \widehat \Gamma;\C)
\xrightarrow{\cong}
\Ext^1\!\bigl(\widetilde{\mathcal L},
[\widetilde{\mathcal L}\to \widetilde{\mathcal L}|_{\widetilde D}]
\bigr)^G.
\]
We first prove commutativity of
\begin{equation}\label{eq:square1}
\begin{tikzcd}
H^1(\Gamma,\partial \widehat \Gamma;\C)
  \arrow[r,"j^*"]
  \arrow[d,"\Phi"']
&
T_{[\wt]}\mathcal X=H^1(\graph,\C)
  \arrow[d,"d\kappa"]\\
\Ext^1\!\bigl(\widetilde{\mathcal L},
[\widetilde{\mathcal L}\to \widetilde{\mathcal L}|_{\widetilde D}]
\bigr)^G
  \arrow[r,"\widetilde\theta"']
&
T_{(\widetilde{\mathcal L},\bm\nu)}\widetilde{\mathcal M}
=
\Ext^1(\widetilde{\mathcal L},\widetilde{\mathcal L})^G.
\end{tikzcd}
\end{equation}
By Proposition~\ref{prop:extlldtoextld},
\[
\widetilde\theta\bigl(\Phi([u,t])\bigr)
=
[\wt u]
\in
\Ext^1(\widetilde{\mathcal L},\widetilde{\mathcal L})^G.
\]
On the other hand, Proposition~\ref{prop:dk} identifies the differential of the spectral
transform with
\[
d\kappa([u])=[\wt u]
\in
\Ext^1(\widetilde{\mathcal L},\widetilde{\mathcal L})^G.
\]
Therefore,
\[
\widetilde\theta\circ\Phi([u,t])
=
[\wt u]
=
d\kappa([u])
=
d\kappa\circ j^*([u,t]),
\]
so the square~\eqref{eq:square1} commutes.

Next, consider the complex
\be \label{eq:chain}
\Bigl[
\begin{array}{c c c}
\C^{\E} & \xrightarrow{\ \partial\ } & \C^{\B}\oplus \C^{\W}\\[-.4ex]
{\scriptstyle 1} && {\scriptstyle 2}
\end{array}
\Bigr],
\qquad
\partial(\b\w)=\w-\b,
\ee
computing $H_1(\graph,\C)$ and the complex~\eqref{eq:chain_ext-d} computing \(\Ext^1(\widetilde{\mathcal L},\widetilde{\mathcal L}(-\widetilde D))^G\) from Proposition~\ref{prop:repext2}. Define a chain isomorphism \[\Xi^\bullet: \eqref{eq:chain} \to \eqref{eq:chain_ext-d}\] by
\[
\Xi^1(u)=-\frac{u}{\wt},
\qquad
\Xi^2(f,g)=(f,-g).
\]
This induces an isomorphism
\[
\Xi:H_1(\graph,\C)
\xrightarrow{\cong}
\Ext^1(\widetilde{\mathcal L},\widetilde{\mathcal L}(-\widetilde D))^G.
\]
We claim that the following diagram commutes:
\[
\begin{tikzcd}T^*_{[\wt]}\mathcal X
  \arrow[r,"\cong"]
&[8em]
H_1(\graph,\C)
  \arrow[r,"{{\PD^{-1}}}","\cong"']
  \arrow[d,"\Xi"']
&
H^1(\widehat \graph, \partial \widehat \graph;\C)
  \arrow[d,"\Phi"]
\\
T^*_{(\widetilde{\mathcal L},\bm\nu)}\widetilde{\mathcal M}
  \arrow[u,"d\kappa^*"]
  \arrow[r,"{\text{Serre duality}}","\cong"']
&
\Ext^1(\widetilde{\mathcal L},\widetilde{\mathcal L}(-\widetilde D))^G
  \arrow[r,"\cong"',"\Psi^{-1}"]
&
\Ext^1\!\bigl(\widetilde{\mathcal L},
[\widetilde{\mathcal L}\to \widetilde{\mathcal L}|_{\widetilde D}]
\bigr)^G.
\end{tikzcd}
\]

For the left square, let \([u]\in H^1(\graph,\C)\) and \([a]\in H_1(\graph,\C) \).
By Proposition~\ref{prop:dk},
\[
d\kappa([u])=[\wt u]
\in
\Ext^1(\widetilde{\mathcal L},\widetilde{\mathcal L})^G.
\]
Hence, using Proposition~\ref{prop:gsdual}, the Serre duality pairing maps
\[
 d\kappa([u]) \otimes \Xi([a]) \mapsto 
-\sum_{\e\in\E} (\wt u)(\e)\left(-\frac{a(\e)}{\wt(\e)}\right)
=
\sum_{\e\in\E}u(\e)a(\e).
\]
which agrees with the pairing between $H^1(\graph,\C)$ and $H_1(\graph,\C)$. Thus the left square commutes.

For the right square, let $[u,t] \in H^1(\widehat\graph,\partial \widehat \Gamma;\C).$ Then
\[
\Psi \circ \Phi([u,t]) = \Psi([\wt u,t]) = \left[\left(\frac{t_{\alpha_\e^-}(\w)-t_{\alpha_\e^+}(\w)}{\wt(\e)}\right)_{\e \in \E}\right].
\]
On the other hand, the Poincar\'e duality isomorphism $\PD$ in Proposition~\ref{prop:PDiso} maps \[[u,t]\mapsto[({t_{\alpha_\e^+}(\w)-t_{\alpha_\e^-}(\w)})_{\e \in \E}]\] which under $\Xi$ also gets sent to \[\left[\left(\frac{t_{\alpha_\e^-}(\w)-t_{\alpha_\e^+}(\w)}{\wt(\e)}\right)_{\e \in \E}\right].\]
We have thus shown commutativity of~\eqref{eq:big_square}, completing the proof of Theorem~\ref{thm:main}.

\begin{appendices}


\section{Equivariant line bundles on \texorpdfstring{$\widetilde{\mathbb{P}}^2$}{P2-tilde} and their cohomology}
\label{appendix:G_sheaves}

For the convenience of the reader, we collect in this appendix the basic facts about \(G\)-equivariant line bundles on \(\widetilde{\P}^2\) needed in the body of the paper. After recalling the standard description of line bundles and their cohomology on projective space, we explain how the \(G\)-linearization picks out distinguished \(G\)-invariant monomials in \(H^0\) and \(H^2\). Equivalently, one may regard this as an ordinary cohomology computation on the quotient stack
$[\widetilde{\P}^2/G]$; see Remark~\ref{rem:stack_why}. These descriptions will be used repeatedly in our computations of equivariant \(\Hom\)- and \(\Ext\)-groups.

Let \(X\) be either \(\widetilde{\P}^2\) or \(\widetilde C\), and let \(G\) be the group introduced in Section~\ref{sec:equiv_bottacin}. A \emph{\(G\)-equivariant vector bundle} on \(X\) is a vector bundle \(\pi:V\to X\) together with an action of \(G\) on the total space of \(V\) such that \(\pi\) is \(G\)-equivariant and the induced map
\[
V_x \rightarrow V_{gx}, \qquad v \mapsto gv,
\]
is linear for every \(g\in G\) and \(x\in X\).

\subsection{Line bundles on \texorpdfstring{$\widetilde{\P}^2$}{P2-tilde}}

We begin by recalling the construction of line bundles on \(\widetilde{\P}^2\). These are indexed by \(\Z\) and denoted \(\cO_{\widetilde{\P}^2}(k)\). Identifying \(\widetilde{\P}^2\) with the quotient \((\C^3\setminus\{0\})/\C^\times\), the line bundle \(\cO_{\widetilde{\P}^2}(k)\) is the quotient
\[
(\C^3\setminus\{0\}\times \C)/\C^\times
\]
with respect to the action
\[
\lambda\cdot (\widetilde x_0,\widetilde x_1,\widetilde x_2,t)
=
(\lambda \widetilde x_0,\lambda \widetilde x_1,\lambda \widetilde x_2,\lambda^k t).
\]

A global section of \(\cO_{\widetilde{\P}^2}(k)\) is represented by a homogeneous polynomial of degree \(k\). Indeed, a monomial
\[
[\widetilde x_0:\widetilde x_1:\widetilde x_2]
\mapsto
[\widetilde x_0:\widetilde x_1:\widetilde x_2,\,
\widetilde x_0^{m_0}\widetilde x_1^{m_1}\widetilde x_2^{m_2}]
\]
is a well defined precisely when
\[
m_0+m_1+m_2=k.
\]
Thus
\begin{equation}\label{eq:h0}
H^0(\widetilde{\P}^2,\cO_{\widetilde{\P}^2}(k))
\cong
\C[x_0,x_1,x_2]_k.
\end{equation}

\subsection{Cohomology of line bundles on \texorpdfstring{$\widetilde{\P}^2$}{P2-tilde}}

Higher cohomology may be computed from the \v Cech complex with respect to the standard affine cover
\[
\widetilde{\mathcal U}=(\widetilde U_i)_{i=0,1,2},
\qquad
\widetilde U_i
=
\bigl\{
[\widetilde x_0:\widetilde x_1:\widetilde x_2]\in\widetilde{\P}^2 : \widetilde x_i\neq 0
\bigr\}.
\]
For \(I\subset \{0,1,2\}\), write
\[
\widetilde U_I:=\bigcap_{i\in I}\widetilde U_i.
\]
By the same argument as above,
\[
\Gamma(\widetilde U_I,\cO_{\widetilde{\P}^2}(k))
=
\C[x_0,x_1,x_2,x_i^{-1}\ (i\in I)]_k,
\]
that is, the homogeneous degree-\(k\) Laurent polynomials in which \(x_j\) appears with nonnegative exponent for \(j\notin I\).

The \v Cech complex for \(\cO_{\widetilde{\P}^2}(k)\) is
\[
0
\to
\check C^0(\widetilde{\mathcal U},\cO_{\widetilde{\P}^2}(k))
\xrightarrow{\ \check d\ }
\check C^1(\widetilde{\mathcal U},\cO_{\widetilde{\P}^2}(k))
\xrightarrow{\ \check d\ }
\check C^2(\widetilde{\mathcal U},\cO_{\widetilde{\P}^2}(k))
\to 0,
\]
where
\begin{align*}
\check C^0(\widetilde{\mathcal U},\cO_{\widetilde{\P}^2}(k))
&=
\Gamma(\widetilde U_0,\cO_{\widetilde \P^2}(k))
\oplus
\Gamma(\widetilde U_1,\cO_{\widetilde \P^2}(k))
\oplus
\Gamma(\widetilde U_2,\cO_{\widetilde \P^2}(k)),\\
\check C^1(\widetilde{\mathcal U},\cO_{\widetilde{\P}^2}(k))
&=
\Gamma(\widetilde U_{01},\cO_{\widetilde \P^2}(k))
\oplus
\Gamma(\widetilde U_{02},\cO_{\widetilde \P^2}(k))
\oplus
\Gamma(\widetilde U_{12},\cO_{\widetilde \P^2}(k)),\\
\check C^2(\widetilde{\mathcal U},\cO_{\widetilde{\P}^2}(k))
&=
\Gamma(\widetilde U_{012},\cO_{\widetilde \P^2}(k)),
\end{align*}
and the differentials are
\[
\check d(\alpha_0,\alpha_1,\alpha_2)
=
(\alpha_1-\alpha_0,\alpha_2-\alpha_0,\alpha_2-\alpha_1),
\qquad \check d(\alpha_{01},\alpha_{02},\alpha_{12})
=
\alpha_{12}-\alpha_{02}+\alpha_{01}.
\]
Taking cohomology gives
\begin{align}
H^1(\widetilde{\P}^2,\cO_{\widetilde{\P}^2}(k))
&=0,
\nonumber\\
H^2(\widetilde{\P}^2,\cO_{\widetilde{\P}^2}(k))
&=
\frac{\C[x_0^{\pm1},x_1^{\pm1},x_2^{\pm1}]_k}
{\Gamma(\widetilde U_{01},\cO_{\widetilde \P^2}(k))
+
\Gamma(\widetilde U_{02},\cO_{\widetilde \P^2}(k))
+
\Gamma(\widetilde U_{12},\cO_{\widetilde \P^2}(k))}
\nonumber\\
&\cong
\frac{1}{x_0x_1x_2}\,
\C[x_0^{-1},x_1^{-1},x_2^{-1}]_{-k-3}.
\label{eq:h2}
\end{align}

\subsection{\texorpdfstring{$G$}{G}-equivariant line bundles on \texorpdfstring{$\widetilde{\P}^2$}{P2-tilde}}

For \(a_0,a_1,a_2 \in \Z\), let
\[
k:=a_0+a_1+a_2.
\]
\begin{definition}
    
The \(G\)-equivariant line bundle
\[
\cO^G_{\widetilde{\P}^2}(a_0D_0+a_1D_1+a_2D_2)
\]
is the line bundle \(\cO_{\widetilde{\P}^2}(k)\) together with the \(G\)-linearization induced from the action
\[
(\zeta_0,\zeta_1,\zeta_2)\cdot [\widetilde x_0:\widetilde x_1:\widetilde x_2:t]
=
[\zeta_0\widetilde x_0:\zeta_1\widetilde x_1:\zeta_2\widetilde x_2:\zeta_0^{a_0}\zeta_1^{a_1}\zeta_2^{a_2}t].
\]
\end{definition}

The diagonal subgroup \(\Delta\) acts trivially, since for \(\zeta\in \bmu_d\),
\[
(\zeta,\zeta,\zeta)\cdot [\widetilde x_0:\widetilde x_1:\widetilde x_2:t]
=
[\zeta \widetilde x_0:\zeta \widetilde x_1:\zeta \widetilde x_2:\zeta^k t]
=
[\widetilde x_0:\widetilde x_1:\widetilde x_2:t].
\]
Hence the action descends to \(G=\bmu_d^{\times 3}/\Delta\).

A section
\[
s([\widetilde x_0:\widetilde x_1:\widetilde x_2])
=
[\widetilde x_0:\widetilde x_1:\widetilde x_2:\widetilde x_0^{m_0}\widetilde x_1^{m_1}\widetilde x_2^{m_2}],
\qquad m_0+m_1+m_2=k,
\]
is \(G\)-invariant if and only if for every \([\zeta_0,\zeta_1,\zeta_2]\in G\),
\[
[\zeta_0\widetilde x_0:\zeta_1\widetilde x_1:\zeta_2\widetilde x_2:
\zeta_0^{a_0}\zeta_1^{a_1}\zeta_2^{a_2}
\widetilde x_0^{m_0}\widetilde x_1^{m_1}\widetilde x_2^{m_2}]
\]
agrees with
\[
[\zeta_0\widetilde x_0:\zeta_1\widetilde x_1:\zeta_2\widetilde x_2:
\zeta_0^{m_0}\zeta_1^{m_1}\zeta_2^{m_2}
\widetilde x_0^{m_0}\widetilde x_1^{m_1}\widetilde x_2^{m_2}],
\]
that is, if and only if
\[
(m_0,m_1,m_2)\in (a_0,a_1,a_2)+d\Z^3.
\]
Therefore
\begin{equation*}
H^0\!\bigl(\widetilde{\P}^2,\cO^G_{\widetilde{\P}^2}(a_0D_0+a_1D_1+a_2D_2)\bigr)^G
=
\bigoplus_{\substack{
m_0,m_1,m_2\ge 0\\
m_0+m_1+m_2=k\\
(m_0,m_1,m_2)\in (a_0,a_1,a_2)+d\Z^3
}}
\C\{\widetilde x_0^{m_0}\widetilde x_1^{m_1}\widetilde x_2^{m_2}\}.
\end{equation*}
In particular,
\begin{align}
H^0(\widetilde{\P}^2,\cO^G_{\widetilde{\P}^2})^G
&=\C\{1\},
\nonumber\\
H^0(\widetilde{\P}^2,\cO^G_{\widetilde{\P}^2}(D_i))^G
&=\C\{\widetilde x_i\},
\qquad i=0,1,2.
\label{eq:h0G}
\end{align}

\subsection{\texorpdfstring{$G$}{G}-equivariant cohomology}

Since the standard affine cover \(\widetilde{\mathcal U}\) is \(G\)-invariant, the groups \(\check C^\bullet(\widetilde{\mathcal U},\mathcal F)\) inherit a \(G\)-action for every \(G\)-equivariant coherent sheaf \(\mathcal F\), and the \v Cech differential is \(G\)-equivariant. Because \(G\) is finite and we work over \(\C\), taking \(G\)-invariants is exact. Hence
\[
H^\bullet(\widetilde{\P}^2,\mathcal F)^G
\cong
H^\bullet\bigl(\check C^\bullet(\widetilde{\mathcal U},\mathcal F)^G\bigr).
\]

Applying this to~\eqref{eq:h2}, we obtain
\begin{align}
H^2\bigl(\widetilde{\P}^2,\mathcal O^G_{\widetilde{\P}^2}(-D)\bigr)^G
&=
\frac{1}{\widetilde x_0\widetilde x_1\widetilde x_2}\,\C\{1\},
\nonumber\\
H^2\bigl(\widetilde{\P}^2,\mathcal O^G_{\widetilde{\P}^2}(-D_i-D)\bigr)^G
&=
\frac{1}{\widetilde x_0\widetilde x_1\widetilde x_2}\,\C\{\widetilde x_i^{-1}\}, \qquad i=0,1,2.
\label{eq:h2G}
\end{align}

\section{Hypercohomology} \label{app:ss}

In the main text, we will compute various $\Ext$ spaces via hypercohomology of \v Cech-$\mathcal Hom$ double complexes. In this appendix, we review the necessary background following~\cite[Chapter~II]{BT}. This is a particularly simple example of a spectral sequence.

Let $K^{\bullet,\bullet}$ be a double complex. Thus each row is a complex with differential $\delta$, each column is a complex with differential $d$, and the two differentials commute:
\[
d \delta = \delta d.
\]
The associated \emph{total complex} $\tot(K)$ is the single complex with entries
\[
\tot(K)^n = \bigoplus_{p+q=n} K^{p,q}, \qquad D = D'+D''
\]
where $D' := \delta$ and $D'' := (-1)^p d$. By definition, the \emph{hypercohomology} 
\[
\mathbb H^\bullet(K^{\bullet,\bullet}) := H^\bullet(\tot(K))
\]
is the cohomology of the total complex.

The following is a simpler version of \cite[Proposition~12.1]{BT}.

\begin{proposition}\label{prop:ttt}
Given a double complex $K^{\bullet,\bullet}$, if
\[
 H_\delta(K^{\bullet,\bullet})
\]
has nonzero entries only in the $p$-column, then 
\[
H_d^{p,q}( H_\delta(K^{\bullet,\bullet})) \cong \mathbb H^{p+q}(K^{\bullet,\bullet}),
\]
where $H_\delta$ (resp. $H_d$) denotes cohomology with respect to $\delta$ (resp. $d$).
\end{proposition}

Here $H_\delta(K^{\bullet,\bullet})$ denotes the table of cohomologies of the rows, each column of which is a $d$-complex, and $H_d^{p,q}( H_\delta(K^{\bullet,\bullet}))$ is the table obtained by taking cohomologies of each of the columns of $H_\delta(K^{\bullet,\bullet})$. We now describe this isomorphism explicitly since it will be important in our calculations.

For $\phi \in K^{p,q}$ with $\delta \phi = 0$, we write
\[
[\phi]_\delta \in H_\delta^{p,q}(K^{\bullet,\bullet})
\]
for its class in the $\delta$-cohomology of the $q$-th row. If $[\phi]_\delta$ is $d$-closed, we write
\[
[\phi]_{d,\delta} \in H_d^{p,q}\bigl(H_\delta(K^{\bullet,\bullet})\bigr)
\]
for its class in the iterated cohomology, and similarly, if it is $D$-closed, we write $[\phi]_D$ for its hypercohomology class.

We first describe the map $h: H_d^{p,q}( H_\delta(K^{\bullet,\bullet})) \rightarrow \mathbb H^{p+q}(K^{\bullet,\bullet})$. Suppose $[\phi]_{d,\delta} \in H_d^{p,q}H_\delta(K^{\bullet,\bullet})$. Then $[\phi]_{d,\delta}$ is represented by an element 
\[
[\phi]_{\delta} \in H_\delta^{p,q}(K^{\bullet,\bullet}), \qquad d([\phi]_{\delta})=0 \in H_\delta^{p,q+1}(K^{\bullet,\bullet}),
\]
which is in turn represented by an element $\phi \in K^{p,q}$ such that 
\[
\delta \phi = 0, \qquad D''(\phi) = - D'(\phi_1) \qquad (\phi_1 \in K^{p-1,q+1}).
\]
We represent this by the diagram
\[
\begin{tikzcd}
\phi_1 \arrow[r,"{D'}"] &  D'(\phi_1)+D''(\phi) =0 & \\
& \phi \arrow[u,"{D''}"] \arrow[r,"{D'}"] & 0
\end{tikzcd}.
\]
However $\phi+\phi_1$ is not a $D$-cocycle since $D''(\phi)$ may be nonzero. 

Since $H_\delta^{p-1,q+2}(K^{\bullet,\bullet}) = 0$, there is an element $\phi_2 \in K^{p-2,q+1}$ such that
\[
\begin{tikzcd}
\phi_2 \arrow[r,"{D'}"] &  D'(\phi_2)+D''(\phi_1) =0 && \\
&\phi_1 \arrow[r,"{D'}"] \arrow[u,"{D''}"] &  D'(\phi_1)+D''(\phi) =0 & \\
&& \phi \arrow[u,"{D''}"] \arrow[r,"{D'}"] & 0
\end{tikzcd},
\]
Continuing in this way, we obtain a $D$-cocycle $\phi + \phi_1 + \cdots + \phi_n$. The map $h$ sends $[\phi]_{d,\delta}$ to $[\phi+ \phi_1 + \cdots + \phi_n]_D \in  \mathbb H^{p+q}(K^{\bullet,\bullet})$.

The inverse map
\[
h^{-1}: \mathbb H^{p+q}(K^{\bullet,\bullet}) \rightarrow H_d^{p,q}\bigl(H_\delta(K^{\bullet,\bullet})\bigr)
\]
is defined as follows. Let $\omega$ be a $D$-cocycle. By modifying $\omega$ within its $D$-cohomology class, we may arrange that all of its components to the right of the nonzero column vanish. Suppose that $\omega=a+b+c+\cdots$, as in the diagram
\[
\begin{tikzcd}
  c && \\
b_1 \arrow[r,"{D'}"] \arrow[u,"{D''}"] & b & \\
& a_1 \arrow[u,"{D''}"] \arrow[r,"{D'}"] & a
\end{tikzcd},
\]
where $c$ lies in the nonzero column.

Since the columns to the right of the nonzero column are trivial in iterated cohomology, we may first choose $a_1$ such that
\[
a=D'(a_1).
\]
Subtracting the total coboundary $D(a_1)$ from $\omega$ removes the component $a$:
\[
\omega-D(a_1)= b-D''(a_1)+c+\cdots .
\]
Next, since \(b-D''(a_1)\) is \(D'\)-closed, we may choose \(b_1\) such that
\[
D'(b_1)=b-D''(a_1).
\]
Subtracting \(D(b_1)\) then removes this term as well, giving
\[
\omega-D(a_1)-D(b_1)=c+\cdots ,
\]
which now begins in the nonzero column. We then define
\[
h^{-1}([\omega]_D):=[c]_{d,\delta}.
\]

In general, one proceeds in the same way: by successively subtracting suitable \(D\)-coboundaries, one may replace \(\omega\) by a cohomologous cocycle whose first nonzero component lies in the nonzero column, and then take the resulting class in \(H_d^{p+q}(H_\delta(K^{\bullet,\bullet}))\).

\section{Proof of Proposition~\ref{prop:scalar}}\label{app:prop:scalar}

We can find $c$ by looking at a single element. Let us take $\delta_\b \in \C^\B \oplus \C^\W$. By Lemma~\ref{lem:Ext2}, it is represented by the \v Cech cocycle
\[
\xi = \frac{1}{\widetilde x_0 \widetilde x_1 \widetilde x_2} \in \check C^2(\mathcal H om(\cale_\b,\cale_\b(-\widetilde D)))^G \subset \check C^2(\mathcal H om(\cale,\cale(-\widetilde D))\oplus\mathcal H om(\calf,\calf(-\widetilde D)))^G
\]
which by the explicit isomorphism in Proposition~\ref{prop:ttt} may be represented by a $D$-$2$-cocycle of the form
\[
\begin{tikzcd}
0 & & \\
\xi' \arrow[u,"-d"] \arrow[r,"\check d"] & \check d \xi' + d \xi =0 & \\
& \xi \arrow[u,"d"] \arrow[r,"\check d"] & 0
\end{tikzcd},
\]
where $\xi' = (\xi'_{01},\xi'_{02},\xi'_{12}) \in \check C^1(\mathcal Hom(\cale,\calf(-\widetilde D)))^G$ may be chosen to only be nonzero for $\b \w \in \E$ where it is given by
\[
  \begin{tikzpicture}[scale=1.5] 
		\coordinate[bvert,label=below:$ \b$] (b) at (2,1.732-0.5773);
			\coordinate[wvert,label=below:$\w_0$] (w1) at (1,0.5773);
			\coordinate[wvert,label=below:$\w_1$] (w2) at (3,0.5773);
			\coordinate[wvert,label=above:$\w_2$] (w3) at (2,1.732+0.5773);

			\draw[-]
				  (b) edge  (w2) edge node[left] {$ (-\frac{c}{\widetilde x_0 \widetilde x_1},0,0)$} (w3) edge  (w1)
				;
				\node[] (no) at (3.1,1.) {$(0,\frac{b}{\widetilde x_0 \widetilde x_2},0)$}; 
				\node[] (no) at (1.3-.75,1.) {$(0,0,-\frac{a}{\widetilde x_1 \widetilde x_2}) $}; 
	\end{tikzpicture} 
\]
where $a,b,c$ denote the weights of the three edges incident to $\b$.

The map $\widetilde\theta$ on $\Ext$ is induced by the map 
\[
K(-{\widetilde D})^\bullet \xrightarrow[]{\widetilde \theta}K_{\widetilde D}^\bullet
\]
multiplying termwise by $\widetilde\theta$. Thus the image $D$-$2$-cocycle for $\check C^{\,p}\!\bigl(\mathcal H om^{q}(K^\bullet,K_{\widetilde D}^\bullet)\bigr)^{G}$ is represented by
the $D$-$2$-cocycle of the form
\[
\begin{tikzcd}
0& &&\\
0\arrow[r,"\check d"] \arrow[u,"d"]&0 & & \\
&\zeta' \arrow[u,"-d"] \arrow[r,"\check d"] & \check d \zeta' + d \zeta =0 & \\
&& \zeta \arrow[u,"d"] \arrow[r,"\check d"] & 0
\end{tikzcd},
\]
where 
\[
\zeta = 1 \in \check C^2(\mathcal H om(\cale_\b,\cale_\b))^G \subset \check C^2(\mathcal H om(\cale,\cale)\oplus\mathcal H om(\calf,\calf))^G
\]
and $\zeta' = (\zeta'_{01},\zeta'_{02},\zeta'_{12}) \in \check C^1(\mathcal Hom(\cale,\calf))^G$ is given by
\[
  \begin{tikzpicture}[scale=1.5] 
		\coordinate[bvert,label=below:$ \b$] (b) at (2,1.732-0.5773);
			\coordinate[wvert,label=below:$\w_0$] (w1) at (1,0.5773);
			\coordinate[wvert,label=below:$\w_1$] (w2) at (3,0.5773);
			\coordinate[wvert,label=above:$\w_2$] (w3) at (2,1.732+0.5773);

			\draw[-]
				  (b) edge  (w2) edge node[left] {$ (-{c}{\widetilde x_2},0,0)$} (w3) edge  (w1)
				;
				\node[] (no) at (3.1,1.) {$(0,{b}{\widetilde x_1},0)$}; 
				\node[] (no) at (1.3-.75,1.) {$(0,0,-{a\widetilde x_0}) $}; 
	\end{tikzpicture} .
\]

Finally, we want to find its class in $\Hom(\cale,\widetilde{\mathcal L}|_{\widetilde D})^G$ under the identification in Lemma~\ref{lem:Fbullet}. To do this, we must rewrite the class in the representative used by Proposition~\ref{prop:ttt}, namely a $D$-cocycle supported entirely in the \v Cech degree $0$ column. Thus we modify the above $D$-cocycle by a $D$-coboundary until only that column remains, i.e. we are looking for $\omega, \omega'$ such that in
\[
\begin{tikzcd}
d \omega'  & & & \\
\omega' \arrow[r,"\check d"] \arrow[u,"d"]& \check d \omega' -d \omega + \zeta'=0 &  & \\
& \omega \arrow[r,"\check d"] \arrow[u,"-d"] & \check d \omega + \zeta = 0  & 
\end{tikzcd},
\]
only $d\omega' \neq 0$. Recall that $\check d(\alpha_{01},\alpha_{02},\alpha_{12}) = \alpha_{12}-\alpha_{02}+\alpha_{01}$. Thus, we may take
\[
\omega = (-1,0,0) \in \check C^1(\mathcal Hom(\cale_\b,\cale_\b))^G \subset  \check C^1(\mathcal H om(\cale,\cale)\oplus \mathcal H om(\calf,\calf))^G
\]
so that $\check d \omega + \zeta = 0$. Then the component of $-d \omega$ in $\check C^1(\mathcal Hom(\cale,\calf))^G$ is 
\[
  \begin{tikzpicture}[scale=1.5] 
		\coordinate[bvert,label=below:$ \b$] (b) at (2,1.732-0.5773);
			\coordinate[wvert,label=below:$\w_0$] (w1) at (1,0.5773);
			\coordinate[wvert,label=below:$\w_1$] (w2) at (3,0.5773);
			\coordinate[wvert,label=above:$\w_2$] (w3) at (2,1.732+0.5773);

			\draw[-]
				  (b) edge  (w2) edge node[left] {$ ({c}{\widetilde x_2},0,0)$} (w3) edge  (w1)
				;
				\node[] (no) at (3.1,1.) {$({b}{\widetilde x_1},0,0)$}; 
				\node[] (no) at (1.3-.75,1.) {$({a\widetilde x_0},0,0) $}; 
	\end{tikzpicture} 
\]
while the component in $\check C^1(\calf,\widetilde{\mathcal L}|_{\widetilde D})$ is $0$. Finally, we need to choose 
\[
\omega' = (\lambda,\mu) \in \check C^0(\mathcal H om(\cale,\calf)\oplus \mathcal H om(\calf,\widetilde{\mathcal L}|_{\widetilde D}))^G
\]
so that $\check d \omega' -d \omega + \zeta'=0$, i.e. so that $d \omega'$ is 
\[
  \begin{tikzpicture}[scale=1.5] 
		\coordinate[bvert,label=below:$ \b$] (b) at (2,1.732-0.5773);
			\coordinate[wvert,label=below:$\w_0$] (w1) at (1,0.5773);
			\coordinate[wvert,label=below:$\w_1$] (w2) at (3,0.5773);
			\coordinate[wvert,label=above:$\w_2$] (w3) at (2,1.732+0.5773);

			\draw[-]
				  (b) edge  (w2) edge node[left] {$ (0,0,0)$} (w3) edge  (w1)
				;
				\node[] (no) at (3.3,1.) {${b}{\widetilde x_1}(-1,-1,0)$}; 
				\node[] (no) at (1.3-.75,1.) {${a\widetilde x_0}(-1,0,1) $}; 
	\end{tikzpicture} 
\]
Recall that $\check d(\alpha_0,\alpha_1,\alpha_2) = (\alpha_1-\alpha_0,\alpha_2-\alpha_0,\alpha_2-\alpha_1)$. Thus, we may take $\lambda$ equal to 
\[
  \begin{tikzpicture}[scale=1.5] 
		\coordinate[bvert,label=below:$ \b$] (b) at (2,1.732-0.5773);
			\coordinate[wvert,label=below:$\w_0$] (w1) at (1,0.5773);
			\coordinate[wvert,label=below:$\w_1$] (w2) at (3,0.5773);
			\coordinate[wvert,label=above:$\w_2$] (w3) at (2,1.732+0.5773);

			\draw[-]
				  (b) edge  (w2) edge node[left] {$ (0,0,0)$} (w3) edge  (w1)
				;
				\node[] (no) at (3.3,1.) {${b}{\widetilde x_1}(1,0,0)$}; 
				\node[] (no) at (1.3-.75,1.) {$-{a\widetilde x_0}(0,1,0) $}; 
	\end{tikzpicture} 
\]
and $\mu=0$. Then 
\[
d \omega' = r \circ \lambda \in \check C^0(\mathcal H om(\cale,\widetilde{\mathcal L}|_{\widetilde D}))^G
\]
is only nonzero over $\widetilde \alpha$ for the zig-zag path $\alpha \in \zz_2$ containing $\w_0,\b,\w_1$; in that case, it is equal to
\[
r \circ (-a \widetilde x_0,b \widetilde x_1, 0).
\]

Since \(\widetilde{\mathcal L}|_{\widetilde D}\cong \coker(\widetilde K|_{\widetilde D})\), two local lifts define the same section of \(\widetilde{\mathcal L}|_{\widetilde D}\) if their difference lies in \(\operatorname{im}(\widetilde K|_{\widetilde D})\). In our conventions, the column of \(\widetilde K\) corresponding to \(\b\) is
\[
(a\widetilde x_0,\; b\widetilde x_1,\; c\widetilde x_2).
\]
Since \(\widetilde x_2=0\) on \(\widetilde \alpha\), the local representatives
\(-a\widetilde x_0\) and \(b\widetilde x_1\) differ by this column, and hence determine the same class in \(\widetilde{\mathcal L}|_{\widetilde D}\). Thus \(d\omega'\) is a global section.

Under the identification in Lemma~\ref{lem:ident}, this global section has only one nonzero component, namely the component indexed by the zig-zag path \(\alpha\), and that component is the functional
\[
\delta_{\b}\in \Hom(\C^{\B\cap\alpha},\C),
\]
i.e. it takes the value \(1\) on \(\b\) and \(0\) on all other black vertices of \(\alpha\). Indeed, exactly as in the proof of Lemma~\ref{lem:map1}, after the chosen gauge normalization the quotient map
\[
\bigoplus_{\w\in\W\cap\alpha}\restr{\calf_\w}{\widetilde\alpha}
\rightarrow
\restr{\widetilde{\mathcal L}}{\widetilde\alpha}
\cong
\coker(\partial_\alpha)\cong \C
\]
is represented by the row vector \((1,\dots,1)\), and the class of
\[
r\circ(-a\widetilde x_0,b\widetilde x_1,0)
\]
maps to \(1\). Since the section is supported on the summand \(\cale_{\b}\), the resulting functional is precisely \(\delta_{\b}\). Therefore
\[
\phi\bigl([d\omega']\bigr)=1.
\]
On the other hand,
\[
(\widetilde\theta\circ\tr)(\delta_{\b})
=
\sum_{\w\in\W}0-\sum_{\b'\in\B}\delta_{\b}(\b')
=
-1.
\]
Hence commutativity of the diagram forces
\[
c=-1.
\]

\end{appendices}
\bibliography{references}
\Addresses
\end{document}